\documentclass[cmp]{svjourmod}

\usepackage{dcolumn}
\usepackage{bm}
\usepackage{graphicx}
\usepackage{placeins}
\usepackage{amssymb}
\usepackage{amsmath}
\usepackage{amsfonts}
\usepackage{mathrsfs}
\usepackage[colorlinks=true, linkcolor=blue, citecolor=blue, urlcolor=blue, linktoc=page]{hyperref}

\newcommand{\p}{\partial}

\newcommand{\dd}{{\rm d}}

\begin{document}

\title{Synchronization of clocks in relativity theory and computer science}

\author{E. Minguzzi}
\institute{Dipartimento di Matematica, Universit\`a degli Studi di Pisa,
Largo B. Pontecorvo 5, I-56127 Pisa, Italy \\
\email{ettore.minguzzi@unipi.it} }
\authorrunning{E. Minguzzi}

\date{}
\maketitle

\begin{abstract}
We study the synchronization of clocks through signal exchanges. We improve the method of Poincar\'e and Einstein by removing the need for a round-trip condition. For every probability measure of full support on the space of clocks there is a consistent synchronization, which corrects Einstein's convention with an average of the Sagnac effect. It respects causality and, under integrability assumptions, minimizes a global anisotropy functional. Its non-local ingredients are studied with Alexander-Spanier and singular cohomology. Locally, for networks the optimal synchronization is the unique Nash equilibrium of a game among the clocks, while for the clocks at rest in a stationary spacetime it is the harmonic synchronization, solution of a wave equation. As the latter can violate causality, we introduce the minimax synchronizations, which minimize the maximal anisotropy of the one-way speed of light. Using the synchronization theorem of a companion paper, based on weak KAM theory, we show that every compact space of clocks admits eikonal minimax synchronizations. They have constant anisotropy almost everywhere and, if the spacetime is distinguishing, are time functions with Cauchy slices orthogonal, in a weak sense allowing for caustics, to the observers freely falling in the metric $g/(-k\cdot k)$, with $k$ the Killing field. Under a curvature condition they then coincide with the harmonic synchronization and have totally geodesic slices. For networks they are computed by a shortest path algorithm, illustrated on a global fiber network. The limits of the $L^p$ minimizers of the anisotropy are also minimax. Finally, for rigidly rotating frames in Kerr and, more generally, in circular stationary axisymmetric spacetimes, the Boyer-Lindquist time is harmonic and, over every axisymmetric region, minimax. The methods apply to the synchronization of computational networks.
\end{abstract}

\clearpage
\setcounter{tocdepth}{2}
\tableofcontents

\section{Introduction} \label{uno}

The Poincar\'e-Einstein's synchronization convention (Einstein's
synchronization for short) establishes that the time of emission and
reception of a signal are separated by a time interval which equals
half the two-way transmission time
\cite{einstein05,poincare00,poincare04a,poincare04b,jammer06}. This
convention is universally used in computer networks and in other
distributed systems in which the definition of a common time becomes
crucial for the logical ordering of the computational processes
\cite{tanenbaum07}. It is at the heart of the Precision Time
Protocol (PTP), which reaches submicrosecond accuracy and which is
regulated through the international standard IEEE 1588-2002 on time
dissemination \cite[Sect. 3.3.4]{eidson06} \cite{ieee1588}; its revisions of 2008 and 2019 aim at sub-nanosecond accuracy in suitably designed networks. The
Network Time Protocol (NTP) \cite{mills06}, now in its version 4 \cite{rfc5905}, which serves millions of
clients worldwide is also based on the Einstein's convention
\cite[Sect. 1.2]{mills06}.

In this work I consider a modification of this synchronization
method that removes the need of a round-trip condition, thus extending its range of applicability. In a model of a global fiber network it also reduces the synchronization errors due to the Sagnac effect with respect to Einstein's synchronization along a spanning tree (Subsection \ref{n:numex}). It
could be implemented by modifying existing synchronization protocols, although its performance in operational deployments remains to be assessed.
A remarkable feature of the method of Section \ref{mqz} is that it requires no knowledge of the geometry of the space of clocks, of the propagation medium or of the gravitational field: it suffices that every clock ``pings'' the other clocks, namely that it exchanges time-stamped signals with them, so as to collect the emission and reception times which give the Fermat's map, from which the resetting of each clock is computed with Eq.\ (\ref{vqh}). For instance, on a rotating planet the method returns the ephemeris coordinate time without any information on the shape of the geoid (Subsection \ref{djc}).
As a byproduct it follows, somewhat surprisingly, that some
non-hypersurface orthogonal reference frames in general relativity
do admit a privileged time foliation, a fact that may have
consequences for the application of canonical quantization schemes.

\paragraph{History of this work and relation with the companion paper} This research began in 2010--2011. In 2010 I posted the preprint \cite{minguzzi10s}, which introduced the round-trip function $r$ and the Sagnac function $w$ of a triangle of clocks, showed that $w$ satisfies the 2-cocycle condition and that the radar distance is a distance, and proposed a modification of the Poincar\'e-Einstein convention which provides a consistent synchronization even when the round-trip condition fails, the correction being obtained by averaging $w$ with respect to a measure on the space of clocks, see Eq.\ (\ref{nka}). The preprint discussed applications to computer networks and to the clocks at rest in stationary spacetimes, and observed that on the spheres of constant radius of Schwarzschild spacetime the method returns the usual coordinate time, while the identification of the resulting time in rotating frames and in Kerr spacetime was left for further work. The method was then presented, together with a historical account of the Poincar\'e-Einstein synchronization, in the conference paper \cite{minguzzi10c}.

In 2011 I wrote a much longer manuscript, of which the present work is a revision, but I did not post or submit it because I felt that the project was incomplete. In that manuscript the method was reformulated through the Fermat's map, characterized through the minimization of an anisotropy functional, and studied in relation with causality and with the topology of the space of clocks; it was shown to recover the time of the non-rotating frame for rigidly rotating frames in static spacetimes (Theorem \ref{uiy}), and it was given a local version for networks and for the clocks at rest in a stationary spacetime. The latter (Section \ref{hod}) selects the {\em harmonic synchronization}, namely the solution of a wave equation, but I could not prove that it respects causality whenever this is possible, nor could I prove the existence of an alternative synchronization with this property, selected by a natural optimality criterion, a result that I would have called the {\em synchronization theorem}.

The synchronization theorem has now been proved in the companion paper \cite{minguzzi12bolza} by regarding the stationary spacetime as the Maupertuis reduction of a spacetime admitting a parallel null vector, and by identifying the simultaneity slices with the projections of the horizons of the latter spacetime, whose existence is established with the methods of weak KAM theory and causality theory. This development allows me to finally present this more physical and operational work. I believe that it further motivates the adoption of the proposed synchronization methods: they are constructive, they do not privilege any clock, and in all the examples that we consider they select a physically distinguished time.

The present work will replace the preprint \cite{minguzzi10s} as the second version of the arXiv preprint 1009.3005. The material of that preprint is essentially contained, in a different formalism, in Subsections \ref{due3} and \ref{mhj} and in the first part of Section \ref{mqz}. With respect to the 2011 manuscript, Subsections \ref{n:graphmm}-\ref{n:numex}, \ref{n:counter}-\ref{n:glob} and Section \ref{n:kerr} are new, and Theorem \ref{n:circ}, which extends Theorem \ref{uiy} to the rigidly rotating frames of circular stationary axisymmetric spacetimes, has been added. The present work and the companion paper \cite{minguzzi12bolza} can be read independently, as the arguments of the latter which are needed here are repeated or summarized. Since the 2011 manuscript was not made public, some of its results have been obtained in the meantime, independently, by other authors; their works are cited at the appropriate places, see e.g.\ \cite{costaesilva17,herrera21,costa21}.

In order to fix the notation let $(S,\mu)$ be a measure space in
which $\mu$ is positive and normalized, $\mu(S)=1$. The points of
$S$ will be called {\em clocks} and we shall assume that signals
propagate on $S$ and are used to synchronize clocks. The actual
mathematical model of propagation will be irrelevant for our
purposes.

The space $S$ can be the quotient space manifold of a congruence of
timelike curves in general relativity \cite{hawking73,minguzzi03}; a
discrete space giving the nodes of a network of stable servers
connected all over the Earth; or a discrete space describing the
processes inside one single computer \cite{lamport78}. The
propagation may be due to light propagating in vacuum, sound in the
air, or electric signals on copper wires.  One should not
necessarily think of the time measured at $s\in S$, as a physical
time. It could be a logical time related to process activity similar
to that described in Lamport's seminal work \cite{lamport78}. What
follows will be largely independent of the particular physical
application.

We shall assume throughout the work that a signal starting from any
point $s_1$ will eventually reach any other point $s_2$ (this may be
called the {\em no horizon} assumption) and that, according to the
order given by the clock at $s_2$, the times at which the
signals emitted at $s_1$ arrive at $s_2$ are bounded from below. In other words these two conditions are
equivalent to the imposition that the time it takes the signal to go
from $s_1$ to $s_2$ (in one and hence any synchronization), defined as the infimum of the arrival times minus the emission time, is
neither $+\infty$ nor $-\infty$, and thus it is finite. We do not require the infimum to be attained, namely the existence of a first arrival, which in the Lorentzian case holds, for instance, if the causal relation is closed (cf.\ Theorem \ref{gdw}(1)).

We shall say that property ${\bf z=0}$ holds if the time of
propagation between any two space points $s_1$ and $s_2$ is
independent of the event of emission at $s_1$. Stated in another way
${\bf z=0}$ holds  in the system if for any pair $s_1,s_2\in S$ any
two signals sent from $s_1$ and separated by an interval $\Delta t$
according to $s_1$,  are received at $s_2$ separated by the same
time interval as measured by $s_2$. This condition intuitively
states that the clocks go at the same  rate and, as one can already
expect, it will be essential to attain synchronization. Indeed, the
process of syntonization of clocks (i.e. assigning to them a
convenient rate) precedes logically that of synchronization. If two
clocks are not syntonized it will be useless to synchronize them as
they will immediately run out of sync. In this work we shall always
assume that ${\bf z=0}$ holds.

At first sight the property ${\bf z=0}$ might seem restrictive,
however, it can be strongly relaxed because for our purposes it must
only be approximately satisfied in the interval of time needed to
synchronize the clocks. This time interval usually depends on the
extension of the system of clocks. On the surface of the Earth the
time needed for synchronizing the clocks is of the order of the time
taken by the signal to circumnavigate the Earth. During this time
(which is minuscule if the signal is light propagating on optical
cables) the property ${\bf z=0}$ is indeed satisfied at sea level
because, while it is true that the angular velocity of the Earth
diminishes due to tide friction, this effect in such a small time
interval can be totally neglected.

Of course, when coming to examples in which the propagation is
determined through a Lorentzian metric the condition ${\bf z=0}$
provides a constraint if imposed all over the manifold. A relevant
class of  spacetimes  which satisfy this constraint in that  of
(conformal) stationary spacetimes.

Denoting with $k$ the conformal timelike Killing vector that
generates the frame we define the space  $S$ to be the quotient of
the spacetime under the flow of the frame. By definition of
conformal Killing vector $L_{k}g= (\p_k\ln(- k\cdot k))\, g$, and
passing to the conformal metric $\hat{g}=g/(-k\cdot k)$ since $L_k
k=0$ it is easy to check $L_k \hat{g}=0$, thus $k$ is a normalized
(as $\hat{g}(k, k)=-1$) Killing vector for the spacetime $(M,
\hat{g})$. The property ${\bf z=0}$ is satisfied on $(M, \hat{g})$
and $S$ can be regarded as a Riemannian manifold with the optical or
Fermat  \cite{perlick90b} space metric $h_{\alpha \beta}=[g_{\alpha
\beta}+k_{\alpha} k_{\beta}/(-k\cdot k)]/(-k\cdot k)$.  The rotating
planet will be studied using this type of propagation model.

\subsection{Rudiments of causality theory} \label{uno2}

In what follows we shall use  some basic facts of causality theory
for (Lorentzian) spacetimes $(M,g)$. We recall them here and fix the
notation following \cite{hawking73,minguzzi19}.

We denote with $(M,g)$ a $C^r$ spacetime (connected, Hausdorff,
time-oriented Lorentzian and hence paracompact manifold), $r \in {3,
\dots, \infty}$ of arbitrary dimension $n \geq 2$ and signature
$(-,+, \dots, +)$. On $M \times M$ the usual product topology is
defined. The subset symbol $\subset$ is reflexive, thus $X \subset
X$. We use the dot to denote the boundary of a set, e.g. $\dot{X}$.

Two points are chronologically related, $(x,y)\in I$ or $x\ll
y$, if there is a future directed timelike curve connecting $x$ to
$y$. We write $x<y$ if there is a future directed causal curve
connecting $x$ to $y$. Points $x,y$, are causally related, in
notation $(x,y)\in J$ or $x\le y$, if $x=y$ or $x<y$. The
relation $I$ is open and the identity
$\overline{I}=\overline{J}$ holds. The relations $I$ and $J$ are subsets of $M\times M$, while $I^{\pm}(x)$ and $J^{\pm}(x)$ denote the chronological and causal future and past of $x$, e.g.\ $J^{+}(x)=\{y: (x,y)\in J\}$.

The chronology violating set $\mathcal{C}\subset M$ is the set of
all $x\in M$ through which there passes a closed timelike curve. The
chronology violating set is the union of the open chronology
violating classes $[x]=I^{+}(x)\cap I^{-}(x)$ for $x\in
\mathcal{C}$. A spacetime is called {\em non-totally vicious} if
$\mathcal{C}\ne M$ and {\em chronological} if
$\mathcal{C}=\emptyset$. It is {\em causal} if $J$ is
antisymmetric, that is, if and only if there is no closed causal
curve on spacetime. The spacetime is future distinguishing if for
every $x\in M$ there is an arbitrarily small neighborhood $V\ni x$
such that there is no future inextendible causal curve which starts
from $x$, escapes $V$ and reenters it. Past distinction is defined
dually and the spacetime is {\em  distinguishing} if it is both past
and future distinguishing. The spacetime is {\em stably causal} if
causality is stable under small perturbations of the metric, namely
if the light cones can be widened without introducing closed causal
curves.

A time function is a continuous function $t:M\to \mathbb{R}$ which
increases over every causal curve, that is, $x<y\Rightarrow
t(x)<t(y)$. A semi-time function \cite{seifert77} is a continuous
function which satisfies the weaker property $x\ll y\Rightarrow
t(x)<t(y)$ (which by continuity implies $x \le y \Rightarrow t(x)\le
t(y)$). A temporal function is a time function which is smooth with
a timelike gradient (necessarily past directed because of our choice
of signature).

Stable causality is equivalent to the existence of a time function
which is equivalent  to the existence of a temporal function
\cite{bernal04}. Let $K$ be the smallest closed and transitive
relation containing the causal relation $J$ \cite{sorkin96}.
Stable causality is also equivalent to the antisymmetry of $K$
\cite{minguzzi08b}. We shall also write $x \le_K y$ in place of
$(x,y)\in K$, however, if we are not dealing with a Lorentzian
spacetime the symbol $\le_K$ may denote just a preorder (transitive
and reflexive relation).

A spacetime is reflective if $y\in
\overline{I^{+}(x)}\Leftrightarrow x \in \overline{I^{-}(y)}$. A
distinguishing and reflective spacetime is {\em causally
continuous}. A spacetime is {\em causally simple} if it is causal
and $\overline{J}=J$. A spacetime is {\em globally
hyperbolic} if it is causal and the causal diamonds $J^{+}(x)\cap
J^{-}(y)$ are compact.

All the previous properties highlighted in italics stay in the
causal ladder of spacetimes, namely they are in a sequence of
increasing strength, thus global hyperbolicity implies causal
simplicity which implies causal continuity which implies stable
causality and so on.

There is also a transverse conformal ladder \cite{minguzzi08b} which
states that the compactness of the causal diamonds implies
$\overline{J}=J$ which implies reflectivity which implies
the transitivity of $\overline{J}$ (namely
$K=\overline{J}$).

\subsection{Plan of the work} \label{uno3}
The work is organized as follows.

In Section \ref{fre} the Fermat's map $a:S\times S\to \mathbb{R}$ is
introduced. Everything we need to know on the  propagation of the
signal is shown to be contained in this function so that its main
properties, like the triangle inequality, and its transformation
properties under resynchronizations, are studied. In Subsection
\ref{hin} the case of a stationary (Lorentzian) spacetime is
considered. In particular a correspondence is given between
causality properties for the spacetime and properties of the
Fermat's map. In this section it is also established, in the general
(non-Lorentzian) case, which properties must satisfy the Fermat's
map to be sure that the  attained synchronization respects
causality. The reader not acquainted with general relativity might
want to skip this subsection but for  Theorem \ref{ghq}.

In Subsection \ref{uhm} the synchronization independent {\em two-way
time} and the {\em Sagnac's function} are introduced and it is
observed that the latter satisfies a cocycle condition. The
Einstein's synchronization is expressed in this formalism.

In Subsection \ref{mhj} the radar distance is introduced and
studied.
The concept of synchronization class is also introduced, in
particular it is proved that a synchronization leads to a time
coordinate which respects causality  if and only if the Fermat's map
for that synchronization gives a quasi-metric on the quotient space.
Finally, some bounds on the Sagnac function $w$ are given.

The problem of synchronization might make sense in different
mathematical settings depending if the space $S$ has a  notion of
{\em neighborhood}. We analyse three non-excluding possibilities:
that in which $S$ is just a measure space; that in which $S$ is made
of vertices of a graph so that the neighborhood of a point is given
by the nodes connected to it; and that of a Riemannian manifold so
that one has arbitrarily small neighborhoods (of the diagonal on
$S\times S$, thus actually a {\em uniformity} \cite{willard70}) a
fact which allows us to pass all the formalism to the (co)tangent
bundle.
We shall see that in the most important cases they select the same
time. Figure \ref{n:scheme} summarizes these settings, the synchronizations that we introduce in each of them, and their relations.

\begin{figure}[t]
\centering
\includegraphics{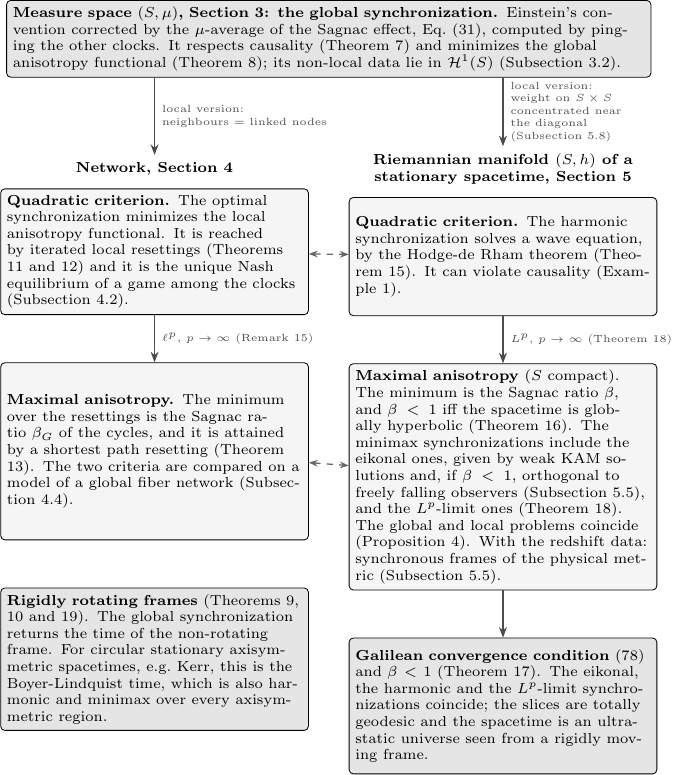}
\caption{The settings for the space of clocks and the synchronizations studied in this work. The arrows from the top box lead to the local versions of the global method, the vertical arrows to the limits $p\to\infty$ of the $L^p$ criteria, and the dashed arrows connect discrete and continuous analogues. The bottom boxes describe the cases in which different synchronizations coincide.} \label{n:scheme}
\end{figure}

The most general approach is of course the first. Its study starts
in Section \ref{mqz} where we introduce a synchronization method
which modifies Einstein's by adding a correction which is given by
the average of the Sagnac effect. It is also proved that the time
constructed in this way is a time function (or more precisely, an
utility) as it respects causality.

In Subsection \ref{sde} it is shown that this new synchronization
method can be obtained through the minimization of a (global)
anisotropy functional which measures how far from isotropy is the
propagation of the signal for a given synchronization of the clocks.
This functional approach will be essential for the generalizations
to graphs and manifolds.

In Subsection \ref{evt} the fact that any solution to the
synchronization problem must present some non-local feature is
stressed. In particular the local data determine the relevant cochain only up to an element of the Alexander-Spanier cohomology group $\mathcal{H}^{1}(S)$, which parametrizes the missing non-local information.

Subsection \ref{vtx} is concerned with applications in which the
signal propagation is described by a Lorentzian metric. There it is
proved that if the synchronization method is applied in a frame
which is rotating with respect to a frame adapted to a static
spacetime then the time one obtains is the one of the non-rotating
frame. The same holds for every rigidly rotating frame of a circular stationary axisymmetric spacetime, such as Kerr spacetime (Theorem \ref{n:circ}).

In Subsections \ref{dja} and \ref{djb} the theorem is applied to the
rotating platform and to the pseudocylindric coordinate system in
Minkowski spacetime, but, more importantly, a study of the rotating
axisymmetric (fluid) planet is given in Subsection \ref{djc} where
it is proved that a computational network synchronized according to
our proposal would produce a time which is the same as the ephemeris
coordinate time, all that without using any information on the
actual shape of the (axisymmetric) Earth geoid. This result
clarifies that the synchronization method  proposed here can
challenge not only theoretically but also practically the common
view synchronization method based on the GPS.

In Section \ref{vtv} we consider the synchronization of a graph of
clocks connected through cables over which the signal propagates.
In this case the signal  which leaves a point $s_1$ to reach a
distant point $s_2$ must first pass through the neighboring  points
of $s_1$. For this reason the Fermat's map can be deduced from that
restricted to the neighboring points. In analogy with the previous
case we introduce a functional which measures the anisotropy of
propagation between neighboring points and prove that there is, up
to a global constant resetting, one and only one synchronization of
the clocks which minimizes the functional.

The algorithm for resetting the clocks in the optimal way is given
in Subsection \ref{alg} through the introduction of the  {\em
conservation of  total time} under resynchronizations.

In Subsection \ref{opt} it is proved that this optimal
synchronization is the unique Nash equilibrium for a game of observers
located at the different points of $S$ which try, by moving their
respective clocks, to minimize the anisotropy of the signal at their
location. In Subsection \ref{n:graphmm} we introduce the minimax synchronizations of a network, which minimize the maximal anisotropy of the links. The minimal value is the maximal Sagnac ratio of the cycles of the network, an optimal resetting of the clocks is obtained with a shortest path algorithm, and it respects causality whenever no closed path can be covered by the signal in a vanishing time. In Subsection \ref{n:numex} we compare the quadratic and the minimax resettings on a model of a global fiber network.

In Section \ref{hod} we come to the synchronization problem for
Riemannian manifolds. In this case the anisotropy functional depends
on  a 1-form field which gives the anisotropy at the point. The fact
that the minimization problem does admit a solution is shown to be
exactly the content of the Hodge-de Rham theorem. In particular, the
non-local aspects of the synchronization problem appear clearly in
the harmonic component of the 1-form field. In Subsection \ref{n:counter} we show that the harmonic synchronization can violate causality, thus answering in the negative the question that had kept the 2011 manuscript of this work unpublished. In Subsection \ref{n:minimax} we introduce the minimax synchronizations, which minimize the maximal anisotropy, and we state the synchronization theorem of the companion paper \cite{minguzzi12bolza}: for a compact space of clocks minimax synchronizations exist, all of them respect causality whenever the spacetime is distinguishing, namely whenever $\beta<1$, and among them the eikonal synchronizations have the same anisotropy at every point of differentiability. In Subsection \ref{n:ff} we show that, if $\beta<1$, the slices of an eikonal synchronization are orthogonal, in a weak sense which allows for caustics, to a congruence of observers freely falling with respect to the metric $\hat g=g/(-k\cdot k)$ normalized by the Killing vector field $k$, so that it defines a stationary synchronous reference system in the sense of Landau and Lifshitz; we also explain how the gravitational redshift allows the clocks to measure the conformal factor and to be syntonized and synchronized coherently, with simultaneity slices orthogonal to observers freely falling in the physical metric. In Subsection \ref{n:curv} we show that, if $\beta<1$, under a curvature condition the eikonal and the harmonic synchronizations coincide and the normalized spacetime is ultrastatic, and in Subsection \ref{n:max} we interpret these results in terms of the mean curvature of the simultaneity slices and compare them with the theory of maximal hypersurfaces. In Subsection \ref{n:glob} we show that for the maximal anisotropy the global and the local problems coincide, and that the limits of the minimizers of the $L^p$ norms of the anisotropy for $p\to\infty$, the harmonic synchronization being the case $p=2$, are minimax synchronizations, which need not be eikonal synchronizations. Finally, Section \ref{n:kerr} is devoted to the rigidly rotating observers of Kerr spacetime, for which the Boyer-Lindquist time is returned by the global method, is a harmonic synchronization, and is selected among the minimax synchronizations by the minimization of the maximal anisotropy over every axisymmetric region.

\section{The Fermat's map} \label{fre} \label{due}

Mathematically, the assumption ${\bf z=0}$ is quite important as it
allows us to introduce the {\em Fermat's map}, $a:S\times S \to
\mathbb{R}$,  as follows. If a signal is sent from $s_1$ at
($s_1$-)time $t_1$ and received at $s_2$ at ($s_2$-)time $t_2$ then
we denote $a(s_1,s_2)=t_2-t_1$. Without the property ${\bf z=0}$
this map would depend also on the instant of emission at $s_1$. It
must be noted that $a(s_1,s_2)$ depends on the synchronization of
the clocks $s_1$ and $s_2$, indeed a resetting of the clocks
produces a change of the form
\begin{equation} \label{res}
a(s_1,s_2)\to a'(s_1,s_2)=a(s_1,s_2)+c(s_2)-c(s_1),
\end{equation}
where $c(s)$ gives how much we have moved forward clock $s$. Through
the Fermat's map all our manipulations will involve functions on $S$
and the concepts of {\em spacetime} or  {\em time} will be mentioned
in order to recall the physical meaning of those manipulations.
Since some of our result will have mathematical interest in their
own right, it is convenient to stress that
\begin{definition}
In this work by {\em resetting}, {\em synchronization} or {\em
resynchronization}, we shall mean an operation of the form
(\ref{res}) on the Fermat's map.
\end{definition}
Of course, physically speaking, given a (portion of) spacetime $M$
which can be written as a product $\mathbb{R}\times S$, with $t: M
\to \mathbb{R}$ projection on the first factor, the function $c$
provides the difference between two possible time coordinates
$t,t':M\to \mathbb{R}$ relative to two different splittings, i.e.
$t'-t=c(s)$. We shall see in a moment that, not the causal but the
$K$-causal relation $\le_K$ on $M$, can be written using the
trivialization $\mathbb{R}\times S$,
\begin{equation} \label{gae}
(t_1,s_1) \le_K (t_2,s_2) \Leftrightarrow t_2-t_1\ge a(s_1,s_2).
\end{equation}
Conversely, if one starts from a pair $(S,a)$  then one  first
defines $M$ as the equivalence class of products $\mathbb{R}\times
S$ under resynchronizations and then  this expression allows to
define on $M$ a preorder independent of the representative.

In the physical models that we are considering the following two
conditions are also satisfied and are therefore added to our
mathematical assumptions 
\begin{itemize}
\item[(i)] $a(s_1,s_3)\le a(s_1,s_2)+a(s_2,s_3)$,  (Fermat)
\item[(ii)] $a(s,s)= 0$. (Chronology)
\end{itemize}
The former condition (i) is  satisfied because a signal sent from
$s_1$ will generically be received several times at $s_3$. By
definition in $a(s_1,s_3)$ one considers the infimum of the times of arrival of the
signals (according to the order at $s_3$) that reach $s_3$
after emission at $s_1$, which is the time of arrival of the first signal when such a signal exists. This definition of time of arrival makes
(i) automatically satisfied. As for (ii), it is clear that a signal
emitted at  $s$ can be immediately received at $s$  thus $a(s,s)\le
0$. The condition (ii) is then imposing that the signal cannot
arrive at $s$ before it is emitted. Both conditions are independent
of the synchronization of the clocks namely, they are invariant
under replacements of the form (\ref{res}).

At this stage we have not yet imposed  any continuity property on
the Fermat map indeed so far we have not assigned to $S$ a topology.
For future reference, I write here a third generic assumption on the
Fermat's map (it will be cited anytime it is used).
\begin{itemize}
\item[(iii)] $a:S\times S\to \mathbb{R}$ is continuous.
\end{itemize}
It is understood that if $S$ is a finite set (e.g. the set of
vertices of a graph) the topology is the discrete one.

\subsection{The Fermat's map of a conformal stationary
(Lorentzian) spacetime} \label{hin} \label{due2}

In this section we study spacetimes $(M,g)$ which are non-totally
vicious and admit a complete  timelike conformal Killing vector $k$.
As it has been mentioned in the introduction, $k$ is a normalized
timelike Killing vector for the metric $\hat{g}=g/(-k\cdot k)$ (and
complete because this notion is independent of the metric).

A theorem by Clarke and Joshi states that if a spacetime admits a
complete timelike Killing vector then it is reflective \cite[Prop.
3.1]{clarke88}. Furthermore they proved that a non-totally vicious
spacetime which is reflective is chronological  \cite[Prop.
2.5]{clarke88}. We conclude that $(M,\hat{g})$ and hence $(M,g)$ are
chronological.

No integral line of the congruence of the Killing field $k$ can pass
arbitrarily close to itself, indeed if the integral line $\gamma$
passing through $x\in M$ intersects indefinitely in the future any
neighborhood of $x$ then taking $y\in \gamma$, $x\ll y$, since
$I$ is open we would obtain a closed timelike curve which is
impossible. The fact that no integral line passes arbitrarily close
to itself  guarantees the existence of a quotient manifold $S$ and a
smooth projection $\pi:M\to S$ (this is a standard construction
based on the local existence of hypersurfaces transverse to the
Killing congruence \cite[Theor. 9.16]{lee03}; they provide charts
where the transition functions are given by the diffeomorphisms
generated by the Killing field \cite{harris92}). The principal
bundle structure and the convexity of the fibers implies the
existence of a global smooth section $\sigma:S\to M$ (use a smooth
partition of unity and the local smooth splitting of the principal
bundle) so that one can introduce a function $t(x)$ such that
$x=\phi_t(\sigma(\pi(x)))$ where $\phi_t$ is the 1-parameter group
of diffeomorphisms generated by $k$. By construction we have $L_k
t=1$. The spacetime then splits diffeomorphically as
$\mathbb{R}\times S$ where $t$ is the projection on the first
factor. On the trivial principal bundle $\pi:M\to S$ the 1-form
$\omega=g(\cdot,k)/(k\cdot k)$ is an abelian connection
\cite{kobayashi63,minguzzi03}, in the sense that $\omega(k)=1$ and
$L_k\omega=0$. Its (minus) potential is the 1-form on $S$,
$b=-\sigma^{*}\omega$ so that
\[
\omega=\dd t-\pi^{*}b.
\]
The quotient space $S$ endowed with the optical metric is a
Riemannian manifold and defined the positive function on $M$,
$\chi=-k\cdot k$, the spacetime metric reads
\begin{equation} \label{jur}
g=\chi[-(\dd t-\pi^{*}b)^2 +\pi^*h].
\end{equation}
Of course, the function $t$ is not unique as it depends on the
section $\sigma$ that trivializes the principal bundle. Under a
change of section $\sigma'(s)=\phi_{\alpha(s)}(\sigma(s))$ with
$\alpha:S\to \mathbb{R}$,  we get the transformation
\begin{align}
t'&=t-\alpha(s), \label{ks1}\\
b'&=b-\dd \alpha. \label{ks2}
\end{align}
In this setting, to solve the synchronization problem means to
provide arguments for fixing the gauge in some special way (see
Section \ref{hod}).

Let us introduce on $S$ the Finsler metric of Randers type $F:TS\to
\mathbb{R}$ (see \cite{bao00} for an introduction to Finsler spaces)
\[
F(s,y)=\sqrt{h(y,y)}+b(y).
\]
For every piecewise $C^1$ curve $\gamma:[0,1]\to S$, $\tau \to
\gamma(\tau)$, connecting $s_0,s_1\in S$, let us define the
functional
\begin{equation} \label{hst}
\Phi_{s_0,s_1}[\gamma]=\int_0^1 F(\gamma,\gamma_\tau) \dd \tau
=l(\gamma)+\int_\gamma b,
\end{equation}
where $l$ is the optical length of $\gamma$. We define the following
two point function which we call {\em Fermat's map} (we shall find
in Theorem \ref{gdw} that this terminology is appropriate)
\begin{equation} \label{gst} a(s_0,s_1)=\inf_{\gamma}
\Phi_{s_0,s_1}[\gamma].
\end{equation}
Let $x_0=(t_0,s_0)$ and $x_1=(t_1,s_1)$ and let $\eta$ be a future
directed lightlike curve connecting $x_0$ to $x_1$. Since the fibers
of the projection are timelike the projection of $\eta$ is a regular
curve $\gamma:[0,1]\to S$ and Eq. (\ref{jur}) together with the fact
that $\eta$ is future directed  implies that
\begin{equation} \label{kun}
 t_1-t_0=\Phi_{s_0,s_1}[\gamma].
\end{equation}
If it were $\Phi_{s_0,s_0}(\gamma)<0$ for some closed curve
$\gamma$ starting at $s_0$, then
$\eta(\tau)=(t_0+\Phi_{s_0,\gamma(\tau)}[\gamma\vert_{[0,\tau]}],\gamma(\tau))$
would be a future directed lightlike curve connecting $x_0$ and
$x_1$ on the same fiber with $x_1\ll x_0$ which would imply
chronology violation. Thus for every closed curve $\gamma$,
$\Phi_{\gamma(0),\gamma(0)}(\gamma)\ge 0$, and hence $a(s,s)\ge 0$.

We have proved one direction of the following theorem. The fact that a chronological spacetime admitting a complete timelike conformal Killing vector field splits as in Eq.\ (\ref{jur}), namely that it is a principal $\mathbb{R}$-bundle over the space of orbits, was first published by Harris \cite{harris92}. The version given here weakens chronology to non-total viciousness through the results of Clarke and Joshi \cite{clarke88} recalled above, and it adds the characterization through the Fermat's map.
\begin{theorem} \label{ggg}
The spacetime $(M,g)$ is a non-totally vicious spacetime admitting a
complete timelike conformal Killing vector $k$ if and only if it is
diffeomorphic to $\mathbb{R}\times S$ where $(S,h)$ is a Riemannian
manifold, $t: M\to \mathbb{R}$ is the projection on the first factor
and $\pi: M\to S$ is the projection on the second, $k=\p_t)_s$,  the
spacetime metric is given by Eq. (\ref{jur}) where $b\in T^{*}S$,
$\chi: M\to (0,+\infty)$, and $a(s,s)\ge
0$. In particular, in this case $(M,g)$ is chronological. 
\end{theorem}

For the other direction the conformal Killing property for $\p_t)_s$
follows from the independence of the coefficients  (but $\chi$) from
$t$, and the completeness follows from the fact that the fibers are
diffeomorphic to $\mathbb{R}$. The non-total vicious property (and
actually chronology) follows from $a(s,s)\ge 0$, indeed if there
were a closed timelike curve passing through a point $x_0=(t_0,s_0)$
then, since $I$ is open we could find $\epsilon>0$,
$x_1=(t_0-\epsilon,s_0)$ and a timelike  curve $\eta$ connecting
$x_0$ to $x_1$. However, this timelike curve can be replaced by a
piecewise $C^1$ lightlike curve $\tilde{\eta}$ with the same
endpoints, and from Eq. (\ref{kun}) we get $a(s_0,s_0)\le
\Phi_{s_0,s_0}(\tilde{\eta})=-\epsilon<0$, a contradiction. This
concludes the proof of Theorem \ref{ggg}. 

\begin{remark} \label{n:harris}
The splitting problem for spacetimes with an $\mathbb{R}$-action with causal orbits has been studied in a systematic way by Costa e Silva and Flores \cite{costaesilva17}. The causality of stationary spacetimes in which the Randers function $F$ is not positive, as is the case for the splittings of Theorem \ref{ggg} whose time $t$ is not a time function, has been studied by Herrera and Javaloyes \cite{herrera21} in the language of pre-Randers metrics, see also Remark \ref{n:ladder}.
\end{remark}

\begin{remark}
In some cases one can infer from the knowledge of $b$ and $h$ that
$a(s,s)=0$, namely that no closed path in $S$ has negative Fermat travel time, and hence that the spacetime with metric (\ref{jur}) is
chronological. For instance, from Eqs. (\ref{hst})-(\ref{gst}) it
follows that a sufficient condition is $F(s,y)>0$ whenever $y\ne 0$, which ensures, moreover, that the travel time around every non-trivial regular loop is positive. This condition is in turn, by the Cauchy-Schwarz inequality,
equivalent to the bound $h^{-1}(b,b)<1$ for all $s\in S$. We shall
clarify the meaning of this bound in Section \ref{hod}. The condition $a(s,s)=0$ alone does not exclude loops of vanishing travel time: on $S=S^1$ with $h=\dd x^2$ and $b=\dd x$ we have $F(v)=\vert v\vert+v$, so the loops traversed in the negative direction take zero time, and the spacetime with metric $-(\dd t-\dd x)^2+\dd x^2$ is chronological but not causal.
\end{remark}

We have just found in the spacetime $(M,g)$ a  function $t:M\to
\mathbb{R}$ which trivializes the bundle. However, this function $t$
might not be a {\em time function} namely a continuous function such
that ``$x_1< x_2 \Rightarrow t(x_1)<t(x_2)$''. Perlick
\cite{perlick00} suggests to call {\em timing functions} these
functions which share the property $\dd t[k]=1$ and whose constant
slices are diffeomorphic to $S$. We shall investigate in a later
section in which circumstances $t$ can be chosen to be a smooth time
function.

In general the notion of time function is quite restrictive as it
exists if and only if the spacetime is  stably causal.
The concept of {\em utility} is more versatile. Given a preorder
$\preceq$, it is defined through the properties
\begin{align*}
& x_1\preceq x_2 \textrm{ and } x_2 \preceq x_1 \Rightarrow
t(x_1)=t(x_2), \\
&x_1\preceq x_2 \textrm{ and } x_2\npreceq x_1 \Rightarrow
t(x_1)<t(x_2).
\end{align*}

The next result gives more information on the properties of the
Fermat's map and establishes a connection with  the causality
properties of the spacetime. Some of its statements overlap with results published elsewhere, which we point out here. The reflectivity of point 5 was proved by Clarke and Joshi \cite[Prop.\ 3.1]{clarke88}. The first equivalence of point 7 holds in every spacetime, since K-causality coincides with stable causality \cite{minguzzi08b}. The existence of a time function under distinction in point 6 also follows from the result by Javaloyes and S\'anchez on the existence of standard splittings \cite{javaloyes08}. For standard stationary spacetimes, namely when the Randers function $F$ is positive, the characterization of global hyperbolicity of point 8 was first published by Caponio, Javaloyes and S\'anchez \cite[Theor.\ 4.3]{caponio10}, while stationary spacetimes with non-standard splittings, for which $F$ need not be positive, have been studied by Herrera and Javaloyes \cite{herrera21}.

\begin{theorem} \label{gdw}
Let $(M,g)$ be a non-totally vicious spacetime admitting a complete
timelike conformal Killing vector $k$, and let us consider the
splitting given by Theorem \ref{ggg}. We have

\begin{enumerate}
\item an equivalent definition for the Fermat's map is
\[
a(s_1,s_2)=\inf_{x_1\le x_2, \ s_i=\pi(x_i)} [t(x_2)-t(x_1)]
\]
where $\le$ is the causal relation.
\item function $a:S\times S\to \mathbb{R}$ satisfies the conditions
(i)-(iii) of a Fermat's map,
\item $K=\overline{J}$ and an equivalent definition of function $a$ is
\[
a(s_1,s_2)=\inf_{x_1\le_K x_2,\ s_i=\pi(x_i)} [t(x_2)-t(x_1)],
\]
\item two events $x_1,x_2$ are
$K$-related, $x_1\le_K x_2$, iff $t(x_2)-t(x_1)\ge a(s_1,s_2)$,
\item $(M,g)$ is reflective, that is, $y\in \overline{I^{+}(x)}\Leftrightarrow
x\in \overline{I^{-}(y)}$,
\item $(M,g)$ admits a $\le_K$-utility which is also a semi-time
function. Furthermore, if $(M,g)$ is distinguishing then it admits a
time function and hence it is stably causal (actually causally
continuous),
\item $(M,g)$ is stably causal if and only if the preorder $\le_K$ is
antisymmetric if and only if for all $s_1,s_2\in S$,
$a(s_1,s_2)+a(s_2,s_1)=0 \Rightarrow s_1=s_2$,
\item $(M,g)$ is globally hyperbolic if and only if $\le_K$ is
antisymmetric and defined the radar function
$d_r(s_1,s_2)=\frac{1}{2}[a(s_1,s_2)+a(s_2,s_1)]$, for all $r\ge 0$
and $s\in S$, the balls $B(s_0,r)=\{s: d_r(s_0,s)\le r\}$ are
compact.
\end{enumerate}
\end{theorem}

\begin{proof}

The inequality $a(s_1,s_2)\ge \inf_{x_1\le x_2, \ s_i=\pi(x_i)}
[t(x_2)-t(x_1)]$ follows from the fact that if $\gamma: [0,1]\to S$
is a piecewise $C^1$  curve starting from $s_1$ and ending at $s_2$
then $x_1=(0,s_1)$ and $x_2=(\Phi_{s_1,s_2}[\gamma],s_2)$ are
causally connected by the  only future directed lightlike curve
which projects on  $\gamma(t)$ (the lightlike lift of $\gamma$)  and
by Eq. (\ref{kun}), $t(x_2)-t(x_1)=\Phi_{s_1,s_2}(\gamma)$.

The inequality $a(s_1,s_2)\le \inf_{x_1\le x_2, \ s_i=\pi(x_i)}
[t(x_2)-t(x_1)]$ follows from the fact that given two causally
related events $x_1,x_2$, respectively on the fibers of $s_1$ and
$s_2$, we can always find a  piecewise lightlike curve connecting
them. Its projection  $\gamma: [0,1]\to S$ is a necessarily
piecewise $C^1$ curve connecting $s_1$ and $s_2$ and by Eq.
(\ref{kun}), $t(x_2)-t(x_1)=\Phi_{s_1,s_2}(\gamma)$. This concludes
the proof of the first point.

In the definition of function $a$, the point $x_1$ can be kept fixed
because of the translational invariance on the fibers. Since the
composition of two causal curves is a causal curve the fact that
function $a$ satisfies the triangle inequality is trivial. The
property $a(s,s)=0$ follows easily from chronology because the
integral lines of $k$ are timelike and the composition of a causal
curve with a timelike curve can be deformed to give a timelike
curve. Let us prove the continuity of $a$.

Let $(s_1^k,s_2^k)\to (s_1,s_2)$ and let us prove that $\limsup_k
a(s_1^k,s_2^k)\le a(s_1,s_2)$. Indeed, if it were $\limsup_k
a(s_1^k,s_2^k)> a(s_1,s_2)$ then we could find an $\epsilon>0$ such
that for some subsequence, denoted with a different index,
$(0,s_1^i)\nleq (a(s_1,s_2)+\epsilon,s_2^i)$. However, $(0,s_1)\ll
(a(s_1,s_2)+\epsilon,s_2)$, and since the chronological relation is
open it must be for sufficiently large $i$, $(0,s_1^i)\ll
(a(s_1,s_2)+\epsilon,s_2^i)$, a contradiction.

Let $(s_1^k,s_2^k)\to (s_1,s_2)$ and let us prove that
$a(s_1,s_2)\le \liminf_k a(s_1^k,s_2^k)$. Indeed, if it were
$\liminf_k a(s_1^k,s_2^k)< a(s_1,s_2)$ then we could find an
$\epsilon>0$ such that for some subsequence, denoted with a
different index, $(\epsilon,s_1^i)\le (a(s_1,s_2)-\epsilon,s_2^i)$.
However, $(\epsilon/2,s_1)\nleq (a(s_1,s_2)-\epsilon/2,s_2)$ and
since the chronological relation is open it must be for sufficiently
large $i$, $(\epsilon/2,s_1)\ll (\epsilon,s_1^i)\le
(a(s_1,s_2)-\epsilon,s_2^i)\ll (a(s_1,s_2)-\epsilon/2,s_2)$, a
contradiction.

The relation $R=\{(x_1,x_2): t(x_2)-t(x_1)\ge
a(\pi(x_1),\pi(x_2))\}$ is transitive (because $a$ satisfies the
triangle inequality) and closed (because $\pi$, $t$ and  $a$ are
continuous). Furthermore, by claim 1, it contains the causal
relation $\le$. Thus, from the definition of the relation $K$,
$K\subset R$. If $(x_1,x_2)\in R$ then $t(x_2)-t(x_1)\ge
a(s_1,s_2)$ thus for every $\epsilon>0$ we can find a causal curve
connecting $x_1=(t(x_1),s_1)$ to $(t(x_2)+\epsilon,s_2)$ which
implies that $(x_1,x_2) \in \overline{J}$. Thus $K\subset
R\subset \overline{J}\subset K$ and hence all these three
relations coincide. The fact that $K=\overline{J}$ is also a
consequence of reflectivity according to the transverse conformal
ladder \cite{minguzzi08b}.

Clearly $\inf_{x_1\le_K x_2,\ s_i=\pi(x_i)} [t(x_2)-t(x_1)]\le
a(s_1,s_2)$ and $K=\{(x_1,x_2): t(x_2)-t(x_1)\ge
a(\pi(x_1),\pi(x_2))\}$ assures that the equality holds.

The reflectivity property  follows from the existence of a complete
timelike Killing vector as proved by Clarke and Joshi
\cite{clarke88}.

Since $K$ is a closed preorder there is a $\le_K$-utility
(Levin's theorem \cite[Theor. 3]{minguzzi09c}). Since,  $(M,g)$ is
chronological and $K=\overline{J}$ this $\le_K$-utility is
also a semi-time function. If $(M,g)$ is distinguishing then, as it
is reflective, it is causally continuous and hence stably causal,
thus it admits a time function \cite{hawking73,minguzzi19}.

Let us use the splitting of $M$. Suppose that $(t_1,s_1)\le_K
(t_2,s_2)\le_K (t_1,s_1)$, which reads $t_2-t_1\ge a(s_1,s_2)$ and
$t_1-t_2\ge a(s_2,s_1)$. If Eq. (\ref{jub}) holds, summing the two
inequalities $a(s_1,s_2)+a(s_2,s_1)\le 0$, but from (i)-(ii),
$0=a(s_1,s_1)\le a(s_1,s_2)+a(s_2,s_1)$ thus $a(s_1,s_2)+a(s_2,s_1)=
0$ and from Eq. (\ref{jub}), $s_1=s_2$. Plugging back into the
starting inequalities and using (ii), $t_2-t_1\ge 0$ and $t_1-t_2\ge
0$, thus $t_1=t_2$ and hence $(t_1,s_1)=(t_2,s_2)$, namely $\le_K$
is antisymmetric. For the converse, consider the points
$x_1=(t_1,s_1)$ and $x_2=(t_1+a(s_1,s_2),s_2)$ so that $x_1\le_K
x_2$. The antisymmetry condition reads $t_1-(t_1+a(s_1,s_2))\ge
a(s_2,s_1) \Rightarrow$ $t_1=t_1+a(s_1,s_2)$ and $s_1=s_2$, which
can be rewritten as Eq. (\ref{jub}).

If $(M,g)$ is globally hyperbolic then it is stably causal and hence
$\le_K$ is antisymmetric. Furthermore, $(M,g)$ is causally simple
thus $K=J$. Defined $x_1=(0,s_0)$ and $x_2=(2r,s_0)$ we have
that the set $K^{+}(x_1)\cap K^{-}(x_2)=J^{+}(x_1)\cap J^{-}(x_2)$
is compact, but $K^{+}(x_1)\cap K^{-}(x_2)=\{(t,s): t-t_1\ge
a(s_0,s), t_2-t\ge a(s,s_0)\}$ that is $K^{+}(x_1)\cap
K^{-}(x_2)=\{(t,s):     t_1+a(s_0,s) \le t \le t_2- a(s,s_0)\}$. The
projection of this set is compact $\pi(K^{+}(x_1)\cap
K^{-}(x_2))=\{s: 2 d_r(s_0,s)\le t_2-t_1=2r\}$ thus one direction is
proved.

Conversely, let $x_1,x_2\in M$, $x_1=(t_1,s_1)$, $x_2=(t_2,s_2)$.
Let us pick a point $s_0\in S$ and let $t_1',t_2'$ be such that
$t_1-t_1'\ge a(s_0,s_1)$ and $t_2'-t_2\ge a(s_2,s_0)$. Defined
$x_1'=(t_1',s_0)$ and $x_2'=(t_2',s_0)$ we have $x_1'\le_K x_1$ and
$x_2\le_K x_2'$. In particular $J^{+}(x_1)\cap J^{-}(x_2)\subset
K^{+}(x_1')\cap K^{-}(x_2')=\{(t,s): t-t_1'\ge a(s_0,s), t_2'-t\ge
a(s,s_0)\}\subset [t_1'+\inf_{s\in C} a(s_0,s), t_2'-\inf_{s\in C}
a(s,s_0)] \times C $ where $C=\{s: d_r(s_0,s)\le (t_2'-t_1')/2\}$ is
compact. We conclude that $J^{+}(x_1)\cap J^{-}(x_2)$ has compact
closure and since $(M,g)$ is stably causal it is also non-total
imprisoning and by \cite[Corollary 3.3]{minguzzi08e} globally
hyperbolic.  $\square$
 \end{proof}

\begin{remark}
For $F>0$ the last point of the previous theorem was proved in
\cite[Theor. 4.3]{caponio10} using the
well known fact \cite{lichnerowicz55,perlick90,perlick90b,perlick00}
that trajectories of lightlike geodesics on $M$ project on
stationary points of the Finsler Lagrangian on $S$. We have avoided
the introduction of this machinery because on the one hand it is not
essential for our later purposes and because on the other hand our
proof shows that it is not strictly required for the above results.
Nevertheless, as it is proved in that paper \cite{caponio10}, causal
simplicity corresponds to  convexity (namely the existence of
connecting trajectories for the Finsler Lagrangian) and this
property is not expressible through the potential $a$.

\end{remark}

\begin{remark} \label{n:ladder}
The characterization of the levels of the causal ladder of a stationary spacetime through the Fermat's map, or equivalently through the Randers metric $F$ and its symmetrization, has been developed in \cite{caponio10} for the case $F>0$ and in \cite{herrera21} for the general case. If the space $S$ is compact the picture simplifies considerably: as shown in \cite{minguzzi12bolza}, the causal behavior of the spacetime is determined by a single synchronization independent number, the Sagnac ratio $\beta$ of Definition \ref{n:sagdef}. The spacetime is totally vicious if $\beta>1$, chronological but non-distinguishing if $\beta=1$, and globally hyperbolic if and only if it is distinguishing, if and only if $\beta<1$ (Theorem \ref{n:sync}). This trichotomy also follows from Harris's classification of stationary-complete spacetimes through the weight of the fundamental cocycle \cite{harris15}, which coincides with $\beta$, together with the compactness of $S$ \cite[Sect.\ 12.2]{minguzzi12bolza}.
\end{remark}

The previous theorem, by making a correspondence between well known
spacetime properties and properties of function $a$, suggests to
extend some definitions to the more abstract settings given by
$(S,a)$ where $a$ satisfies (i)-(iii). For instance one could give a
definition of global hyperbolicity for $(S,a)$ by imposing the
compactness of the closed balls obtained through the radar two-way
function (we shall see later that it is a distance).

The next theorem does not assume the Lorentzian spacetime case.

\begin{theorem} \label{ghq}
On $\mathbb{R}\times S$ endowed with the preorder $\le_K$ defined by
(\ref{gae}) where the Fermat's map satisfies  (i)-(ii) the preorder
$\le_K$ is antisymmetric if and only if for all $s_1,s_2\in S$
\begin{equation} \label{jub}
a(s_1,s_2)+a(s_2,s_1)=0 \Rightarrow s_1=s_2.
\end{equation}
Furthermore the projection $t:\mathbb{R}\times S\to \mathbb{R}$ is
a utility if and only if for all $s_1,s_2\in S$, $a(s_1,s_2)> 0$ or
$a(s_1,s_2)=a(s_2,s_1)=0$.
\end{theorem}

\begin{proof}

The proof of the first statement goes exactly as in the proof of
Theorem \ref{gdw}.

Let us come to the last statement and let us prove that $t$ is an
utility. From $(t_1,s_1)\le_K (t_2,s_2)\le_K (t_1,s_1)$, we obtain,
arguing as above, $a(s_1,s_2)+a(s_2,s_1)\le 0$ which implies,
$a(s_1,s_2)=a(s_2,s_1)= 0$, thus plugging back into the starting
inequalities $t_2-t_1\ge a(s_1,s_2)$ and $t_1-t_2\ge a(s_2,s_1)$ we
get $t_2=t_1$, that is, the first property of a utility is proved.
The second property can be written ``$x_1\le_K x_2
 \Rightarrow t(x_1)<t(x_2)$ or $x_2\le_K x_1$". But $x_1\le_K x_2$ reads
$t_2-t_1\ge
 a(s_1,s_2)$ thus either $t_2>t_1$ or $t_2=t_1$ and
 $a(s_1,s_2)=a(s_2,s_1)=0$. In the last case it follows that
 $x_2\le_K x_1$, thus the second property is proved.

 For the converse suppose $t$ is a utility and let $x_1=(0,s_1)$,
 $x_2=(a(s_1,s_2),s_2)$ so that $x_1\le_K x_2$. The first condition
 for a utility leads to $ -a(s_1,s_2)\ge a(s_2,s_1) \Rightarrow
 a(s_1,s_2)=0$, or equivalently $a(s_1,s_2)+a(s_2,s_1)=0\Rightarrow
 a(s_1,s_2)=a(s_2,s_1)=0$. The second condition reads instead $ -a(s_1,s_2)< a(s_2,s_1) \Rightarrow
 a(s_1,s_2)>0$, or equivalently $a(s_1,s_2)+a(s_2,s_1)>0\Rightarrow
 a(s_1,s_2),a(s_2,s_1)>0$. Taken together these condition are
 equivalent to that given in the theorem.  $\square$

\end{proof}

While Theorem \ref{gdw} proves that in the Lorentzian spacetime case
a utility exists, it gives no clue on how to obtain it through a
synchronization of clocks. The previous theorem makes a step in this
direction establishing some conditions on the Fermat's map.

\subsection{Einstein's synchronization and Sagnac's function}
\label{due3}
 \label{uhm}

There are two {\em synchronization independent} functions which will
be important in what follows (i.e. they are invariant under
replacements of the form (\ref{res})). They are the {\em two-way
time} ($r$ stands for {\em radar} or {\em round-trip})
\begin{equation}
r(s_1,s_2)=a(s_1,s_2)+a(s_2,s_1),
\end{equation}
and the the function $w(s_1,s_2,s_3)$ defined as the difference
between the times taken by the signal to cover the triangular path
$s_1s_2s_3s_1$ and the opposite path $s_1s_3s_2s_1$
\[
w(s_1,s_2,s_3)=[a(s_1,s_2)+a(s_2,s_3)+a(s_3,s_1)]-[a(s_1,s_3)+a(s_3,s_2)+a(s_2,s_1)].\]
The function $w$ gives therefore the Sagnac effect over  the
triangular path (the Sagnac holonomy is half the Sagnac effect).
Each term in square brackets is independent of the clock
synchronization, thus the same holds for $w$. We shall say that
${\bf \triangle}$ holds if $w=0$ namely if the time taken by the
signal to cover a triangular path does not depend on the direction
taken by the signal along the triangle. This property is usually
called {\em Reichenbach's round-trip condition}. For a discussion of
the history of the round-trip properties the reader is referred to
\cite{minguzzi10c}. I wish to stress here that by {\em triangle} it
is not meant a geometric object but only a sequence of three points.
The signal might actually propagate as a wave so that the optical
geometry limit would not hold and one could not even speak of the
trajectory of the signal. Nevertheless, even in this case, the
Fermat's map would make sense and so our formalism.

The Einstein convention establishes
\cite{einstein05,poincare04a,poincare04b,jammer06} that clock
$s_2\in S$ is synchronized with clock $s_1\in S$  if, calling
${t}_1$ the time of departure at ${s}_1$, the signal is received at
${s}_2$ at
\begin{equation}
{t}_2={t}_1+r({s}_1,{s}_2)/2
\end{equation}
or in other words, if $a({s}_1,{s}_2)=r({s}_1,{s}_2)/2$.

It is well established \cite{macdonald83} (see also
\cite{laue11,silberstein14,minguzzi02d,jammer06}) that Einstein's
synchronization method is reflexive, symmetric and transitive and
thus can be consistently applied to give a time foliation  if and
only if ${\bf z=0}$ and ${\bf \triangle}$ hold (this fact will
follow as a particular case of our subsequent result).

Unfortunately, ${\bf \triangle}$ may not hold for fundamental
reasons. For instance, in the case of light propagating in vacuum,
but then also in case of other types of signals, the rotation of the
Earth frame spoils it, (the contribution to $w$ is indeed twice the
Sagnac holonomy which can reach  one-half of microsecond
\cite{post67,ashtekar75,anandan81,minguzzi03,ashby03}). It can, and
in fact is \cite[Sect. 1.2]{mills06} most often violated for less
fundamental reasons for instance because the propagation speed is
actually anisotropic with respect to an exterior isotropic
reference. Part of this work will be devoted to the resolution of
this problem through a systematic development of the theory of the
 synchronization  of clocks.

\subsection{Radar distance and synchronization classes} \label{mhj}
\label{due4}

Let us define the {\em radar distance} (we use the same units for
time intervals and lengths)
\begin{equation}
d_r(s_1,s_2)=\frac{1}{2}\,r(s_1,s_2) .
\end{equation}

From (i) and (ii) it is possible to prove

\begin{lemma} \label{hwx}
Assume that the Fermat's map  satisfies (i) and (ii), then the
relation defined by $s\sim s' \Leftrightarrow d_r(s,s')=0$ is an
equivalence relation. The function  $d_r$ depends only on the
(synchronization) equivalence classes and once passed to the
quotient space provides a distance: thus it is non-negative, it
satisfies the triangle inequality
\begin{equation} \label{gvd}
d_r([s_1],[s_3])\le d_r([s_1],[s_2])+d_r([s_2],[s_3]).
\end{equation}
and $d_r([s],[s'])=0 \Leftrightarrow [s]=[s']$. Clocks belonging to
the same equivalence class $[s]$ can all be reset in such a way that
for any pair $s,s'\in [s]$, we have $a(s,s')=0$. Once this is done
for each class, the function $a$ passes to the quotient space,
namely it is independent of the representatives of the classes.
Since $w$ is constructed from $a$, it also passes to the quotient.
Finally, the equality holds in the triangle inequality (\ref{gvd})
if and only if
\begin{align}
a([s_1],[s_3])&= a([s_1],[s_2])+a([s_2],[s_3]), \\
a([s_3],[s_1])&= a([s_3],[s_2])+a([s_2],[s_1]).
\end{align}

\end{lemma}

\begin{proof}
From (ii) $0=a(s,s)=r(s,s)/2$ thus $\sim$ is reflexive. Since $r$ is
symmetric $\sim$ is symmetric.
 From (i) we have
\begin{align}
a(s_1,s_3)&\le a(s_1,s_2)+a(s_2,s_3),  \label{bq1}\\
a(s_3,s_1)&\le a(s_3,s_2)+a(s_2,s_1). \label{bq2}
\end{align}
Summing the two equations we get
\[
d_r(s_1,s_3)\le d_r(s_1,s_2)+d_r(s_2,s_3).
\]
In particular if  $d_r(s_1,s_2)=d_r(s_2,s_3)=0$ then
$d_r(s_1,s_3)\le 0$, and since from (ii) $0=d_r(s_1,s_1)\le
d_r(s_1,s_3)+d_r(s_3,s_1)=2d_r(s_1,s_3)$ we have $d_r(s_1,s_3)=0$
and thus $\sim$ is transitive. Let $s_1,s_1'\in [s_1]$ and let us
show that $d_r(s_1,s_2)=d_r(s_1',s_2)$. A similar proof repeated for
the second argument shows then that $r$ depends only on the
equivalence classes. We have $d_r(s_1,s_2)\le
d_r(s_1,s_1')+d_r(s_1',s_2)\le d_r(s_1',s_2)$, and with the roles of
$s_1$ and $s_1'$ inverted we get the opposite inequality and hence
$d_r(s_1,s_2)=d_r(s_1',s_2)$. On the quotient space $d_r$ is a
distance because for every $[s]$, $d_r([s],[s])=d_r(s,s)=0$ and if
$d_r([s],[s'])=0$ then $d_r(s,s')=0$ which means that $s$ and $s'$
stay in the same equivalence class, that is, $[s]=[s']$. The fact
that $d_r$ is non-negative follows from the triangle inequality and
(ii): $0=d_r(s_1,s_1)\le d_r(s_1,s_2)+d_r(s_2,s_1)=2d_r(s_1,s_2)$.

It remains to show that clocks can be reset inside the same class in
such a way that $a$ vanishes if evaluated on a pair belonging to the
same class. The fact that $d_r$ vanishes inside the class means that
$a$ is antisymmetric there.

Let $s_1,s_2,s_3\in [s]$, we have
\begin{align*}
0&=r(s_2,s_3)=a(s_2,s_3)+a(s_3,s_2)\le
a(s_2,s_3)+a(s_3,s_1)+a(s_1,s_2), \\
0&=r(s_2,s_1)=a(s_2,s_1)+a(s_1,s_2)\le
a(s_2,s_1)+a(s_1,s_3)+a(s_3,s_2).
\end{align*}
However, since $a$ is antisymmetric on the class these two
right-hand sides are one the opposite of the other thus we conclude
\begin{equation} \label{bqo}
a(s_1,s_2)+a(s_2,s_3)+a(s_3,s_1)=0,
\end{equation}
which is a relation independent of the synchronization of the
clocks. Now we establish a synchronization procedure for the clocks
in the class, namely clock $s_2$ is synchronized with clock $s_1$ by
resetting clock $s_2$ in such a way that $a(s_1,s_2)=0$. Since
$0=r(s_1,s_2)=a(s_1,s_2)+a(s_2,s_1)$ and $r$ is independent of the
synchronization, one immediately finds that this synchronization is
symmetric. It is also reflexive because of (ii) and finally it is
transitive because of Eq. (\ref{bqo}). Thus a synchronization of the
clocks inside the same class can be accomplished so that for every
$s,s' \in [s]$, $a(s,s')=0$. (Equivalently, fix $s_1\in [s]$ and
move forward the generic clock $s'\in [s]$ by a quantity
$c(s')=a(s',s_1)$; the fact that this operation makes the Fermat map
vanish on the class follows from Eq. (\ref{bqo}).)

Now suppose that clocks belonging to $[s_1]$ and $[s_2]$ have been
synchronized in this way, we have to show that if $s_1' \sim s_1$
and $s_2' \sim s_2$ then $a(s_1',s_2')=a(s_1,s_2)$. But
$a(s_1',s_2')\le a(s_1',s_1)+a(s_1,s_2)+a(s_2,s_2')\le a(s_1,s_2)$,
and exchanging the pair $(s_1',s_2')$ with $(s_1,s_2)$ we get the
thesis.

The last statement follows immediately from Eqs.
(\ref{bq1})-(\ref{bq2}).
\end{proof}

The condition $a=0$ inside the class means that there is no two-way
delay in the communication between the clocks and that two events
with different space projections can causally influence each other.
Physically this means that there is {\em causality violation}
whenever the synchronization classes are not trivial (see also
Section \ref{hin}). Since we shall be mostly interested in assuming
causality, we shall, in those cases, keep the letter $S$ for the
quotient space.

For most part of  the present synchronization application the
classes can be considered trivial. Nevertheless, we shall show in a
different work that the previous lemma can also be applied to the
syntonization problem in relativity theory.
There  each syntonization class
will be an inertial frame and the elements of the class will be the
clocks at rest in the inertial frame.

Let us recall that a quasi-metric has all the properties of a metric
but for the symmetry. The concept of synchronization class allows to
simplify the last statement of Theorem \ref{ghq} as follows
\begin{theorem} \label{uhk}
In the spacetime $\mathbb{R}\times S$  endowed with the preorder
(\ref{gae}) and in which the Fermat's map satisfies (i)-(ii), the
projection $t:\mathbb{R}\times S\to \mathbb{R}$ is a utility if and
only if
\begin{itemize}
\item[(a)] $a(s_1,s_2)\ge 0$,
\item[(b)] $a(s_1,s_2)=0$ iff $s_1\sim s_2$,
\item[(c)] $a(s_1,s_3)\le a(s_1,s_2)+a(s_2,s_3)$.
\end{itemize}
Moreover, this is the case if and only if the Fermat's map passes to
the quotient $S/\!\!\sim$ and there it gives a quasi-metric.
\end{theorem}

\begin{proof}
It is a straightforward reformulation of the last statement of
Theorem \ref{ghq} but for the fact that $a$ passes to the quotient.
This is so because if $s_1\sim s'_1$ then, using (b), $a(s_1,s_2)\le
a(s_1,s_1')+a(s_1',s_2)\le a(s_1',s_2)$, and inverting the roles of
$s_1$ and $s_1'$ we find $a(s_1,s_2)= a(s_1',s_2)$, that is, the
Fermat's map passes to the quotient on the first argument. The proof
for the second argument is analogous.  $\square$
\end{proof}

The next result establishes some bounds on the Sagnac function $w$

\begin{theorem}
Let us denote
\begin{align*}
w&= w(s_1,s_2,s_3),\\
P&=2S=d_r(s_1,s_2)+d_r(s_2,s_3)+d_r(s_3,s_1), \\
A&=\sqrt{S(S-d_r(s_1,s_2))(S-d_r(s_2,s_3))(S-d_r(s_3,s_1))} ,
\end{align*}
where $P$ and $S$ are the perimeter and semi-perimeter of the
triangle of vertices $s_1,s_2,s_3$ and $A$ is a quantity calculated
through  Heron's expression for the area of the triangle in
Euclidean geometry. Then the following inequalities hold
\begin{align}
\vert w\vert&\le\frac{2}{3} \,P, \label{hdg}\\
\vert w \vert& \le \frac{4}{\sqrt[4]{3}} \,\sqrt{A} . \label{hdf}
\end{align}
\end{theorem}

\begin{proof}
The inequalities (\ref{bq1})-(\ref{bq2}) can be rewritten
\begin{align*}
\frac{1}{2}w(s_1,s_2,s_3)&\ge
-[d_r(s_1,s_2)+d_r(s_2,s_3)-d_r(s_3,s_1)] ,\\
\frac{1}{2}w(s_1,s_2,s_3)&\le
[d_r(s_1,s_2)+d_r(s_2,s_3)-d_r(s_3,s_1)] ,
\end{align*}
thus
\[
\vert w\vert\le 4(S-d_r(s_3,s_1)).
\]
With a cyclic permutation we also obtain
\begin{align*}
\vert w\vert&\le 4(S-d_r(s_1,s_2)),\\
\vert w\vert&\le 4(S-d_r(s_2,s_3)).
\end{align*}
Summing them we get Eq. (\ref{hdg}) and multiplying them and using
Eq. (\ref{hdg}) we get Eq. (\ref{hdf}).  $\square$
\end{proof}

The inequality (\ref{hdf}) involving $A$ is especially interesting
since it proves that when the triangle is degenerate (the triangle
inequality holds with the equality sign) then $w$ vanishes. This
fact suggest that $w$ is connected to the area.

The inequalities hold because of   causality. Indeed, the time it
takes the signal to cover a closed path is generically the sum of
two contributions, e.g.
\[
a(s_1,s_2)+a(s_2,s_3)+a(s_3,s_1)=P+\frac{1}{2}w(s_1,s_2,s_3) ,
\]
that is, the dynamic phase which is proportional to the length of
the path, and the (Sagnac) holonomy. The latter cannot be larger
than the former (it is summed or subtracted depending on the
direction taken along the path) otherwise the signal would arrive
before it departed.

As a case study recall  that in general relativity the Sagnac effect
over the path $\sigma$ is given by the integral of the vorticity
2-form over a surface $\Sigma$ such that $\sigma=\p \Sigma$
\cite{ashtekar75,minguzzi03} (this formula is obtained from Eq. (23)
of \cite{minguzzi03}).
\begin{equation} \label{sagnac3}
2 \int_{\Sigma}w_{i j} \, \dd x^{i} \wedge \dd x^{j} \, .
\end{equation}

As a consequence, for small area elements the Sagnac effect is
proportional to the area and to the scalar product of the vorticity
vector with the normal to the area element. In other words, provided
the area element is small, the Sagnac effect goes quadratically with
the size (diameter) of the surface considered. The  bound
(\ref{hdf}) proves that this quadratic behavior cannot hold for
large areas because the Sagnac effect is {\em linearly} bounded with
respect to the size of the surface. This bound is satisfied for
small areas because of the mentioned quadratic behavior. As the area
increases the vorticity vector must (i) decrease in magnitude at
least inversely proportional with  the distance from a reference
point, (ii) have an increasing angle with respect to the surface
normal (possibly with a change of sign of the scalar product as it
happens  on the equipotential surface of the Earth).

\section{Synchronization over a measure space} \label{mqz}
\label{tre}

In this section we are going to answer in the positive the following
two natural questions.
\begin{itemize}
\item[(a)] Is it possible to conceive a synchronization method which
works in absence of round-trip assumptions and does not privilege
any space point (clock)?
\item[(b)] Is it possible to solve (a) in such way that  the coordinate time so constructed respects causality
(namely be a continuous utility).
\end{itemize}

Concerning (a), the mentioned difficulty with the possible failure
of ${\bf \triangle}$ can be solved modifying Einstein's convention
as follows.

The definition of function $w$ can be written
\[
w(s_1,s_2,s_3)=[a(s_1,s_2)+a(s_2,s_3)+a(s_3,s_1)]-[a(s_1,s_3)+a(s_3,s_2)+a(s_2,s_1)],
\]
from which we infer that $w$ is skew-symmetric and that it satisfies
the 2-cocycle condition \cite{spanier66}
\begin{equation} \label{njw}
w(s_2,s_3,s_4)-w(s_3,s_4,s_1)+w(s_4,s_1,s_2)-w(s_1,s_2,s_3)=0 .
\end{equation}

\begin{theorem} \label{gaw}
Suppose that there is a   skew-symmetric synchronization independent
function $\delta(s_1,s_2)$ built from the Fermat's map $a:S\times
S\to \mathbb{R}$, such that
\begin{equation} \label{nse}
w(s_1,s_2,s_3)=\delta(s_1,s_2)+\delta(s_2,s_3)+\delta(s_3,s_1).
\end{equation}
in such a way that $w$ becomes a 2-coboundary. The synchronization
method, which modifies Einstein's,  and consists in sending at time
$t_1$ a signal from $s_1$ to $s_2$ and in setting $s_2$ so that it
reads a time
\begin{equation}
{t}_{2}={t}_1+[r({s}_1,{s}_2)+\delta({s}_1,{s}_2)]/2 ,
\end{equation}
or equivalently
\begin{equation} \label{hwc}
a({s}_1,{s}_2)=[r({s}_1,{s}_2)+\delta({s}_1,{s}_2)]/2 ,
\end{equation}
 is reflexive, symmetric and transitive and thus
leads to a consistent synchronization.
\end{theorem}

\begin{proof}
 Reflexivity is obvious
because
$r(s,s)=\delta(s,s)=0$. Symmetry follows from 
 \[
a(s_1,s_2)+a(s_2,s_1)=r(s_1,s_2) .
\]
Indeed if $s_2$ is synchronized with clock $s_1$
\[a(s_1,s_2)=[r(s_1,s_2)+\delta(s_1,s_2)]/2 ,\]
which plugged into the previous equation gives
\[a(s_2,s_1)=[r(s_2,s_1)+\delta(s_2,s_1)]/2,\]
namely $s_1$ is  synchronized with $s_2$. Thanks to symmetry we can
generically say that clocks {\em are} synchronized without
specifying the direction.

In order to prove transitivity let $s_2$ be synchronized with $s_1$
and let $s_3$ be synchronized with $s_2$. From the definition of $w$
\[
w(s_1,s_2,s_3)=[a(s_1,s_2)-a(s_2,s_1)]+[a(s_2,s_3)-a(s_3,s_2)]+[a(s_3,s_1)-a(s_1,s_3)].
\]
Since clocks $s_1$ and $s_2$ are synchronized, and clocks $s_2$ and
$s_3$ are synchronized we get
\begin{equation} \label{ujn}
w(s_1,s_2,s_3)=\delta(s_1,s_2)+\delta(s_2,s_3)+[a(s_3,s_1)-a(s_1,s_3)].
\end{equation}
and using $a(s_1,s_3)+a(s_3,s_1)=r(s_1,s_3)$ and Eq. (\ref{nse}) we
obtain
\[
a(s_1,s_3)=[r(s_1,s_3)+\delta(s_1,s_3)]/2
\]
namely $s_1$ and $s_3$ are synchronized.  $\square$
\end{proof}

By using Eq. (\ref{njw}) it is easy to show that the function
\begin{equation} \label{nka}
\delta(s_1,s_2)=\int_S w(s_1,s_2,s) \,\dd \mu(s),
\end{equation}
is synchronization independent and does indeed satisfy Eq.
(\ref{nse}).
 This formula establishes a new (consistent)
synchronization convention which reduces to Einstein's for $w=0$
(e.g. for irrotational frames). In short one has to introduce an
offset of the clocks with respect to Einstein's given by the average
of the Sagnac holonomy.

It is interesting to consider what happens if the measure $\mu$ is
entirely supported at a point  $s_0\in S$, namely $\mu(s_0)=1$. In
this case $\delta(s_0,s)=0$ for every $s\in S$, which means that
every clock is Einstein synchronized with $s_0$. However, this
choice would privilege clock $s_0$ with  no reason. In fact the
problem is in general not that of synchronization (otherwise one
could just emit a signal from a point and fix the zero of any other
clock at the time of reception of the signal) but rather that of
{\em democratic} synchronization, that is, in absence of arbitrary
features that privilege some space point. If this can be done the
democratic time becomes an intrinsic distinguished property of the
system.

Our general solution proves that a natural time foliation exists in
all those cases in which  $S$ admits a natural finite measure. In
some cases (we shall see some examples) the space $S$ can have
infinite measure but in this case one has to specify how the
integral must be taken. In other cases the integral is convergent
for instance if the Sagnac function $w$ vanishes (as it happens in
Minkowski spacetime) or goes rapidly to zero.

In practical applications one has most often a finite measure
because the clocks are finite.
This is the case of a computational network in which each node can
be assigned the same measure (as one would do in the client
democratic Berkeley protocol \cite{tanenbaum07}), or in which the
measure can be adapted to express the relative importance of the
nodes (as in the hierarchical NTP protocol).

Coming to (b) one would like to prove that this time assignment
respects causality. In other words we need to prove the next theorem
(we recall that a quasi-metric is a metric which lacks symmetry)

\begin{theorem} \label{czo}
Let $(S,\mathscr{T})$ be a topological space and let $\mu$ be a
strictly positive unitary measure over a $\sigma$-algebra $\Sigma$
that contains the topology (thus $\Sigma$ is at least as fine as the
Borel $\sigma$-algebra on $S$ and every open set has non-vanishing
measure). Moreover, let the function $a: S\times S \to \mathbb{R}$
satisfy (i), (ii) and (iii), and let $\delta_a(s_1,s_2)=a(s_1,s_2)-a(s_2,s_1)$ be its antisymmetric part (which, unlike the function $\delta$ of Eq.\ (\ref{nka}), depends on the synchronization). Then $\vert w(s_1,s_2,s)\vert\le r(s_1,s_2)$ for all $s_1,s_2,s\in S$, so that the integral in Eq.\ (\ref{nka}) exists. Clocks are synchronized according to
Eqs. (\ref{hwc})-(\ref{nka}), namely $a': S\times S \to \mathbb{R}$
is such that
\begin{equation} \label{hch}
a'(s_1,s_2)=\frac{1}{2}[r(s_1,s_2)+\int_S w(s_1,s_2,s )\dd \mu(s)]
\end{equation}
if and only if  $a'$ is related to $a$ by (resetting  of the clocks)
\begin{equation} \label{vaw}
a'(s_1,s_2)=a(s_1,s_2)+c(s_2)-c(s_1)
\end{equation}
 where, for some (and hence every) $s_0\in S$ and up to an irrelevant additive constant,
\begin{equation} \label{vqz}
c(s_1)=\frac{1}{2} \int_{S} [\delta_a(s_1,s)-\delta_a(s_0,s)]\, \dd \mu(s),
\end{equation}
the integrand being bounded. If, moreover, the functions $s\mapsto r(s_0,s)$ and $s\mapsto \delta_a(s_0,s)$ are $\mu$-integrable for some (and hence every) $s_0\in S$, and $\delta_a$ is $\mu\otimes\mu$-measurable (all this holds if $S$ is compact and second countable), then the additive constant can be chosen so that
\begin{equation} \label{vqh}
c(s_1)=\frac{1}{2} \int_{S} [a(s_1,s)-a(s,s_1)] \dd \mu(s),
\end{equation}
and this is the only $\mu$-integrable resetting (i.e. Eq. (\ref{vaw})) such that (Eq. (\ref{pyw}) is needed to remove the global
translation freedom)
\begin{align}
\int_S c(s)\dd \mu(s)&=0, \label{pyw}\\
\frac{1}{2} \int_{S} [a'(s_1,s)-a'(s,s_1)] \dd \mu(s)&=0
\label{pyz}.
\end{align}
The Fermat's map  $a'$ for the so synchronized clocks is continuous,
can pass to the quotient $S/\!\!\sim$, and satisfies
\begin{itemize}
\item[(a)]  $a'(s_1,s_2)\ge 0$,
\item[(b)] $a'(s_1,s_2)=0$ iff $s_1\sim s_2$,
\item[(c)] $a'(s_1,s_3)\le a'(s_1,s_2)+a'(s_2,s_3)$.
\end{itemize}
In other words the function $a'$ is a quasi-metric on $S/\!\!\sim$,
and hence the synchronization leads to a time which is a continuous
utility for the preorder given by Eq. (\ref{gae}) on
$\mathbb{R}\times S$ (see Theorem \ref{uhk}).

\end{theorem}

\begin{proof}
Let $a$ be the two-point function in any synchronization. The inequalities
\begin{align}
a(s_1,s_3)&\le a(s_1,s_2)+a(s_2,s_3), \label{ou1}\\
a(s_3,s_2)&\le a(s_3,s_1)+a(s_1,s_2), \label{ou2}
\end{align}
can be rewritten respectively
\begin{align*}
r(s_1,s_3)-r(s_1,s_2)-r(s_2,s_3)&\le w(s_1,s_2,s_3), \\
r(s_3,s_2)-r(s_3,s_1)-r(s_1,s_2)&\le  w(s_3,s_1,s_2) ,
\end{align*}
Summing the two equations we get
\begin{equation} \label{ihj}
-r(s_1,s_2)\le w(s_1,s_2,s_3),
\end{equation}
and exchanging $s_1$ and $s_2$, by the antisymmetry of $w$, $w(s_1,s_2,s_3)\le r(s_1,s_2)$. Since $w(s_1,s_2,\cdot)$ is continuous, hence measurable, and bounded, the integral in Eq.\ (\ref{hch}) exists. Similarly, by (i), $-a(s_0,s_1)\le a(s_1,s)-a(s_0,s)\le a(s_1,s_0)$ and $-a(s_0,s_1)\le a(s,s_0)-a(s,s_1)\le a(s_1,s_0)$, thus
\begin{equation} \label{n:dbound}
-2a(s_0,s_1)\le \delta_a(s_1,s)-\delta_a(s_0,s)\le 2a(s_1,s_0) \qquad \textrm{for every } s\in S,
\end{equation}
so the integrand in Eq.\ (\ref{vqz}) is bounded, and $c$ changes by a constant if $s_0$ is changed. Moreover, $\vert c(s_1)-c(s_0)\vert\le \max\{\vert a(s_0,s_1)\vert,\vert a(s_1,s_0)\vert\}$, which goes to zero for $s_1\to s_0$ by (ii) and (iii), thus $c$ is continuous, and so is $a'$. From the
definition of $w$
\[
w(s_1,s_2,s_3)=\delta_a(s_1,s_2)+[\delta_a(s_2,s_3)-\delta_a(s_0,s_3)]-[\delta_a(s_1,s_3)-\delta_a(s_0,s_3)].
\]
Integrating this equation in $s_3$
\[
\int_S w(s_1,s_2,s )\dd \mu(s)=\delta_a(s_1,s_2)+2c(s_2)-2c(s_1)
\]
thus
\[
\frac{1}{2}[r(s_1,s_2)+\int_S w(s_1,s_2,s )\dd
\mu(s)]=a(s_1,s_2)+c(s_2)-c(s_1)
\]
that is, Eq. (\ref{hch}) holds if and only if Eqs.
(\ref{vaw})-(\ref{vqz}) hold, the resetting being determined by $a'$ up to a constant.

Assume now that $r(s_0,\cdot)$ and $\delta_a(s_0,\cdot)$ are integrable. Since $\vert a\vert, \vert a^T\vert\le (r+\vert\delta_a\vert)/2$, Eq.\ (\ref{n:dbound}) gives $\vert \delta_a(s_1,s)\vert\le \vert\delta_a(s_0,s)\vert+r(s_0,s_1)+\vert\delta_a(s_0,s_1)\vert$, thus $\delta_a(s_1,\cdot)$ is integrable for every $s_1$ and $\delta_a$ is $\mu\otimes\mu$-integrable. Thus Eq.\ (\ref{vqz}) differs from Eq.\ (\ref{vqh}) by a constant, and with the choice (\ref{vqh}) $c$ is integrable and $\int_S c\, \dd\mu=0$ by Fubini's theorem and the antisymmetry of $\delta_a$. Moreover, $\frac{1}{2}\int_S[a'(s_1,s)-a'(s,s_1)]\dd \mu(s)=c(s_1)+\int_S c\,\dd\mu-c(s_1)=0$. Conversely, for an integrable $c$, from Eqs.
(\ref{pyw})-(\ref{pyz}) and Eq. (\ref{vaw}) it follows Eq.
(\ref{vqh}).

As for (a), by Eqs.
(\ref{hwc})-(\ref{nka}) and (\ref{ihj}), if $a'$ is the function obtained with the
new synchronization method,
\[
2 a'(s_1,s_2)=r(s_1,s_2)+\int_S w(s_1,s_2,s )\dd \mu(s)=\int_S
[r(s_1,s_2)+w(s_1,s_2,s )]\dd \mu(s)\ge 0 ,
\]
The equality cannot hold if $r(s_1,s_2)> 0$ indeed in this case for
$s=s_1$ we have  $r(s_1,s_2)+w(s_1,s_2,s )>0$ and by continuity the
inequality holds in an open neighborhood of $s_1$. By strict
positivity of the measure this fact implies that the right-hand side
of the previous displayed equation is
 positive.

The converse, namely the fact that $s_1\sim s_2$ implies
$a'(s_1,s_2)=0$ follows easily observing that $s_1\sim s_2$ reads
$0=r(s_1,s_2)=a'(s_1,s_2)+a'(s_2,s_1)$ and using the non negativity
of $a'$. The function $a'$ satisfies the triangle inequality because
this equality is synchronization independent and from that and
$s_1\sim s_2 \Rightarrow a'(s_1,s_2)=0$ it is easy to show that $a'$
passes to the quotient. $\square$

\end{proof}

\begin{remark}
Eq.  (\ref{vqh}) (or Eq.\ (\ref{vqz}) without integrability assumptions) suggests another equivalent method for resetting
the clocks. Starting from any synchronization one measures the
quantities $a(s_1,s_2)$ and moves the generic clock $s$ forward by
the quantity $c(s)$.
\end{remark}

This theorem is particularly interesting in the Lorentzian spacetime
case. We have obtained the existence of a  continuous utility in
Theorem \ref{gdw}, by using rather general topological theorems
(Levin's theorem). On the contrary here we have proved this
existence result in a constructive way by selecting a particular
time coordinate. This result proves that our way of constructing a
time coordinate that respect causality is very satisfying as it
allows to recover causality results obtained by different means.

\subsection{Minimization of the anisotropy functional} \label{sde}
\label{tre2}

Let $(S,a)$ be a space endowed with a Fermat's map which satisfies
(i)-(iii).  There is isotropy of propagation between points $s_1$
and $s_2$ if $a(s_1,s_2)=a(s_2,s_1)$ thus the quantity
$b=[a(s_1,s_2)-a(s_2,s_1)]^2/4$ gives a measure of the anisotropy of
propagation between the two points. Of course by moving forward
$s_1$ of a quantity $c=[a(s_1,s_2)-a(s_2,s_1)]/2$ one could send
this quantity $b$ to zero. The quantity $b$ therefore represents the
square of the error one would make assigning to $s_1$ the starting
synchronization in place of the final one which is Einstein's.
However this resetting operation may introduce an anisotropy between
points $s_1$ and, say, $s_3$ for which the propagation was
previously isotropic. One would like to find a synchronization of
the clocks that minimizes the following anisotropy functional
\begin{equation} \label{wer}
A=\frac{1}{4}\int_{S\times S}\! [a(s_1,s_2)-a(s_2,s_1)]^2\, \dd
\mu(s_1)\dd \mu(s_2) .
\end{equation}
The just given interpretation of the argument of integration proves
that this is basically a least square problem: we are trying to
minimize the departure from the Einstein's (isotropic)
synchronization.

Let $a'(s_1,s_2)=a(s_1,s_2)+c(s_2)-c(s_1)$ be the Fermat's map for
the generic synchronization obtained by moving forward the generic
clock $s$ of a quantity $c(s)$. We cannot expect to fix $c$
completely through the minimization problem. This is due to the
global time resettings which clearly do not change the Fermat map
and hence the anisotropy functional. We remove this freedom by
imposing $\int_S c(s) \dd \mu(s)=0$.

\begin{theorem} \label{qaw}
Under the assumptions of Theorem \ref{czo}, assume that $\delta_a$ is $\mu\otimes\mu$-measurable and that the anisotropy functional (\ref{wer}) of the initial synchronization is finite, namely that $\delta_a\in L^2(\mu\otimes\mu)$ (this holds if $S$ is compact and second countable). Among the resettings $c\in L^2(\mu)$ such that $\int_S c(s) \dd \mu(s)=0$, the anisotropy functional attains its minimum at a unique $c$ (as an element of $L^2(\mu)$), which is the synchronization investigated in this section (see Eq.
(\ref{vqh})) namely
\begin{equation} \label{jij}
c(s_1)=\frac{1}{2} \int_S [a(s_1,s)-a(s,s_1)] \dd \mu(s).
\end{equation}
Furthermore, the minimum $A'$ of the anisotropy functional is
\begin{equation} \label{n:amin}
A'= \frac{1}{12} \int_{S\times S\times S} w(s_1,s_2,s_3)^2 \, \dd
\mu(s_1) \dd \mu(s_2) \dd \mu(s_3) ,
\end{equation}
in particular $A'=0$ if and only if $w=0$, namely if and only if the optimal synchronization is Einstein's.
\end{theorem}

\begin{proof}
Since $\delta_a\in L^2(\mu\otimes\mu)$, by Fubini's theorem $\delta_a(s_1,\cdot)\in L^2(\mu)$ for almost every $s_1$, and hence, by Eq.\ (\ref{n:dbound}), for every $s_1$. Let $\hat c$ be the function defined by the right-hand side of Eq.\ (\ref{jij}). By the Cauchy-Schwarz inequality $\hat c(s_1)^2\le \frac{1}{4}\int_S \delta_a(s_1,s)^2\dd\mu(s)$, thus $\hat c\in L^2(\mu)$, and $\int_S\hat c\,\dd\mu=0$ by the antisymmetry of $\delta_a$. For $c\in L^2(\mu)$ with $\int_S c\, \dd\mu=0$ the Fermat's map $a'(s_1,s_2)=a(s_1,s_2)+c(s_2)-c(s_1)$ has antisymmetric part $\delta_a(s_1,s_2)+2[c(s_2)-c(s_1)]\in L^2(\mu\otimes\mu)$, and expanding the square
\begin{align*}
A'[c]&=A+\int_{S\times S}\!
\{[c(s_2)-c(s_1)]^2+\delta_a(s_1,s_2)[c(s_2)-c(s_1)] \}\, \dd
\mu(s_1)\dd \mu(s_2)\\
&=A+2\Vert c\Vert^2-4\langle c,\hat c\rangle=A-2\Vert\hat c\Vert^2+2\Vert c-\hat c\Vert^2 ,
\end{align*}
where $\langle\cdot,\cdot\rangle$ and $\Vert\cdot\Vert$ are the scalar product and the norm of $L^2(\mu)$, and we used $\int_S c\,\dd\mu=0$ and the antisymmetry of $\delta_a$. Thus $\hat c$ is the unique minimizer.

For the optimal synchronization $a'$ the antisymmetric part $\delta_{a'}$ satisfies, by Eq.\ (\ref{pyz}), $\int_S\delta_{a'}(s_1,s)\dd\mu(s)=0$ for every $s_1$, and hence also $\int_S\delta_{a'}(s,s_1)\dd\mu(s)=0$. Since $w$ is synchronization independent, $w(s_1,s_2,s_3)=\delta_{a'}(s_1,s_2)+\delta_{a'}(s_2,s_3)+\delta_{a'}(s_3,s_1)$. Squaring and integrating over $S\times S\times S$, each of the three squared terms gives $\int_{S\times S}\delta_{a'}^2\,\dd\mu\,\dd\mu=4A'$, since $\mu$ is unitary, while the mixed terms vanish: for instance, by Fubini's theorem,
\begin{align*}
\int_{S^3}&\delta_{a'}(s_1,s_2)\delta_{a'}(s_2,s_3)\,\dd\mu(s_1)\dd\mu(s_2)\dd\mu(s_3)\\
&=\int_{S^2}\delta_{a'}(s_1,s_2)\Big[\int_S\delta_{a'}(s_2,s_3)\dd\mu(s_3)\Big]\dd\mu(s_1)\dd\mu(s_2)=0 ,
\end{align*}
and similarly for the other two. Thus $\int w^2=12A'$. Finally, $w$ is continuous and $\mu$ is strictly positive, thus $\int w^2=0$ if and only if $w=0$, in which case $a'=r/2$ by Eq.\ (\ref{hch}). $\square$
\end{proof}

\begin{remark} \label{n:integr}
The integrability assumptions cannot be dropped from Theorems \ref{czo} and \ref{qaw}. Let $S=\mathbb{R}$, $\dd\mu(y)=\dd y/[\pi(1+y^2)]$ and $a(x,y)=\vert y-x\vert+\frac{1}{2}(y-x)$, which satisfies (i)-(iii). Here $w=0$, and the synchronization of Theorem \ref{czo} is Einstein's, $a'(x,y)=\vert y-x\vert$, which is obtained from Eq.\ (\ref{vqz}) with $c(x)=-x/2$ (up to a constant). However, $c$ is not $\mu$-integrable, so it cannot be normalized as in Eq.\ (\ref{pyw}), the integral in Eq.\ (\ref{vqh}) does not exist, and the anisotropy functional of the initial synchronization is infinite. Of course, Eqs.\ (\ref{hch}) and (\ref{vqz}) remain meaningful. If $\mu(S)=\infty$, as in the examples of the next subsections, Eq.\ (\ref{vqh}) requires a prescribed order of integration, see Theorems \ref{uiy} and \ref{n:circ}.
\end{remark}

\begin{remark}
The fact that our transitive synchronization is the one one would
obtain by trying to force Einstein's intransitive  synchronization
is important as it shows that it could have been already used,
without much awareness, in applications.
\end{remark}

\subsection{Locality, space topology and Alexander-Spanier
cohomology} \label{evt} \label{tre3}

We have already met expressions such as {\em coboundary} and {\em
cocycle}. This terminology can be made rigorous with the
introduction of the Alexander-Spanier cohomology theory
\cite{alexander35,spanier48}.

Let us define a $p$-cochain ($p$-function) as any continuous
function $f: S^{p+1}\to \mathbb{R}$.  The set of all $p$-cochains is
denoted $C^{p}(S)$ (which is the same as the set of continuous
functions in $p+1$ variables). The boundary operator $\p: C^{p}(S)
\to C^{p+1}(S)$ is defined by
\[
\p f(s_0,s_1,\ldots,s_{p+1})=\sum_{j=0}^{p+1}(-1)^j
f(s_0,\ldots,\hat{s}_j,\ldots s_{p+1}) ,
\]
where the hat means that the variable has been omitted. From this
definition it follows that $\p^2=0$. The set $K^p(S)$ is made of all
the $p$-{\em cocycles} namely all the $p$-cochains $f$ which satisfy
$\p f=0$. The set $B^p(S)$ is made of all the $p$-{\em coboundaries}
namely the $p$-cochains that can be expressed as $\p f$ for some
$f\in C^{p-1}(S)$. Since $\p^2=0$, $B^p(S)\subset K^p(S)$.

Coming to the synchronization problem it is convenient to rephrase
it in terms of the above mathematical setting. The triangle
inequality satisfied by the Fermat map $a$ can be expressed as
\[
\p a\ge 0 .
\]
Theorem \ref{czo} then states that given $a \in C^{1}(S)$, if
$a(s,s)\ge 0$ and $\p a\ge 0$ then up to a 1-coboundary $a\ge 0$
i.e. $a$ is a quasi-metric on $S/\!\!\sim$. The relation between the
quantity $w$ and $\delta$ given by Eq. (\ref{ujn}) can also be
written $\p \delta=w$ and Eq. (\ref{njw}) is nothing but $\p w=0$.

The mathematical theory we have just introduced, although useful to
reexpress some calculations that we have met, is not the {\em right}
Alexander-Spanier cohomology theory, as it is oversimplified. In
fact it is easy to check that $K^{p}(S)=B^{p}(S)$, that is if the
cohomology group is defined as $H^p(S)=K^p(S)/B^p(S)$ then all these
groups are trivial as they contain only one element. To show this
assume that $f\in K^{p}(S)$, that is,
\begin{align}
0&=(\p f)(s_0,s_1,\ldots,s_{p+1})=\sum_{j=0}^{p+1}(-1)^j
f(s_0,\ldots,\hat{s}_j,\ldots s_{p+1})\\
&=(-1)^{p+1}f(s_0,\ldots, s_{p})+ \sum_{j=0}^{p}(-1)^j
f(s_0,\ldots,\hat{s}_j,\ldots s_{p+1})
\end{align}
We can fix the value of $s_{p+1}$ and define $g(s_0,\ldots,
s_p)=(-1)^p f(s_0,\ldots, s_p,s_{p+1})$, so that the previous
equation reads $f=\p g$, namely $f\in B^p(S)$.

The interesting Alexander-Spanier cohomology arises if the
$p$-cochain is defined as above up to a function which vanishes for
close $p+1$ points. More precisely, a locally zero $p$-function for $p>0$,
$h(s_0,\ldots,s_{p})$, is a continuous function which vanishes on a neighborhood of the diagonal of $S^{p+1}$, equivalently, for which there is an open cover $\mathcal{U}$ of $S$ such that $h(s_0,\ldots,s_{p})=0$ whenever $s_0,\ldots,s_p$ belong to a common element of $\mathcal{U}$ \cite{spanier66,connes90}. Since on $S$ we have a natural metric $d_r$, one could be tempted to ask for some $\epsilon>0$ such that $h$ vanishes whenever $d_r(s_i,s_j)<\epsilon$ for every
$i,j=0,\ldots,p$. If $S$ is compact this uniform condition is equivalent to the previous one (take $\epsilon$ a Lebesgue number of the cover), but for non-compact $S$ it is more restrictive and we shall not use it. For $p=0$ a locally zero
$0$-function $h(s_0)$ is just the function  $h(s_0)=0$. The set of
locally zero $p$-functions is denoted $N^p(S)$. A cochain is a class
of elements in $C^p(S)$ which differ by a locally zero $p$-function.
The cocycles and coboundaries are defined similarly, every element
being a $p$-function  up to a locally zero function.

We use calligraphic notation for this non-trivial construction, thus
the set of $p$-cochains is $\mathcal{C}^p(S)$, that of cocycles
$\mathcal{K}^p(S)$, and that of coboundaries $\mathcal{B}^p(S)$. The
map $\pi: C^p(S)\to \mathcal{C}^p(S)$ sends the cochains defined at
the beginning  of this section to the Alexander-Spanier cochains and
commutes with $\p$. This projection will also be denoted $c\to
[c]_{AS}$. The cohomology group
$\mathcal{H}^p(S)=\mathcal{K}^p(S)/\mathcal{B}^p(S)$ is then
non-trivial and for manifolds coincides with that obtained using
singular cohomology.

We stress that the difference with the former trivial case stands on
the fact that the new definition of cochain disregards the old
cochain values for distant points. This does not mean that the new
cochain definition cannot bring global information as we shall see
in a moment.
One can correctly suspect
that by taking $p+1$ arbitrarily close points one is actually
probing the tangent space of $S$. This is true in those cases in
which a notion of tangent bundle exists for $S$, namely for
manifolds. In this case it can be shown that the Alexander-Spanier
cohomology coincides with the singular (de Rham) cohomology which
uses differentiability concepts such as 1-forms from the start.
Nevertheless, it must be remarked that the Alexander-Spanier
cohomology is more general and shares some better duality properties
than singular cohomology.

For manifolds one can consider the smooth $p$-cochains and develop
from them the Alexander-Spanier cohomology. There is then an
interesting canonical morphism which sends the Alexander-Spanier
$p$-cochain to the singular $p$-cochain which is an isomorphism in
cohomology \cite{connes90,gorokhovsky98}. Let $f$ be a smooth
Alexander-Spanier $p$-cochain, let $s_0\in S$ and let $v_1,\ldots
v_p\in T_{s_0}S$. Introduced curves $s_i(x)$, $i=1,\ldots, p$,
respectively of tangent vector $v_i$ at $s_0$, the differential form
$\nu(f)$ is defined by
\begin{equation} \label{gql}
\nu(f)(v_1,\ldots,v_p)=\frac{1}{p!}\sum_{\sigma}sgn(\sigma)
\p_{x_1}\ldots
\p_{x_p}f(s_0,s_{\sigma(1)}(x_1),\ldots,s_{\sigma(p)}(x_p))\vert_{x_j=0}
.
\end{equation}

The Alexander-Spanier cohomology helps to clarify the local aspects
of the synchronization problem. The starting element is the Fermat's
map $a$. In most physical cases the propagation of the signal is
deduced from the local propagation of the signal between neighboring
points. For this reason there is no physical loss of content in
passing from $a$ to its Alexander-Spanier cochain $[a]_{AS}$.

The situation is quite different for the Sagnac function $w$.
Consider for instance three points on the rim of a rotating disk in
special relativity, and let $S$ be this rim of $S^1$ topology and
let the signal be light propagating on a circular optical fiber laid
along the rim. When the three points are close $w$ vanishes
identically which means that $[w]_{AS}=0$. Nevertheless, $w$ does
not vanishes. If two of the three points over which $w$ is evaluated
are moved in opposite directions then there will be a critical
distance at which $w$ will begin to grow up to a maximum value (the
Sagnac effect). This behavior depends on the fact that the path
covered by the signal going along say, $s_1s_2s_3s_1$, does not need
to probe the whole rim according to the Fermat's definition of
arrival event. This happens only if the points are sufficiently
separated.

We have found in the previous section that we have a consistent
synchronization for every function $\delta \in C^1(S)$ such that $\p
\delta=w$. In the example of the rotating disk, we have $\p
[\delta]_{AS}= [w]_{AS}=0$ which means that $[\delta]_{AS}\in
\mathcal{K}^1(S)$ and hence, in this case, $[\delta]_{AS}$ defines a class in
$\mathcal{H}^1(S)$. This element retains the physical information on
the global Sagnac effect that is lost in $[w]_{AS}$. Indeed, since
the Alexander-Spanier cohomology groups coincide with those of the
singular cohomology we have, $\mathcal{H}^1(S)=\mathbb{R}$. This
real number associated to $[\delta]_{AS}$ is the global Sagnac
effect.

In short given $[w]_{AS}$ we cannot obtain a consistent
synchronization by selecting any $\delta$ such that $\p
[\delta]_{AS}=[w]_{AS}$. That synchronization would be transitive
locally but non-necessarily globally (the proof that the
synchronization would  be locally transitive is the same as that
given in Theorem \ref{gaw} but for sufficiently close points). The
choice of $\delta$ has to be done wisely because among all the
choices that satisfy $\p [\delta]_{AS}=[w]_{AS}$ there is one which
works also non-locally, i.e. retains the global information of our
physical problem. Thus by looking at local physical quantities such
as $[w]_{AS}$ we are left with a number of possibilities concerning
the spacetime. These possibilities form an affine space modelled on
$\mathcal{H}^1(S)$: the solutions $\eta\in \mathcal{C}^1(S)$ of $\p\eta=[w]_{AS}$ exist (take $\eta=[\delta]_{AS}$ with $\delta$ given by Eq.\ (\ref{nka})), if $\p [\delta_1]_{AS}=[w]_{AS}$ and $\p
[\delta_2]_{AS}=[w]_{AS}$ then the difference
$[\delta_2]_{AS}-[\delta_1]_{AS}$ is a cocycle, and hence defines an element of
$\mathcal{H}^1(S)$, and conversely adding a cocycle to a solution gives a solution, which is equivalent to the former modulo coboundaries if and only if the cocycle is a coboundary. Notice that if $[w]_{AS}\ne 0$ a single solution $[\delta]_{AS}$ is not a cocycle, and so it does not define an element of $\mathcal{H}^1(S)$ by itself.

As another example, consider for instance the stationary spacetime
of metric
\[
\dd s^2=-(\dd t+\alpha\dd \theta+\beta \dd \phi)^2+\dd \theta^2+\dd
\phi^2 ,
\]
 and toroidal space $S=S^1\times S^1$. Here $\theta$ and $\phi$ are the
angular coordinates of the two factors and $\alpha,\beta \in
\mathbb{R}$. The 1-forms $\dd \theta$ and $\dd \phi$ are closed but
non-exact on the two $S^1$ factors, in particular
$\mathcal{H}^1(S)=\mathbb{R}\times \mathbb{R}$. The local Sagnac
effect vanishes, $[w]_{AS}=0$, but only the right choice of
$[\delta]_{AS}$ allows to globally synchronize the clocks, in fact
different choices of $[\delta]_{AS}$ correspond to consistent
synchronizations but in different spacetimes parametrized by
$\alpha$ and $\beta$.

We learn from this analysis that the problem of synchronizing clocks
is essentially a non-local one. A synchronization method must give
$\delta$ once $w$ has been given. A local synchronization method
would give you $[\delta]_{AS}$ once $[w]_{AS}$ has been given. The
previous examples show that the latter  strategy cannot work for
global topological reasons because in $[w]_{AS}$ there is not enough
information to determine $[\delta]_{AS}$ (unless some topological
information is added). The synchronization method proposed in
Section \ref{mqz} has been shown to work in any circumstance and in
fact it is non-local because of the integral appearing in equation
(\ref{nka}) which uses the value of the Sagnac effect for distant
points.

It remains to mention what is the content of $[d_r]_{AS}$, the
symmetric part of $[a]_{AS}$. This is the distance for very close
points. If $S$ is a manifold endowed with a Riemannian metric $h$ such that $d_r(s,s')=\sqrt{h(v,v)}+o(\vert v\vert)$ for $s'=\exp_s v$, as it happens for the stationary spacetimes of Section \ref{hod} (see Eq.\ (\ref{fvt})), then $[d_r]_{AS}$ determines $h$ through $h(v,v)=\lim_{x\to 0} d_r(s,\gamma(x))^2/x^2$, where $\gamma$ is any $C^1$ curve with $\gamma(0)=s$ and $\dot\gamma(0)=v$. This recovery of the symmetric tensor $h$ from the near-diagonal values of a symmetric cochain should not be confused with the map $\nu$ of Eq.\ (\ref{gql}) below, which is alternating and sends cochains to differential forms.

Finally, one can even introduce an Alexander-Spanier 1-cochain (with
an additional dependence on $t_1,t_2$) which corresponds to the
spacetime metric as follows
\begin{equation} \label{nqz}
[g]_{AS}=-(t_2-t_1-[\delta]_{AS}/2)^2+([d_r]_{AS})^2.
\end{equation}
On the right-hand side it is understood that $[\delta]_{AS}$ has to
be replaced by any 1-cochain in its equivalence class, and the same
for $[d_r]_{AS}$. Taking the squares it is easy to check that this
procedure defines $[g]_{AS}$ up to an irrelevant locally zero
1-cochain. In the smooth stationary spacetimes of general relativity
$[g]_{AS}$ brings the same information of the usual spacetime metric
$\dd s^2$ and thus it gives all the information one may want on the
spacetime. In fact we note that in its expression there appears
$[\delta]_{AS}$ which, as we stressed,  carries some global
topological content  which cannot be found in  $[w]_{AS}$. In Eq.
(\ref{nqz}) after a resetting of the clocks, $t\to t-c(s)$,
$[\delta]_{AS}$ changes by a 1-coboundary term, which means that in
the definition of the metric what really matters is the class of $[\delta]_{AS}$ modulo $\mathcal{B}^1(S)$, namely an element of the affine space of the solutions of $\p\eta=[w]_{AS}$ modulo coboundaries, which is modelled on $\mathcal{H}^1(S)$, and which coincides with $\mathcal{H}^1(S)$ if $[w]_{AS}=0$.

\subsubsection{A simple $S^1$ example which calls for non-local
information} \label{tre4}

Let us consider $1+1$ Minkowski spacetime of coordinates $(t,x)$ and
metric $g=-\dd t^2+\dd x^2$. We make the identification $(t,0)
\leftrightarrow (t+\delta,b)$ where $0\le \delta<b$. We shall be
interested in the reference frame whose congruence is given by the
timelike geodesics $x=const$ of tangent vector $u=\p_t)_x$. The
space $S$ is a circle $S^1$ of length $b$.  Since $u$ is Killing and
normalized the condition ${\bf z=0}$ is satisfied. This example
illustrates clearly that in order to solve the synchronization
problem one needs non-local information. Locally the spacetime is
indistinguishable from Minkowski thus according to a local strategy
one would use Einstein's convention which however  reproduces the
slices $t=cnst$ which are not globally integrable  for $\delta\ne
0$.

\begin{figure}[t]
\begin{center}
 \includegraphics[width=7cm]{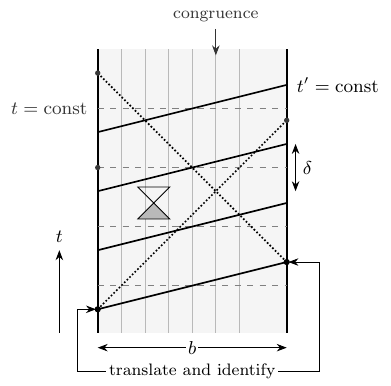}
\end{center}
\caption{A Minkowski 1+1 spacetime strip with side identified after
a vertical translation. The reference frame is given by the vertical
lines. The synchronization method leads to time $t'$ while the local
application of Poincar\'e-Einstein would lead to the horizontal
slices which however are not globally integrable for $\delta\ne 0$.
The slices $t'=\mathrm{const}$ are spacelike, as shown by the light cone, and the dotted lines are the two light rays emitted from the same event in opposite directions, which return to the emitting observer after the times $b-\delta$ and $b+\delta$. The example shows that the solution to the synchronization problem
must use non-local information.} \label{sone}
\end{figure}

Under the coordinate change (resynchronization)
$t'=t-\frac{\delta}{b} x$ the metric reads
\begin{equation} \label{ftf}
g=-\dd t'^2-2\frac{\delta}{b}\,\dd t' \dd x
+\frac{b^2-\delta^2}{b^2}\,\dd x^2
\end{equation}
where $t',x\in \mathbb{R}$ under the identification $(t',0)
\leftrightarrow (t',b)$. The speed of light in the new time
coordinate takes the value $1/(1-\delta/b)$ to the right and
$1/(1+\delta/b)$ to the left. In this synchronization, due to the
just given interpretation, is easy to calculate the function
$a'(s_1,s_2)$. Let $y$ be a space coordinate with the same
orientation of $x$ and origin in $s_1$. Let it denote the position
of $s_2$ with respect to $s_1$ then $a'$ is a periodic function with
period $b$ and such that
\begin{align*}
a'(y)&=\frac{b-\delta}{b} y \quad  \textrm{for} \   y\in
[0,(b+\delta)/2] , \\
a'(y)&=\frac{b^2-\delta^2}{2b}-\frac{b+\delta}{b}
(y-\frac{b+\delta}{2}) \quad \textrm{for} \ y\in [(b+\delta)/2,b].
\end{align*}
Now, letting $\dd \mu =\dd x$ we note that
\begin{align*}
\int_S a(s_1,s)\dd \mu(s)&=\int_0^b a'(y)\dd y=\frac{b^2-\delta^2}{4},\\
\int_S a(s,s_1)\dd \mu(s)&=\int_0^b a'(b-y)\dd
y=\frac{b^2-\delta^2}{4}.
\end{align*}
We conclude from Eq. (\ref{vqh}) that $c=0$ that is, the function
$t'$ is the result of our proposed synchronization method. This is
the simplest result one could guess from the geometry of the
problem.

In the next section we prove a general theorem which allows to prove
that   $t'$  is the result of our synchronization just by looking at
the metric (\ref{ftf}).

\subsection{Recovering the time of the non-rotating frame}
\label{vtx} \label{tre5} Let us return to the special case in which
the propagation of the signal is determined by a Lorentzian metric.

In this section it is proved that if the frame rotates with respect
to a frame adapted to a static axisymmetric spacetime then the time
one gets through the application of the synchronization method in
the rotating frame is that of the non-rotating frame.

As the proof clarifies, the next theorem holds also in the case in
which $\mu(S)$ is not finite under the stipulation that the
integrals which appear in the synchronization method are first taken
along the orbits of the spacelike $S^1$ Killing field (the integral along each orbit is finite, since the orbits are compact and the integrand is continuous, and it vanishes by the argument of the proof).

\begin{theorem} \label{uiy}
Let $(Q,\gamma)$ be a $(k+1)-$dimensional Riemannian manifold, $k\ge
0$, of coordinates $\{\phi, x^i, i=1,\ldots k\}$, $\phi\in[0,2\pi]$,
and metric  $\gamma$ of the form
\[
\gamma= f_{ij}({ x})\dd x^i\dd x^j +f_{\phi}({x}) \dd \phi^2
\]
where $\p_\phi$ is a Killing field which generates an $S^1$ isometry
group and $\phi \to -\phi$ is also an isometry. Let
$M=\mathbb{R}\times Q$ be a static spacetime endowed with a metric
of the form
\[
g=f_0(\phi,x)(-\dd t^2+\gamma).
\]
Let $\varphi=\phi-\omega t$, and consider the reference frame
determined by the flow $\p_t)_{x,\varphi}$. The space $S$ is the
quotient space restricted to $f_\phi<1/\omega^2$. The space $S$ is
naturally endowed with the optical metric
\[
h=\frac{ f_{ij}({\bf x})\dd x^i\dd x^j
}{1-f_\phi\omega^2}+\frac{f_\phi\dd \varphi^2}{(1-f_\phi\omega^2)^2}
,
\]
through which the  spacetime metric can be rewritten
\begin{equation} \label{met}
g=f_0 (1-f_\phi\omega^2)\{ -(\dd t-\frac{f_\phi
\omega}{1-f_\phi\omega^2}\, \dd \varphi)^2+h\}.
\end{equation}
Assume that the clocks on $S$ are syntonized so that ${\bf z=0}$
holds, i.e. in such a way that on their worldlines they measure $\dd
t$. Let $\mu$ be the volume element on $(S,h)$. The time $\tilde{t}$
obtained from Eq. (\ref{vqh}) (minimization of the global anisotropy
functional on $S$)  coincides with $t$.
\end{theorem}

\begin{proof}
We know that $\tilde{t}=t+c(s)$ where (see Eqs. (\ref{vqh}) and
(\ref{jij}))
\begin{equation}
c(s_1)=\frac{1}{2} \int_{S} [a(s_1,s)-a(s,s_1)] \dd \mu(s).
\end{equation}
We split this integral into two contributions, that given by the
portion of the integral over the set  $S^{+}=\{s\in S: \varphi_1 \le
\varphi(s)\le \varphi_1+\pi\}$ and that given by the integral over
$S^{-}=\{s\in S: \varphi_1-\pi \le \varphi(s)\le \varphi_1\}$, where
$S=S^+\cup S^-$ and $\mu(S^{+}\cap S^{-})=0$. Our idea is to prove
that the two contributions cancel out.

The trajectories of the light signals moving on $M$ are determined
by the metric $-\dd t^2+\gamma$ since the conformal rescaling
changes only the parametrization of the lightlike geodesics. This
metric tells us that the time it takes a signal to move from points
belonging to $Q$ equals their distance. Let $D(x_1,\phi_1,
x_2,\phi_2)$ be the Riemannian distance between $q_1=(x_1,\phi_1)$
and $q_2=(x_2,\phi_2)$. Because of the isometries
\begin{align}
D(x_1,\phi_1+y, x_2,\phi_2+y)&= D(x_1,\phi_1,
x_2,\phi_2), \qquad \textrm{ for every } x \in [0,2\pi),\\
D(x_1,-\phi_1, x_2,-\phi_2)&= D(x_1,\phi_1, x_2,\phi_2).
\end{align}
The time $a(s_1,s_2)$ it takes a light beam to leave
$s_1=(x_1,\varphi_1)$ and reach  $s_2=(x_2,\varphi_2)$ is the
smallest value which satisfies
\begin{equation} \label{hyt}
a(s_1,s_2)=D(x_1,\varphi_1,x_2,\varphi_2+\omega a(s_1,s_2))
\end{equation}
because the points $s_2$ displaces itself while the signal
propagates to reach it. Analogously, the time $a(s_2,s_1)$ it takes
a light beam to leave $s_2$ and reach $s_1$ is the smallest value
which satisfies
\begin{equation} \label{vfj}
a(s_2,s_1)=D(x_1,\varphi_1+\omega a(s_2,s_1),x_2,\varphi_2).
\end{equation}

By using the basic isometries we can construct the isometry
$p:\phi\to-(\phi-\phi_1)+\phi_1=2\phi_1-\phi$ which sends $S^{+}$ to
$S^-$ and conversely (and leaves invariant $s_1$). This isometry
leaves invariant the measure $\mu$, thus we have only to show that
if $s_3=p(s_2)$ then the contributions to $c(s_1)$ coming from $s_2$
and $s_3$ cancel out. This is done by proving the identities
\begin{align}
a(s_1,s_3)&=a(s_2,s_1), \label{kk1} \\
a(s_3,s_1)&=a(s_1,s_2).\label{kk2}
\end{align}
Concerning the former we have
$a(s_1,s_3)=D(x_1,\varphi_1,x_3,\varphi_3+\omega a(s_1,s_3))$ and
replacing $\varphi_3=2\varphi_1-\varphi_2$, $x_3=x_2$,
\begin{align*}
a(s_1,s_3)&=D(x_1,\varphi_1,x_2,2\varphi_1-\varphi_2+\omega
a(s_1,s_3))\\
&=D(x_1,\varphi_1,x_2,\varphi_2-\omega a(s_1,s_3))\\
&=D(x_1,\varphi_1+\omega a(s_1,s_3),x_2,\varphi_2)
\end{align*}
The quantity $a(s_1,s_3)$ is the smallest value which satisfies this
condition and since it coincides with the condition on $a(s_2,s_1$),
see Eq. (\ref{vfj}), we conclude that Eq. (\ref{kk1}) holds. As for
Eq. (\ref{kk2}) we have,
\begin{align*}
a(s_3,s_1)&=D(x_1,\varphi_1+\omega a(s_3,s_1),x_3,\varphi_3)\\
&=D(x_1,\varphi_1+\omega a(s_3,s_1),x_2,2
\varphi_1-\varphi_2)\\
&=D(x_1,\varphi_1-\omega a(s_3,s_1),x_2,\varphi_2)\\
&=D(x_1,\varphi_1,x_2,\varphi_2+\omega a(s_3,s_1))
\end{align*}
A comparison with Eq. (\ref{hyt}) proves that Eq. (\ref{kk2}) holds.
We conclude
\begin{align*}
c(s_1)&=\frac{1}{2} \int_{S^+}\{ [a(s_1,s)-a(s,s_1)]+
[a(s_1,p(s))-a(p(s),s_1)]\}\dd \mu(s)\\
&=\frac{1}{2} \int_{S^+}\{0\}\dd \mu(s)=0.
\end{align*}
 $\square$
\end{proof}

\begin{remark}
The angular velocity of the rotating reference frame can be
established by measuring the Sagnac holonomy (half the Sagnac
effect) over an $S^1$ orbit of the Killing field $\p_\varphi$ on
$S$. From Eq. (\ref{met}) it is not difficult to obtain the Sagnac
holonomy  in units of proper time
\[
\textrm{Sagnac holonomy}= \sqrt{\frac{f_0(x)}{1-f_\phi(x) \omega^2}}
\, 2\pi f_\phi(x) \omega .
\]
\end{remark}

\begin{remark}
In most interesting applications there is no need to resyntonize the
clocks because one is interested in the subset of $S$ on which ${\bf
z=0}$ holds. For instance in the case of a rotating planet we will
be interested in the equipotential surface  of sea level although
  the rotating frame extends outside it.
\end{remark}

The symmetry which makes the proof of Theorem \ref{uiy} work is the invariance of the problem under the reflection $\phi\to -\phi$ combined with the time reversal $t\to -t$. This observation leads to a considerable generalization, which includes the rotating frames of stationary, non-static, spacetimes such as Kerr spacetime (Section \ref{n:kerr}). We say that a spacetime $(M,g)$ is {\em circular stationary axisymmetric} if it admits two commuting Killing vector fields $\p_t$ and $\p_\phi$, the former timelike somewhere and complete, the latter generating an isometric $S^1$-action, and an isometry $P$ which maps $\p_t\to -\p_t$ and $\p_\phi\to -\p_\phi$. Typically, in adapted coordinates $(t,\phi,x^A)$, $P$ is the map $(t,\phi,x^A)\to (-t,-\phi,x^A)$, and the metric reads
\begin{equation} \label{n:papa}
g=g_{tt}\,\dd t^2+2g_{t\phi}\,\dd t\,\dd \phi+g_{\phi\phi}\,\dd \phi^2+g_{AB}\,\dd x^A \dd x^B ,
\end{equation}
where the coefficients depend only on the coordinates $x^A$. The existence of such coordinates, namely the orthogonality of the orbits of the group generated by $\p_t$ and $\p_\phi$ to a family of 2-surfaces, is the circularity property. It holds for vacuum and electrovacuum spacetimes \cite{carter69}, and more generally whenever the Ricci tensor satisfies a suitable orthogonality condition, which includes the case of perfect fluids in purely azimuthal motion, as shown by Kundt and Tr\"umper \cite{kundt66}. Static metrics of the form considered in Theorem \ref{uiy} are of this type, since in that case both $t\to -t$ and $\phi\to -\phi$ are isometries. Given $\Omega\in \mathbb{R}$, the Killing vector field $k_\Omega=\p_t+\Omega \p_\phi$ generates the frame of the observers which rotate rigidly with angular velocity $\Omega$. We denote by $M_\Omega$ the open region where $k_\Omega$ is timelike, which is invariant under the flow of $k_\Omega$ and under $P$, and by $S$ its quotient under the flow of $k_\Omega$. The function $t$ satisfies $\dd t(k_\Omega)=1$, namely it is a synchronization of the clocks of the frame (which measure the parameter of the flow of $k_\Omega$, see the condition ${\bf z=0}$). The Fermat's map is defined by the causal relation of the whole spacetime $M$, because the signals can leave the region $M_\Omega$.

\begin{theorem} \label{n:circ}
Let $(M,g)$ be circular stationary axisymmetric, let $\Omega\in\mathbb{R}$ and let $S$ be the space of the observers of the rigidly rotating frame $k_\Omega$. Let $\mu$ be the volume form of the optical metric $h$ of the frame. The time $\tilde{t}$ obtained from Eq.\ (\ref{vqh}) (minimization of the global anisotropy functional on $S$), with the integrals first taken along the orbits of $\p_\phi$ on $S$, coincides with $t$.
\end{theorem}

\begin{proof}
For $\psi_0\in \mathbb{R}$ let $R_{\psi_0}$ be the rotation of angle $\psi_0$, namely the flow of $\p_\phi$ at the parameter $\psi_0$. The isometry $P_{\psi_0}=R_{\psi_0}\circ P\circ R_{-\psi_0}$ maps $k_\Omega$ to $-k_\Omega$, thus it preserves $M_\Omega$ and the congruence of the frame, and it induces a map $p:S\to S$ which is an involutive isometry of the optical metric $h$, hence it preserves $\mu$, and which maps every orbit of $\p_\phi$ on $S$ into itself. Moreover, $t\circ P_{\psi_0}=-t$. Let $s_1\in S$. Since the orbits of $\p_\phi$ on $S$ are circles, and $p$ acts on them as a reflection, we can choose $\psi_0$ so that $p(s_1)=s_1$.

The map $P_{\psi_0}$ reverses the time orientation, thus it sends future directed causal curves into past directed causal curves. As a consequence, if a future directed causal curve connects an event of the fiber of $s$ to an event of the fiber of $s'$, so that its endpoints differ in $t$ by $\Delta t$, then its image, with reversed parametrization, is a future directed causal curve which connects the fiber of $p(s')$ to the fiber of $p(s)$, and whose endpoints differ in $t$ by the same amount $\Delta t$. By the definition of the Fermat's map in terms of the causal relation of $M$ (cf.\ Theorem \ref{gdw}(1)) we get $a(p(s'),p(s))=a(s,s')$ for every $s,s'\in S$. Taking $s=s_1$ we obtain $a(p(s'),s_1)=a(s_1,s')$, thus
\[
\int_S a(s,s_1)\, \dd \mu(s)=\int_S a(p(s'),s_1)\, \dd \mu(s')=\int_S a(s_1,s')\, \dd\mu(s'),
\]
where we used the invariance of $\mu$ under $p$. The same cancellation holds orbit by orbit, which is what matters if $\mu(S)=\infty$ or if $a(s_1,\cdot)$ is not $\mu$-integrable. Indeed, the measure induced by $\mu$ on an orbit $O$ of $\p_\phi$ on $S$ is invariant under the rotations $R_\psi$, hence it is a multiple of the uniform measure of the circle $O$, which is invariant under the reflection $p\vert_O$. Thus $\int_O [a(s_1,s)-a(s,s_1)]\,\dd\mu_O(s)=0$, the integral being finite because $O$ is compact and the integrand is continuous. By Eq.\ (\ref{vqh}), with the integrals first taken along the orbits, $c(s_1)=0$, and since $s_1$ is arbitrary, $\tilde{t}=t$. $\square$
\end{proof}

For a Kerr black hole and for the gravitational field of a rotating star the function $t$ is the time whose slices are orthogonal to the {\em zero angular momentum observers} (ZAMOs), whose angular velocity is $\omega=-g_{t\phi}/g_{\phi\phi}$ \cite{bardeen72}. Thus Theorem \ref{n:circ} states that the clocks of any rigidly rotating frame, synchronized with the method of this section, reproduce the simultaneity of the ZAMOs, a fact that does not depend on the angular velocity $\Omega$ of the frame. In Section \ref{n:kerr} we shall see that $t$ is also a harmonic synchronization in the sense of Section \ref{hod}, and that it is distinguished among the minimax synchronizations by a regional minimality property.

\subsection{Example: The rotating platform}\label{dja} \label{tre6}
Let us consider 3+1 Minkowski metric in polar coordinates
\[
g=-\dd t^2+\gamma ,\qquad \textrm{with } \gamma=\dd \rho^2+\rho^2
\dd \phi^2 +\dd z^2,
\]
and let us pass to the coordinates $(t,\rho,\varphi,z)$ of a
rotating frame by setting $\phi=\varphi+\omega t$. The rotating
frame is generated by the Killing vector $k=\p_t)_{\rho,\varphi,z}$
on the region where it is timelike, namely for
$\rho<1/\vert\omega\vert$. The spacetime metric takes the form
\[
g=(1-\omega^2\rho^2)[-(\dd t-\frac{\rho^2}{1-\rho^2\omega^2}\,
\omega \dd \varphi)^2+h],
\]
where the optical metric on $S=\{(\rho,\varphi,z): \rho<1/\vert\omega\vert\}$ is
\[
h=\frac{ \dd \rho^2 +\dd z^2}{1-\rho^2\omega^2}+\frac{\rho^2\dd
\varphi^2}{(1-\rho^2\omega^2)^2} ,
\]
and the Sagnac holonomy is $2 \pi \rho^2\omega /(1-\rho^2\omega^2)$.
Theorem \ref{uiy} proves that if (a) clocks at rest on the rotating
frame are syntonized so as to reproduce $\dd t$, namely their rate
is slowed of a position dependent factor $\sqrt{1-\omega^2\rho^2}$,
so that ${\bf z=0}$ holds, and (b) the new synchronization method is
applied then the time one gets is $t$, the time of the initial
inertial (non-rotating) frame.

\subsection[Example: Rotating pseudocylindric coordinates]{Example: Rotating pseudocylindric coordinate system in Minkowski
spacetime} \label{djb} \label{tre7}

For each Killing vector field in flat spacetime Letaw and Pfautsch
\cite{letaw82} have given an adapted coordinate system, in the sense
that $k=\p_{t}$. Any replacement $t \to t + c(x)$ preserves the
relation $k=\p_{t}$; Letaw and Pfautsch fix the gauge so as to
simplify the expression of the metric. Their choice of coordinate
time implies a synchronization (and syntonization) of the clocks at
rest in the frame, thus their convention arises from a criteria of
simplicity of the metric element. The most interesting Killing
vector field  that they consider leads to the {\em rotating
pseudocylindrical coordinates}. The  Minkowski metric in these
coordinates is
\begin{equation}
g=-(\xi^{2}-r^{2}\omega^{2})\dd t^2+\dd \xi^{2} + \dd r^{2}+
r^{2}\dd \varphi^{2}+2r^{2}\omega \dd \varphi \dd t,
\end{equation}
and $k=\p_{t})_{r,\xi,\varphi}$. We restrict ourselves to the open
set $S$ where $\xi^{2}>r^{2}\omega^{2}$, since there the Killing
vector field is timelike. By defining $\phi=\varphi+\omega t$, the
metric becomes
\begin{equation}
g=\xi^{2}[-\dd t^2+\gamma], \qquad \gamma=\frac{1}{\xi^2}(\dd
\xi^{2} + \dd r^{2}+ r^{2}\dd \phi^{2}).
\end{equation}
We are in the condition to apply Theorem \ref{uiy} to this case (it
is understood in the application of the synchronization method  the
integration with respect to $\varphi$ is performed first). As a
consequence of the theorem Letaw and Pfautsch's time, chosen by them
according to simplicity criteria, is indeed the time one would get
by applying the new synchronization method.

\subsection{Example: The rotating axisymmetric planet} \label{djc}
\label{tre8}

Let us consider clocks synchronized at a equipotential surface of a
rotating (fluid) planet by means of a web of optical cables laid
over the surface. The synchronization method used is the one which
minimizes  the global anisotropy functional, and we denote with
$\tilde{t}$ be the time obtained from that synchronization. In this
section we are going to prove that the ephemeris time $t$ coincides
with $\tilde{t}$. This result is important because it proves that a
network of communicating clocks could provide a good mean to
establish the `right' global time over the surface of a rotating
planet. All this through the application of an algorithm for
synchronization which does not use any model for the planet shape or
for the gravitational field of the body and does not need a
measurement of the angular velocity with respect to an exterior
frame (on the contrary, in the Earth case, the common view
synchronization method which takes advantage of the GPS has to use
this type of information).

For a weak gravitational field and for a source in slow motion the
spacetime metric in spherical coordinates is \cite[Eq.
18.15c]{misner73}
\[
g=-(1+2V)\dd t^2+(1-2V)\{\dd r^2+r^2(\dd \theta^2+\sin^2 \theta \dd
\phi^2)\},
\]
where $V$ is the gravitational potential satisfying the Poisson's
equation $\Delta V=4\pi \rho$ with $\rho=T_{00}$, and $t$ is
referred to as the ephemeris time. We consider an axially symmetric
solution, so that $\p_\phi$ is a Killing field and $\phi \to -\phi$
is an isometry. Although we shall not use this expression I mention
that the first perturbation with respect to a spherically symmetric
solution is given by
\[
V=-\frac{m}{r}\,[1+J_2(\frac{a_1}{r})P_2(\cos \theta)].
\]
For the Earth the quadrupole coefficient $J_2$ and the equatorial
radius $a_1$ are respectively $J_2=1.0826300  \times 10^{-3}$ and
$a_1=6.3781370 \times 10^6 m$ (see \cite{ashby03}).

Let us pass to a rotating coordinate system by setting
\begin{equation} \label{hyh}
\phi=\omega t+\varphi,
\end{equation}
so that the new reference frame is given by the congruence generated
by  the Killing field $\p_t)_{r,\theta,\varphi}$. The metric reads
\[
g=-(1+2\Phi)\dd t^2+(1-2V)\{\dd r^2+r^2(\dd \theta^2+\sin^2 \theta
\dd \varphi^2)\}+2(1-2V)\omega r^2 \sin^2\theta \dd \varphi \dd t,
\]
where
\[
\Phi=V-\frac{1}{2}\omega^2 r^2 \sin^2 \theta,
\]
is a new potential which includes a centrifugal term.

The condition $\Phi(\theta,r)=\Phi_0$ determines the shape of the
equipotential (sea) level in spherical coordinates. Over this
surface the condition ${\bf z=0}$ holds. For the Earth
\cite{ashby03} $\Phi_0=-6.96927 \times 10^{-10}$.

We are  interested on the propagation of light constrained to move
on the equipotential timelike hypersurface $S$ defined by
$\Phi(\theta,r)=\Phi_0$. In practice we can imagine that light moves
on optical cables laid over the equipotential slice. Alternatively,
one can consider light as moving freely between two shells
$\Phi=\Phi_0\pm \epsilon$, $\epsilon>0$, placed at negligible
distance, and reflecting on them. If so constrained light propagates
over the timelike hypersurface, of equation $\Phi(\theta,r)=\Phi_0$,
according to the induced connection which coincides with that
obtained from the induced metric \cite{hawking73}. As a consequence
the signal propagation is described by the metric
\begin{align*}
g_0&=-(1+2\Phi_0) \dd t^2+(1-2V(\theta))\,\{(\dd
r(\theta))^2+r(\theta)^2(\dd \theta^2+\sin^2 \theta \dd
\varphi^2)\}\\
&\quad +2(1-2V(\theta))\omega r(\theta)^2 \sin^2\theta \dd \varphi
\dd t,
\end{align*}
or
\[
g_0=(1+2\Phi_0)[-(\dd t-\frac{(1-2V(\theta))r(\theta)^2
\sin^2\theta}{1+2\Phi_0} \omega \dd \varphi )^2+h] ,
\]
where
\begin{align*}
h=\,&(1+2\Phi_0)^{-1}(1-2V(\theta))[(\dd
r(\theta))^2+r(\theta)^2(\dd \theta^2+\sin^2 \theta \dd \varphi^2)]
\\
&+(1+2\Phi_0)^{-2}(1-2V(\theta))^2r(\theta)^4 \sin^4\theta\, \omega^2 \dd \varphi^2
\end{align*}
and $r(\theta)$ comes from $\Phi(\theta,r)=\Phi_0$. The details on
the dependence $r(\theta)$ in the Earth case can be found in
\cite{ashby03}. Returning to the starting frame,
$\phi=\varphi+\omega t$, the metric reads
\begin{align} \label{njh}
g&=f_0[-\dd t^2+\gamma]
\end{align}
where
\begin{align*}
f_0&=1+2\Phi_0+(1+2\Phi_0-4V(\theta))\omega^2r^2(\theta)\sin^2\theta+\omega^4r^4(\theta)
\sin^4\theta,\\
 \gamma&=f_0^{-1} (1-2V(\theta))\{(\dd
r(\theta))^2+r(\theta)^2(\dd \theta^2+\sin^2 \theta \dd \phi^2)\}.
\end{align*}
The mathematical conditions that allow the application of Theorem
\ref{uiy} are satisfied thus we infer that $\tilde{t}=t$.

\begin{remark}
One could have applied directly Theorem \ref{uiy} to the starting
metric
\[ g=(1+2V)[-\dd t^2+\gamma], \qquad\gamma=\frac{1-2V}{1+2V}[\dd r^2+r^2(\dd
\theta^2+\sin^2 \theta \dd \phi^2)].
\]
However, from that one infers that if clocks are syntonized on the
rotating frame so that {\bf z=0} holds, and then synchronized
according to our method then what one gets is the time appearing in
the above metric, i.e. $t$. In this case, however, the light signal
through which clocks are synchronized is not constrained to move on
the equipotential surface but rather must be  free to move in any
direction. In practice this is not the case since the signal cannot
pass through the planet, a fact which is not mathematically taken
into account in metric.
\end{remark}

\begin{remark}
In the previous model the gravitomagnetic field generated by the rotation of the planet, responsible for the Lense-Thirring frame dragging, has been neglected. It can be included: the weak field metric of a slowly rotating axisymmetric body, which contains the additional term $g_{t\phi}\,\dd t\,\dd\phi$ due to its angular momentum, is circular stationary axisymmetric in the sense of Theorem \ref{n:circ}, and so is the metric induced on the timelike hypersurface on which the clocks of the rotating frame satisfy ${\bf z=0}$, since this hypersurface, defined by the condition $g(k_\Omega,k_\Omega)=const.$, is invariant under the isometry $P$. Thus Theorem \ref{n:circ} applies, and the synchronized clocks still reproduce the time $t$ of the non-rotating coordinates, whose slices are orthogonal to the zero angular momentum observers.
\end{remark}

\section{Synchronization over a graph: local anisotropy} \label{vtv}
\label{quat}

In Section \ref{mqz} we have obtained a synchronization which
minimizes the {\em global} anisotropy of the signal propagation. In
most cases the Fermat's map is deduced from the local propagation of
the signal and hence from the local value of the Fermat's map.
Indeed, we have argued that physically the triangle inequality for
the Fermat's map comes  from this interpretation.

A circumstance in which this situation is particularly clear is that
of a graph. Here the Fermat's map for distant vertices can be
deduced from the Fermat's map for directly linked vertices. We need
first some preliminary definitions.

A relation $R$, is a subset $R\subset S\times S$. Let
$R^{-1}=\{(x,y)\in S: (y,x)\in R\}$ and $\Delta=\{(x,x): x\in S\}$.
The relation is transitive if $R\circ R \subset R$, reflexive if
$\Delta \subset R$, symmetric if $R=R^{-1}$. The smallest transitive
relation containing $R$ is
\[
R^T=\bigcup_{i=1}R^i.
\]
A (simple) graph  is a relation $G\subset S\times S$ which is
reflexive, symmetric and connected, in the sense that $G^T=S\times
S$. The set $S$ is the set of vertices or nodes of the graph. We can
imagine to draw a segment for every unordered pair belonging to $G$
made of different points (thus there is at most one link for any
pair of vertices and there is no loop starting and ending at the
same vertex). The set of these pairs is denoted $E(G)$, and its
elements are called links. In principle we could generalize the next
results to the case in which $S$ is infinite, nevertheless in this
section we shall assume that $S$ is finite so that $G$ is indeed the
usual concept of graph. Note that the graph could connect all the
points, i.e. $G=S\times S$.

The set $N(s)=\{(y: (x,y)\in G, \ y\ne x\}$ is the set of vertices
directly linked to $s$. Sometimes we shall refer to it as the
neighborhood of $s$. It must be noted that $S$ is naturally endowed
with the discrete topology (so that the results of the previous
sections on the minimization of the global anisotropy functional
could be applied to this case) thus $N(s)$ is not a neighborhood of
$s$ in the topological sense as  $\{s\}$ is  a neighborhood for $s$.
 We have $G(s)=N(s)\cup\{s\}$.

The idea that the signal propagation is determined by the local
propagation between neighboring points is expressed mathematically
through the condition that the Fermat's map $a:S\times S\to
\mathbb{R}$ can be recovered from the restricted Fermat's map
$a\vert_{G}: G\to \mathbb{R}$ as follows
\[
a(s,s')=\textrm{min}\sum_{i=0}^{k-1} a\vert_G(s_i,s_{i+1}),
\]
where the minimum is taken over all chains of arbitrary length such
that $s_0=s$ and $s_k=s'$. The next synchronization method will
depend only on $a\vert_G$ although we shall denote it $a$  for
short.

 Let $m$ be a positive measure on $G$ which is symmetric,
i.e. $p^{*}m=m$, where $p:S\times S\to S\times S$ is the map
$p((x,y))=(y,x)$. Since to every element of $s'\in N(s)$ there is
$(s,s')\in G$, for fixed $s\in S$ the measure $m$ can be restricted
to $N(s)$ in the obvious way. We also define the measure $\mu$ on
$S$ through $\mu(s)=m(N(s_1))$. This definition implies $\mu
(S)=m(G)$.

We want now to establish a synchronization between the elements of
$S$ by minimizing the local anisotropy functional
\begin{equation} \label{wet}
A=\frac{1}{4}\int_{G}\! [a(s_1,s_2)-a(s_2,s_1)]^2\, \dd m.
\end{equation}

The algorithm for obtain the minimal synchronization depends on the
following

\begin{lemma} \label{hxd}
For $s_1\in S$ let
\begin{equation} \label{dgg}
k(s_1)=\frac{1}{2m(N(s_1))}\int_{N(s_1)}\! [a(s_1,s)-a(s,s_1)]\, \dd
m .
\end{equation}
By moving forward the clock $s_1$ of a quantity $k(s_1)$, without
resetting the other clocks, the anisotropy functional diminishes of
$2k^2m(N(s_1))$. After this operation $k(s_1)\to 0$ thus repeating
the operation on the node $s_1$ does not diminishes the anisotropy
functional (thus one has first to do the same operation on another
node and then return to $s_1$).
\end{lemma}

\begin{proof}
Let $a'(s_2,s_3)=a(s_2,s_3)+c(s_3)-c(s_2)$ where $c(s)$ vanishes for
$s\ne s_1$ and $c(s_1)=k(s_1)$. We have
\begin{align*}
A'&=A+\int_{G}\! \{[c(s')-c(s)]^2+[a(s,s')-a(s',s)][c(s')-c(s)] \}\,
\dd m(s,s')\\
&=A+2k^2 m(N(s_1))-2\int_{G}\! [a(s,s')-a(s',s)]c(s) \,
\dd m(s,s')\\
&=A-2k^2 m(N(s_1)).
\end{align*}
After resetting of $s_1$ the Fermat map passes from $a$ to $a'$ and
$k_1$ changes  to
\begin{align*}
k'(s_1)&=\frac{1}{2m(N(s_1))}\int_{N(s_1)}\! [a'(s_1,s)-a'(s,s_1)]\,
\dd m\\
&=\frac{1}{2m(N(s_1))}\int_{N(s_1)}\! [a(s_1,s)-a(s,s_1)-2k]\, \dd
m=0.
\end{align*}
 $\square$
\end{proof}

\begin{remark} \label{has}
Since the quantity $m(N(s))>0$ is bounded from below over the finite
set $S$, by doing this operation repeatedly passing from one node to
the next, the values of $k$ at the nodes being updated go to zero
(for otherwise $A$ would diminish by steps that do not go to zero which
is impossible since $A\ge 0$). This fact alone does not imply that $k(s)$ goes to zero at every node, a conclusion which requires the argument of Theorem \ref{cbv}.
\end{remark}

\begin{remark}
In the expression for $k(s_1)$, for $s\in N(s_1)$, the quantity
$[a(s_1,s)-a(s,s_1)]/2$ is how much you would move forward clock
$s_1$ to Einstein's synchronize it with $s$. The expression of
$k(s_1)$ shows that the resetting of the lemma is an average over
all neighboring points of these basic movements.
\end{remark}

The next result proves that for the anisotropy minimizing
synchronization the neighborhood average of the error with respect
to Einstein's synchronization vanishes.

\begin{theorem} \label{fgt}
The  synchronization that minimizes the local anisotropy functional,
if it exists (existence will be proved below, see Theorem
\ref{cbv}), has a Fermat's map
$a'(s_1,s_2)=a(s_1,s_2)+c(s_2)-c(s_1)$ which satisfies
\begin{equation} \label{dgj}
\frac{1}{2}\int_{N(s_1)}\! [a'(s_1,s)-a'(s,s_1)]\, \dd m=0 \qquad
\textrm{ for all }s_1\in S .
\end{equation}
Furthermore, there is at most one synchronization which satisfies
this equation.

\end{theorem}

\begin{proof}

The first claim is a consequence of Lemma \ref{hxd}. Let
$\tilde{a}(s_1,s_2)=a(s_1,s_2)+\tilde{c}(s_2)-\tilde{c}(s_1)$ be a
second synchronization that satisfies Eq. (\ref{dgj}), denoting
$h(s)=c(s)-\tilde{c}(s)$ and subtracting the equations we get
\[
\int_{N(s_1)}\! [h(s)-h(s_1)]\, \dd m=0,
\]
or $h(s_1)=\int_{N(s_1)}\! h(s)\dd m/m(N(s_1))$. This equation
proves that $h$ is at a point the average of the function on the
neighboring points. However, let $s_1$ be a point at which $h$ takes
the maximum value, then $h(s)=h(s_1)$ if $s\in N(s_1)$. Repeating
the argument for the elements of $N(s_1)$ and continuing in this way
we prove that $h$ is constant over $S$ and hence
$\tilde{a}(s_1,s_2)=a'(s_1,s_2)$.  $\square$

\end{proof}

\subsection{Conservation of total time under resynchronizations}
\label{alg} \label{quat2}

In order to prove the existence of the minimal synchronization, and
in order to provide an algorithm that leads to it, it is convenient
for each $s_1\in S$ to introduce the operation
\begin{itemize}
\item[(P):] clock $s_1$ is moved forward by a quantity $k(s_1)$ as  in Lemma \ref{hxd}, and after
this first operation all clocks (including $s_1$) are moved backward
of a quantity $k(s_1) \mu(s_1)/\mu(S)$.
\end{itemize}
Notice that the second step does not change the Fermat's map  since
all clocks are reset by the same amount. This operation has the
property that the space average of the forward movements and
backward movements taken with the opposite signs gives zero. This
operation preserves in a well defined sense the {\em total time}:
imagine to foliate the spacetime $M=S\times \mathbb{R}$ so that the
slices $S_t$ are transverse to the lines $s=const.$ and intersect
each of them only once. On each slice one can calculate a quantity
$T(t)$, called total time, as average (or integral) of the readings
of all the clocks at the events belonging to the slice. The
operation $(P)$, if performed in the interval $[t_0,t_1]$, does not
change the value of $T$ outside that interval, that is, it preserves
the total time. Moreover, the interval could be arbitrarily small
and it could be $t_1=t_0$ if the information needed for the
application of the operation (P) is collected in advance.

As the next theorem proves, in order to attain the optimal
synchronization of the nodes it is sufficient to apply (P)
repeatedly over the nodes.

\begin{theorem} \label{cbv}
The repeated application of operation (P) on all the nodes so that
(P) is applied on each node an infinite number of times moves
forward, in the limit, each clock $s$ by a finite amount $c(s)$,
such that $\int c(s)\dd \mu=0$, and makes the Fermat's map converge
to the value that minimizes the local anisotropy functional (in
particular the minimizing Fermat's map exists).
\end{theorem}

\begin{proof}
Since $S$ is finite, a resetting is a vector $c\in\mathbb{R}^S$, and the Fermat's map and $A$ depend only on its class in the quotient $V=\mathbb{R}^S/\mathbb{R}$ by the constant vectors, which is identified with the hyperplane $\int_S c\, \dd\mu=0$. Operation (P) preserves this normalization. On $V$ we introduce the scalar product
\[
\langle u,v\rangle_H=\int_G [u(s')-u(s)][v(s')-v(s)]\, \dd m(s,s'),
\]
which is positive definite because $G$ is connected: if $\langle u,u\rangle_H=0$ then $u$ takes the same value at neighboring points, hence it is constant. Denoting $A[c]$ the anisotropy functional (\ref{wet}) after the resetting $c$, $A[c]$ is a polynomial of second degree in $c$ whose quadratic part is $\Vert c\Vert^2_H$, thus, $V$ being finite dimensional, it attains its minimum at a unique $c^*\in V$, and the first variation at $c^*$ gives Eq.\ (\ref{dgj}), cf.\ Theorem \ref{fgt}. Expanding around $c^*$ and using Eq.\ (\ref{dgj}), the symmetry of $m$ and the antisymmetry of $a'-a'^T$, the linear term vanishes, thus
\[
A[c]=A[c^*]+\Vert c-c^*\Vert_H^2 .
\]
By Lemma \ref{hxd}, operation (P) on the node $s_1$ minimizes $A[c+t\, e_{s_1}]$ over $t\in\mathbb{R}$, where $e_{s_1}$ is the characteristic function of $\{s_1\}$, namely, setting $x=c-c^*$, it replaces $x$ by its $\langle\cdot,\cdot\rangle_H$-orthogonal projection $P_{s_1}x$ on the hyperplane $e_{s_1}^\perp$ of $V$. Thus, if $c_n$ is the resetting after the $n$-th operation, performed on node $s_{a_n}$, we have $c_{n}-c^*=P_{s_{a_n}}\cdots P_{s_{a_1}}(c_0-c^*)$. The vectors $e_s$, $s\in S$, span $V$, thus $\bigcap_{s\in S} e_s^\perp=\{0\}$. By the theorem of Amemiya and Ando \cite{amemiya65}, a product of finitely many orthogonal projections in which each projection appears infinitely often converges weakly, hence in our finite dimensional space in norm, to the projection on the intersection of their ranges. Thus $c_n\to c^*$, which proves the theorem.

If the update gaps are bounded, namely each node is updated at least once in every $N$ consecutive operations, as it happens for a cyclic schedule, the following elementary argument suffices. The sequence $c_n$ is bounded since $A[c_n]\le A[c_0]$, and since $\sum_n \Vert c_{n+1}-c_n\Vert_H^2=A[c_0]-\lim_n A[c_n]<\infty$, we have $\Vert c_{n+1}-c_n\Vert_H\to 0$. Let $c_{n_r}\to \bar c$. For every $s\in S$ there is $m_r\in [n_r,n_r+N]$ such that $s$ has just been updated in $c_{m_r}$, so that the quantity $k(s)$ of Eq.\ (\ref{dgg}) computed for $c_{m_r}$ vanishes. Since $\Vert c_{m_r}-c_{n_r}\Vert_H\le N \max_{n\ge n_r}\Vert c_{n+1}-c_n\Vert_H\to 0$, we have $c_{m_r}\to \bar c$, and by continuity $k(s)=0$ at $\bar c$ for every $s$, thus $\bar c=c^*$ by Theorem \ref{fgt}. Every convergent subsequence has the same limit, thus $c_n\to c^*$. $\square$

\end{proof}

\subsection{The optimal synchronization is the only  Nash
equilibrium} \label{opt} \label{quat3}

Let us introduce the {\em anisotropy at point} $s_1\in S$ through
the expression
\[
\mathcal{A}(s_1)=\frac{1}{4 m(N(s_1))}\int_{N(s_1)}\!
[a(s_1,s)-a(s,s_1)]^2\, \dd m .
\]
It represents the average (with respect to the neighboring points)
quadratic error with respect to the Einstein synchronization. Let us
assume that the observers placed at the vertices of the graph try to
minimize their own anisotropy. If the observer at, say $s_1$,  had the
freedom to adjust the clocks of his neighborhood he would be able to
get zero as minimum. Nevertheless, we assume that the observer at
$s_i$ can only adjust his own clock trying to minimize
$\mathcal{A}(s_i)$. This is a game for which the optimal
synchronization whose existence and unicity has been established in
the previous section, represents a point of Nash equilibrium.

Let us recall that a solution for a game is a Nash equilibrium if no
player has anything to gain by changing only his or her own strategy
unilaterally. In our game we have to show that if one
decides to move his own clock  while there is no resetting on the
other nodes, then what he gets is that his own anisotropy does not
decrease (actually I will show that it increases, thus we are in
presence of a strict Nash equilibrium).

First note that
\[
A=\int_S \mathcal{A}(s) \, \dd \mu(s).
\]
Let $\bar{s}$ be the clock which we are going to move by a quantity $c$. We can write
\[
A=\mathcal{A}(\bar{s}) \mu(\bar{s}) +\int_{S\backslash\{\bar{s}\}}
\mathcal{A}(s) \, \dd \mu(s).
\]
The quantity $\mathcal{A}(\bar{s}) \mu(\bar{s})$ is actually a sum
over the links connecting $\bar{s}$ to the points in its
neighborhood. An equivalent sum is contained in
$\int_{S\backslash\{\bar{s}\}} \mathcal{A}(s) \, \dd \mu(s)$ because
each link is counted twice in the expression of the anisotropy
functional Eq. (\ref{wet}), one for each direction. As a consequence
\[
A'-A=\Delta A=2 \Delta \mathcal{A}(\bar{s}) \mu(\bar{s})
\]
where for $s\in N(\bar{s})$ $a'(\bar{s},s)=a(\bar{s},s)-c$ and
$a'(a,\bar{s})=a(s,\bar{s})+c$. However, $\Delta A>0$
because the starting point was the minimizing synchronization (which is unique). We
infer $\Delta \mathcal{A}(\bar{s})>0$ as claimed.

Furthermore, let us show that the minimizing synchronization is the unique Nash equilibrium.
Let $c$ be how much we move forward clock $\bar{s}$. If we are in a Nash equilibrium $\Delta \mathcal{A}(\bar{s})\ge 0$.
Since for $s\in N(\bar{s})$, $a'(\bar{s},s)=a(\bar{s},s)-c$ and
$a'(s,\bar{s})=a(s,\bar{s})+c$ we have $\mathcal{A}'(\bar{s})=\mathcal{A}(\bar{s})+c^2-2c k(\bar{s})$,
thus for some choice of $c$, $\Delta \mathcal{A}(\bar{s})< 0$ unless $k(\bar{s})=0$.
Thus in a Nash equilibrium it must be $k(s)=0$ for all $s\in S$, a condition which by Theorem \ref{fgt} characterizes  the
minimizing synchronization.

\subsection{Minimax synchronizations of a network} \label{n:graphmm}

The anisotropy functional (\ref{wet}) is a quadratic average over the links. In many applications one is interested in controlling the worst link, namely the maximal error with respect to Einstein's synchronization. In this subsection we show that the corresponding minimax problem has a simple and computable solution, which is the discrete counterpart of the synchronization theorem of Section \ref{n:minimax}.

We assume that the two-way time $r(s,s')=a(s,s')+a(s',s)$ is positive for every link $(s,s')\in G$, $s\ne s'$ (otherwise the synchronization classes of Lemma \ref{hwx} would be non-trivial), and we define the {\em anisotropy of the link}
\[
\varepsilon(s,s')=\frac{\vert a(s,s')-a(s',s)\vert}{a(s,s')+a(s',s)} .
\]
It depends on the synchronization, and it is the relative error that one makes by synchronizing $s'$ with $s$ through Einstein's convention: $s'$ has to be moved backward by $\frac{1}{2}[a(s,s')-a(s',s)]$, that is, by the fraction $\varepsilon(s,s')$ of half the round-trip time, $r(s,s')/2$. In the manifold case of Section \ref{hod}, Eq.\ (\ref{fvt}) shows that for neighboring points $s'=\exp_s v$ we have $\varepsilon(s,s')\simeq \vert b(v)\vert/\sqrt{h(v,v)}$, whose maximum over the directions is the norm of $b$, namely the anisotropy of the one-way speed of light at $s$.

A {\em cycle} is a closed chain $s_0s_1\cdots s_k$, $s_k=s_0$, with $(s_i,s_{i+1})\in G$. Let $T^+=\sum_i a(s_i,s_{i+1})$ and $T^-=\sum_i a(s_{i+1},s_i)$ be the times taken by the signal to go around the cycle in the two directions. They are synchronization independent, and $T^+-T^-$ is the Sagnac effect of the cycle. We define the {\em Sagnac ratio} of the network
\begin{equation} \label{n:betaG}
\beta_G=\max_{\textrm{cycles}} \frac{T^+-T^-}{T^++T^-} \in [0,1] ,
\end{equation}
where the maximum exists because it can be taken over the cycles which do not pass twice through the same vertex, which are finite in number (a cycle splits into such simple cycles, and the ratio of sums of positive denominators is bounded by the maximal ratio). It is non-negative because inverting the direction changes its sign, and the upper bound follows from $T^->0$ (or $T^-=0$ in a causality violating network, see below). Note that $\beta_G$ can be measured with a single clock placed on the cycle.

\begin{theorem} \label{n:graphthm}
Let $(S,G)$ be a finite connected graph with Fermat's map satisfying (i)-(ii) and $r>0$ on the links. Then
\[
\min_{c} \ \max_{(s,s')\in G,\, s\ne s'}\varepsilon_c(s,s')=\beta_G,
\]
where $\varepsilon_c$ is the anisotropy of the links after the resetting $a\to a+c(s')-c(s)$. The minimizers are the functions $c$ which satisfy
\begin{equation} \label{n:diffc}
2[c(s')-c(s)]\le \beta_G\, r(s,s')-[a(s,s')-a(s',s)] \qquad \textrm{for every } (s,s')\in G,
\end{equation}
and for every $s_0\in S$ a minimizer is given by $c(s)=D(s_0,s)$, where $D$ is the distance from $s_0$ in the directed graph in which the link from $s$ to $s'$ has the length $w(s,s')=\frac{1}{2}\{\beta_G\, r(s,s')-[a(s,s')-a(s',s)]\}$. If $\beta_G<1$, then for every minimizer the new Fermat's map satisfies $a'(s,s')\ge \frac{1}{2}(1-\beta_G)\,r(s,s')>0$ on every link, so the synchronized time is a continuous utility (Theorem \ref{uhk}). Finally, $\beta_G=1$ if and only if there is a cycle which the signal can cover in a vanishing time in one of the two directions.
\end{theorem}

\begin{proof}
Let $\delta(s,s')=a(s,s')-a(s',s)$, so that $\varepsilon(s,s')=\vert\delta(s,s')\vert/r(s,s')$ and a resetting changes $\delta(s,s')$ into $\delta(s,s')+2[c(s')-c(s)]$. Along a cycle the sum of the differences $c(s_{i+1})-c(s_i)$ vanishes, thus for every $c$ and every cycle
\[
T^+-T^-=\sum_i \delta_c(s_i,s_{i+1})\le \max_G \varepsilon_c \sum_i r(s_i,s_{i+1})=\max_G \varepsilon_c\,(T^++T^-),
\]
which proves that the maximal anisotropy is not smaller than $\beta_G$. Conversely, the length $w$ has a non-negative sum over every cycle, because $\sum_i w(s_i,s_{i+1})=\frac{1}{2}[\beta_G(T^++T^-)-(T^+-T^-)]\ge 0$ by the definition of $\beta_G$. Thus the directed graph has no negative cycles, $D(s_0,\cdot)$ is finite, and it satisfies the triangle inequality $D(s_0,s')\le D(s_0,s)+w(s,s')$, which is (\ref{n:diffc}) for $c=D(s_0,\cdot)$. Writing (\ref{n:diffc}) for $(s,s')$ and for $(s',s)$, and using the antisymmetry of $\delta$ and the symmetry of $r$, we get $-\beta_G r\le \delta+2[c(s')-c(s)]\le \beta_G r$, namely $\varepsilon_c\le \beta_G$ on every link. The same argument shows that every solution of (\ref{n:diffc}) is a minimizer, and conversely a minimizer satisfies $\vert\delta_c\vert\le \beta_G r$, which contains (\ref{n:diffc}).

If $\varepsilon_c\le\beta_G$ then $a'(s,s')=\frac{1}{2}[r(s,s')+\delta_c(s,s')]\ge \frac{1}{2}(1-\beta_G)r(s,s')$ on the links, and since the Fermat's map is determined by its restriction to the links, $a'(s,s')>0$ for every $s\ne s'$, which by (ii) and Theorem \ref{uhk} implies that the synchronized time is a continuous utility. Finally, $T^+-T^-=T^++T^-$ if and only if $T^-=0$, while $T^-<0$ is excluded by (ii) because the Fermat's map is determined by its restriction to the links. $\square$
\end{proof}

\begin{remark}
The number $\beta_G$ is a maximum cycle ratio. It can be computed efficiently, e.g.\ with Lawler's parametric method \cite{lawler76}, which combines a bisection on the value of $\beta$ with the detection of negative cycles, or with the closely related algorithms for the minimum cycle mean \cite{karp78}. The optimal resetting is then obtained by the classical reduction of systems of difference constraints to shortest path problems \cite[Sect.\ 24.4]{cormen09}, and it can be computed in a distributed way with the Bellman-Ford iteration used by the distance-vector routing protocols. The minimizers form a convex polytope; among them one can select, for instance, the one which minimizes the quadratic functional (\ref{wet}). In the language of the max-plus algebra, whose relevance for the synchronization of discrete event systems is well known \cite{baccelli92}, $\beta_G$ plays the role of an eigenvalue and the optimal resettings that of the eigenvectors. These objects are the discrete counterparts of Ma\~n\'e's critical value and of the weak KAM solutions which enter the synchronization theorem of Section \ref{n:minimax}, and the condition $\beta_G<1$ is the discrete counterpart of global hyperbolicity.
\end{remark}

\begin{remark} \label{n:strict}
The shortest path solution $c=D(s_0,\cdot)$ is an extremal minimizer: every clock $s'\ne s_0$ is joined to its predecessor $s$ in a tree of shortest paths by a link for which (\ref{n:diffc}) holds with the equality sign, namely $\varepsilon_c(s,s')=\beta_G$, so that every clock has a link of maximal anisotropy. This is the discrete counterpart of the constant anisotropy of the weak KAM solutions of Theorem \ref{n:sync}. A canonical minimizer which, on the contrary, lowers the anisotropy wherever possible is obtained by minimizing lexicographically the vector of the link anisotropies arranged in decreasing order: one first minimizes the maximal anisotropy, then, keeping fixed the anisotropy of the links on which it equals $\beta_G$ for every minimizer, the maximal anisotropy of the remaining links, and so on. This {\em strict minimax} resetting is unique up to an additive constant, it can be computed with a finite sequence of linear programs, and it synchronizes with Einstein's convention the two ends of every link which does not belong to a cycle. It is the limit, as $p\to\infty$, of the minimizers of the $\ell^p$ functionals $\sum_{G}\varepsilon_c^p$, by Descloux's theorem on the strict Chebyshev approximation \cite{descloux63}, while for $p=2$ these minimizers are those of the quadratic functional (\ref{wet}) for the measure $m(s,s')=r(s,s')^{-2}$, namely the Nash equilibrium of Subsection \ref{opt} for this measure. The continuous counterpart of these statements is discussed in Subsection \ref{n:glob}.
\end{remark}

\subsection{A numerical example} \label{n:numex}

In order to illustrate Theorem \ref{n:graphthm}, and to compare the quadratic and the minimax resettings, we consider a schematic model of a global fiber network with 28 nodes and 43 links, hence with 16 independent cycles, inspired by the main terrestrial and submarine routes (Fig.~\ref{n:globe} and Appendix \ref{n:app}).

\begin{figure}[ht]
\begin{center}
 \includegraphics[width=\textwidth]{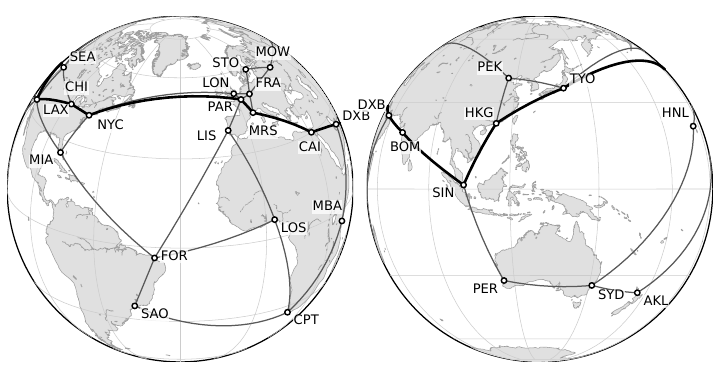}
\end{center}
\caption{The network of Subsection \ref{n:numex} in two orthographic views of the Earth, centered at $30^\circ$W, $20^\circ$N and at $130^\circ$E, $0^\circ$. The nodes are labelled with the abbreviations listed in Appendix \ref{n:app}, Dubai appearing in both views, and the links are drawn along great circles. The thick links form the critical cycle of scenario A, which realizes the Sagnac ratio $\beta_G$.} \label{n:globe}
\end{figure}
 The clocks are at rest on the rotating Earth, whose angular velocity is $\Omega=7.292\times 10^{-5}\,\textrm{s}^{-1}$, and the time $t$ of the non-rotating geocentric frame plays the role of the true time. Each link follows the great circle route between its ends, with a length $L$ increased by a factor $1.3$ to account for the actual path of the cable, and the signals propagate in a fiber of group index $n=1.468$. With respect to $t$ the Fermat's map of a link is $a(s,s')=\frac{1}{2}[r(s,s')+\delta(s,s')]$, where $r=2nL/c$ is the round-trip time and
\[
\delta(s,s')=\frac{4\Omega A_z(s,s')}{c^2}+\eta(s,s') .
\]
Here $A_z(s,s')$ is the area swept by the projection of the radius vector on the equatorial plane along the route from $s$ to $s'$, so that the first term is the Sagnac term \cite{ashby03}, which reaches $66\,$ns for the longest links, while $\eta=-\eta^T$ is a residual asymmetry of the link, due for instance to uncalibrated differences between the fibers used in the two directions. We consider two scenarios: (A) $\eta=0$; (B) the values of $\eta$ on the links are independent and normally distributed with standard deviation $10\,$ns. The round-trip times range from $4.4\,$ms (London-Paris) to $134\,$ms.

We compare four resettings: Einstein's convention along the spanning tree of the fastest routes from Frankfurt, which is the analog of a hierarchical protocol such as PTP \cite{ieee1588}\footnote{In PTP the hierarchy is built by the best master clock algorithm: the clocks compare the data sets advertised by the candidate masters (priorities set by the operator, clock class, accuracy, variance, identity), and among otherwise equivalent masters they choose the one with fewer intermediate devices \cite{ieee1588}. In NTP each client selects, among several servers, the one with the smallest root distance, namely half the round-trip delay to the reference clock plus the accumulated dispersion, and combines the offsets of the best servers with weights inversely proportional to their root distances \cite{rfc5905}. In both cases the hierarchy is chosen according to the quality of the clocks and to the delays, not to the geometry of the routes, and it changes whenever a better master appears or a link fails. Each link is synchronized with Einstein's convention. A known delay asymmetry can be corrected link by link, but it must be measured or computed, e.g.\ the Sagnac term from the geometry of the route, while our methods do not need this information. The spanning tree of the fastest routes is a stylized version of these hierarchies.}; the quadratic resetting of Subsection \ref{opt} for the measure $m=r^{-2}$, which minimizes $\sum_G \varepsilon^2$; the minimax resetting $c=D(s_0,\cdot)$ of Theorem \ref{n:graphthm} with $s_0$ Frankfurt; and the strict minimax resetting of Remark \ref{n:strict}. All of them are covariant under resettings, thus we can start from the true time, and then $c(s)$ is the error of the clock $s$, up to an additive constant. Table \ref{n:tab} reports, for each resetting, the maximal and the root mean square anisotropy of the links, the largest error between two clocks and the root mean square deviation of $c$ from its mean $\bar c$.

\begin{table}[!ht]
\caption{Resettings of the clocks of the network of Subsection \ref{n:numex}. The anisotropies of the links are in units of $10^{-7}$ and the clock errors, with respect to the time $t$ of the non-rotating frame, in nanoseconds.} \label{n:tab}
\begin{center}
\begin{tabular}{lcccc}
\hline
resetting & $\max_G \varepsilon$ & rms$_G\, \varepsilon$ & $\max\vert c(s)-c(s')\vert$ & rms$(c-\bar c)$ \\
\hline
\multicolumn{5}{l}{Scenario A: Sagnac effect only, $\beta_G=5.60\times 10^{-7}$}\\
\hline
true time $t$ & 7.87 & 4.96 & 0.0 & 0.0 \\
Einstein (tree) & 103.05 & 16.95 & 132.5 & 29.2 \\
quadratic & 9.92 & 4.11 & 23.1 & 5.8 \\
minimax (shortest path) & 5.60 & 4.96 & 19.7 & 3.6 \\
strict minimax & 5.60 & 4.91 & 18.5 & 3.5 \\
\hline
\multicolumn{5}{l}{Scenario B: Sagnac effect and asymmetries, $\beta_G=6.34\times 10^{-7}$}\\
\hline
true time $t$ & 16.84 & 6.23 & 0.0 & 0.0 \\
Einstein (tree) & 117.27 & 19.36 & 153.7 & 32.5 \\
quadratic & 9.70 & 4.58 & 25.8 & 6.3 \\
minimax (shortest path) & 6.34 & 5.84 & 26.7 & 6.8 \\
strict minimax & 6.34 & 5.10 & 14.4 & 3.0 \\
\hline
\end{tabular}
\end{center}
\end{table}

In scenario A the Sagnac ratio is $\beta_G=5.60\times 10^{-7}$. The critical cycle goes around the Earth through New York, Paris, Marseille, Cairo, Dubai, Mumbai, Singapore, Hong Kong, Tokyo, Los Angeles and Chicago, and its Sagnac effect is $T^+-T^-=255\,$ns over a travel time around the cycle of $227\,$ms. For comparison, a route along the parallel of latitude $\lambda$ would give the ratio $\Omega R\cos\lambda/(1.3\,nc)$, $R$ being the radius of the Earth, which is $7.0\times 10^{-7}$ for $\lambda=30^\circ$. Einstein's convention along the tree is exact on its 27 links, but the Sagnac terms accumulate on the other 16 links, whose anisotropy reaches $10^{-5}$, and clocks on opposite sides of the Earth disagree by $132\,$ns. The quadratic and the minimax resettings reduce this error by a factor between $5.7$ and $7.2$. The quadratic resetting has the smallest mean square anisotropy, but its maximal anisotropy is $1.8$ times $\beta_G$, while the minimax resettings attain the bound $\beta_G$ of Theorem \ref{n:graphthm} with a slightly larger mean anisotropy. None of the resettings recovers $t$ exactly, and $t$ itself is not optimal for either criterion, its maximal anisotropy being $7.9\times 10^{-7}$: the network has no symmetry, and an internal method cannot distinguish the part of the Sagnac terms of the form $2[c(s')-c(s)]$ from an offset of the clocks. In the continuous model of Subsection \ref{djc}, on the contrary, the time of the non-rotating frame is recovered exactly.

In scenario B the asymmetries raise the Sagnac ratio to $\beta_G=6.34\times 10^{-7}$ and change the critical cycle, which now crosses the Pacific Ocean through Honolulu and Sydney and the Indian Ocean through Perth. The strict minimax resetting gives the smallest clock errors, about one half of those of the quadratic resetting, while the shortest path solution, which saturates the bound $\beta_G$ on the links of a whole spanning tree, is less accurate. Finally, the $\ell^p$ minimizers converge to the strict minimax resetting as expected: for $p=2,4,8,16,32,64,128$ the maximal distance between the two resettings in scenario A is $13.7$, $5.1$, $1.9$, $0.82$, $0.40$, $0.24$, $0.13\,$ns, roughly proportional to $1/p$. The strict minimax resetting is obtained in 16 stages, as many as the independent cycles of the network.

\FloatBarrier
\section{Synchronization over a Riemannian manifold} \label{hod}
\label{cinq}

In this section we continue the study of non-totally vicious
Lorentzian spacetimes $(M,g)$ admitting a complete timelike conformal
Killing field $k$. We have already proved in Theorem \ref{ggg}
that they are exactly the spacetimes that are diffeomorphic
to $\mathbb{R}\times S$ with metric
\begin{equation} \label{jup}
g=\chi[-(\dd t-\pi^{*}b)^2 +\pi^*h],
\end{equation}
where $k=\p_t$ and in which the Fermat's map satisfies conditions (i)-(iii).

\begin{proposition}\label{ouv}
The following statements concerning the splitting of the spacetime  $(M,g)$ are equivalent
\begin{itemize}
\item[(a)] function $t:M\to \mathbb{R}$ is a smooth time function with timelike gradient (necessarily past directed),
\item[(b)] one and hence every slice, $t=cnst$, is a spacelike hypersurface,
\item[(c)] $(S, h-b \otimes b)$ is a Riemannian manifold,
\item[(d)] $h^{-1}(b,b)<1$ all over $S$,
\item[(e)] The Finsler metric of Randers type
$F(s,y)=b(y)+\sqrt{h(y,y)}$, satisfies $F(s,y)> 0$ for $y\ne  0$,
\item[(f)] The pair $(S,F)$, with $F(s,y)$ as above 
is a Finsler manifold, namely $F:TS\to \mathbb{R}$ is strongly
convex: $\frac{1}{2} \p_{y^i}\p_{y^j}F^2$ is positive definite for
$y\ne 0$.
\end{itemize}
In particular, in this case for all $s\in S$, $a(s,s)\ge 0$.
\end{proposition}

\begin{proof}
(a) $\Rightarrow$ (b). If the gradient of $t$ is timelike then it is
past directed because $g(k,\nabla t)=\dd t[k]=1>0$. Since $t$ is
smooth the hypersurface $t=cnst$ is smooth and if $v$ is  tangent to
the hypersurface it is spacelike as $0=\dd t[v]=g(v,\nabla t)$.

(a) $\Leftarrow$ (b). The way we have defined $t$ from the section
$t=0$ (Sect. \ref{hin}) implies  that $\nabla t$ exists at $x\in M$
and is the only vector perpendicular to the tangent space of the
slice passing through $x$ and such that $g(k,\nabla t)=1$, in
particular it is timelike.

(b) $\Leftrightarrow$ (c). The metric $h-b \otimes b$ is that
induced by $(M,\hat{g})$ on the slice $t=cnst$ which is
diffeomorphic to $S$. Thus the slice is spacelike if and only if the
metric $h-b \otimes b$ is positive definite \cite{hawking73}.

(c) $\Rightarrow$ (d). Let $v=h^{-1}(\cdot,b)$, since $h-b\otimes b$
is positive definite, $0\le
h(v,v)-b(v)^2=h^{-1}(b,b)[1-h^{-1}(b,b)]$ where the equality holds
iff $v=0$ that is iff $b=0$. Thus $h^{-1}(b,b)<1$.

(d) $\Rightarrow$ (c). Since $h$ is positive definite, $h-b\otimes
b$ is positive definite wherever $b=0$. Assume $b\ne 0$ and let
$v=h^{-1}(\cdot, b)$. The vector $v$ is $h$-orthogonal (and
$(h-b\otimes b)$-orthogonal) to the linear subspace $\textrm{ker} b$
over which $h-b\otimes b$ is positive definite as there it coincides
with $h$. Since every vector can be decomposed into two
$h$-orthogonal components, one proportional to $v$ and the other
belonging to $\textrm{ker} b$, the metric $h-b\otimes b$ is positive
definite if the norm of $v$ with respect to it is positive, and this
is indeed the case because
$h(v,v)-b(v)^2=h^{-1}(b,b)[1-h^{-1}(b,b)]>0$.

(d) $\Rightarrow$ (e). Let $v=h^{-1}(\cdot, b)$ and use the
Cauchy-Schwarz inequality $-b(y)\le \vert b(y)\vert= \vert
h(v,y)\vert\le \sqrt{h(v,v)}\sqrt{h(y,y)}<\sqrt{h(y,y)}$.

(e) $\Rightarrow$ (d). If $b=0$ the thesis is obvious. Let $b\ne 0$
and define $y=-h^{-1}(\cdot, b)$, then $F(s,y)>0$ reads
$\sqrt{h^{-1}(b,b)}-h^{-1}(b,b)>0$ which implies $h^{-1}(b,b)<1$.

(e) $\Rightarrow$ (f). This is a result which can be found in
\cite[Chap. 11]{bao00}.

(f) $\Rightarrow$ (e). If $b=0$ the thesis is obvious, thus let
$b\ne 0$. A calculation proves that (see \cite[Eq. 11.1.4]{bao00})
for $y\ne 0$,
\[
(\frac{1}{2} \p_{y^i}\p_{y^j}F^2)
w^iw^j=(1+\frac{b(y)}{\sqrt{h(y,y)}})
h(w,w)-\frac{b(y)}{\sqrt{h(y,y)}} \frac{b(w)^2}{h^{-1}(b,b)},
\]
choosing any $w\ne 0$ such that $b(w)=0$ we obtain that
$1+\frac{b(y)}{\sqrt{h(y,y)}}>0$ and since this equation must hold
for every $y$, the choice $y=-h^{-1}(\cdot,b)$ gives the thesis.

The fact that $a(s,s)\ge 0$ follows from (e) and Eqs.
(\ref{hst})-(\ref{gst}).  $\square$
\end{proof}

\begin{remark}
The strongly convex Randers metrics are exactly those that appear in
Zermelo's navigation problem. This problem asks to minimize the
travel time between two points of a Riemannian manifold under a mild
current described by a vector field \cite{bao04}. The spacetime
approach has several advantages and has been explored in
\cite{gibbons09,caponio10}.
\end{remark}

We add now the condition of distinction on the spacetime $(M,g)$. As
proved by Clarke and Joshi \cite[Prop. 3.1]{clarke88}) the presence
of a complete timelike Killing field $k$ on $(M,\hat{g})$,
$\hat{g}=g/\chi$, implies that $(M,\hat{g})$ is reflective (the same
holds for $(M,g)$ because they are in the same class of conformally
related spacetimes). Together with the distinction property this
fact implies causal continuity  which implies stable causality and
hence the existence of a smooth time function with timelike gradient
for $(M,g)$. One can therefore suspect that the function $t: M\to
\mathbb{R}$ appearing in Eq. (\ref{jup}) could be chosen to be a
time function. This has indeed been proved in \cite{javaloyes08} and
we offer here a simplified proof.

\begin{theorem}
In Theorem \ref{ggg} the spacetime $(M,g)$ is  distinguishing
(equivalently, stably causal) if and only if the function $t:M\to
\mathbb{R}$ which projects on the first factor in the splitting can
be chosen to be a smooth time function  with timelike gradient.
\end{theorem}

\begin{proof}

Assume distinction and hence stable causality. Since stable
causality implies the  existence of a smooth time function with
timelike gradient \cite{bernal04}, one constant slice of this
function provides a (edgeless) spacelike hypersurface $\Sigma$ on
$M$. Since $\Sigma$ is edgeless $\pi(\Sigma)$ is an open subset of
$S$ and since $\Sigma$ is a locally Lipschitz graph with respect to
the fibration (as $k$ is timelike and $\Sigma$ is achronal),
$\pi(\Sigma)$ is closed. Thus $\pi(\Sigma)$ is a connected component
of $S$ and since $M$ is connected and $M\sim \mathbb{R}\times S$,
$S$ is connected, thus $\pi(\Sigma)=S$. Given the smooth spacelike
hypersurface the expression $x=\phi_t(\pi^{-1}(s)\cap\Sigma)$
establishes the diffeomorphism $(t,s) \to x$ of the splitting.
Further, each slice $t=const.$ is of the form $\phi_t(\Sigma)$ thus
it is spacelike and hence by Prop. \ref{ouv}, $t$ has timelike
gradient. Since $g(k,\nabla t)=\dd t[k]=1>0$, $\nabla t$ is past
directed and $t$ is a time function because if $v$ is a future
directed causal vector $\dd t[v]=g(v,\nabla t)>0$. The converse is
trivial because the presence of a time function implies stable
causality.  $\square$
\end{proof}

\subsection{The anisotropy functional and its motivation}
\label{cinq2}

In this section we assume that the spacetime $(M,g)$ of metric
(\ref{jup}) is distinguishing and hence causally continuous. Over
the space $S$ we can either minimize the global anisotropy
functional or take an approach which makes one further step towards
locality as suggested in the synchronization over a graph. There we
have used the Fermat's map over a ``neighborhood'' $G$ of the
diagonal of $S\times S$. For this reason the synchronization
depended only on the propagation of the signal between neighboring
points, the idea being that in a graph the propagation of the signal
between far apart points can be regarded as a succession of
propagation steps between closer points. In the present case given
the Riemannian manifold $S$ we have a product topology on $S$ and in
fact a Riemannian distance which induces a {\em uniformity}
\cite{willard70} namely a filter of neighborhoods of the diagonal
$\Delta \subset S\times S$. Since we can take arbitrary small
neighborhoods of $\Delta$ it is expected that the whole formalism
will pass to the (co)tangent fiber bundle.

Let $s,s'\in S$ be close points on $S$ and let us write $s'= \exp_s
v$ with $v\in TS_s$. The spacetime $(M,g)$ is causally continuous
and thus strongly causal \cite{hawking73,minguzzi19} namely it
admits arbitrarily small causally convex neighborhoods. Thus if
$x\in M$ there is an arbitrarily small neighborhood $V\ni x$ which
is causally convex and in particular $J\cap V\times V=J_V$,
where $J_V$ is the causal relation on $(V,g\vert_V)$. In
particular $V$ can be taken inside a globally hyperbolic
neighborhood (every point in any spacetime admits a globally
hyperbolic neighborhood \cite{minguzzi19}) and thus $V$ can be
chosen globally hyperbolic as well (because a causally convex subset
of a globally hyperbolic spacetime is globally hyperbolic). In
particular any two causally related points of $V$ is connected by a
causal geodesic.

Let us denote, using the splitting, $x=(t,s)$. We can use the light
cone induced by the metric (\ref{jup}) to determine that, up to
first order on $\vert v\vert=\sqrt{h(v,v)}$, the time taken by the
signal to move from $s$ to $s'$ (i.e. from $x$ to $x'=(s',t')$)
is\footnote{Recall that under distinction the splitting can be
chosen so that $(S,F)$ is a Finsler manifold.  In short we have
given a spacetime argument which, with few modifications, proves
that short Finsler geodesics are minimizing \cite[Sect.
6.3]{bao00}.}
\begin{equation} \label{fvt}
a(s,s')=t'-t=F(s,v)+o(\vert v\vert)=\sqrt{h(v,v)}+b(v) +o(\vert
v\vert).
\end{equation}
From this equation we deduce for close $s,s'$, with $s'=\exp_s v$,
\begin{align}
\delta_a(s,s')&=2b(v)+o(\vert v\vert),\\
d_r(s,s')&=\sqrt{h(v,v)}+o(\vert v\vert),
\end{align}
where, as in Theorem \ref{czo}, $\delta_a(s,s')=a(s,s')-a(s',s)$ is the antisymmetric part of the Fermat's map of the synchronization $t$ of the splitting, so that $b$ is the 1-form of that synchronization. The correspondence $[\delta_a]_{AS}\leftrightarrow 2b$ is then that
given by the cochain map $\nu$ of Eq. (\ref{gql}), which induces the isomorphism of cohomologies between
Alexander-Spanier's and De Rham's. It must
be noted that the cochain of the Fermat's map $a(s_1,s_2)$ has no
counterpart because the morphism works for the {\em smooth}
Alexander-Spanier's cohomology while Eq. (\ref{fvt}) proves that the
Fermat's map is not smooth on the diagonal due to the square root.
The same isomorphism sends $\p [\delta_a]_{AS}=[w]_{AS}(s_1,s_2,s_3)$
to $2 \dd b$ and indeed given an oriented surface $\Sigma\subset S$
with regular boundary $\p\Sigma$, $\int_\Sigma 2\dd
b=2\int_{\p\Sigma} b$ is the Sagnac effect in coordinate time.

Independently of the alternating cochain map $\nu$, which cannot produce a symmetric tensor, the optical metric is recovered from the leading near-diagonal quadratic term of the symmetric cochain $[d_r]_{AS}$ by polarization: setting $s_i(\epsilon)=\exp_{s}(\epsilon v_i)$, $i=2,3$,
\[
h(v_2,v_3) = \lim_{\epsilon\to 0}\frac{1}{2\epsilon^2}
[d_r(s,s_2(\epsilon))^2+d_r(s,s_3(\epsilon))^2-d_r(s_2(\epsilon),s_3(\epsilon))^2] .
\]
Similarly, the spacetime metric is recovered from local asymptotic data. Let $x_1=(t_1,s)$ and, for $i=2,3$, $x_i(\epsilon)=(t_1+\epsilon \tau_i, s_i(\epsilon))$, and let $X_i=\tau_i\p_t+v_i$ (horizontal lift of $v_i$ understood). Then
\begin{align*}
(g/\chi)(X_2,X_3) = \lim_{\epsilon\to 0}\frac{1}{\epsilon^2}\Big\{&
-[t_3-t_1-\delta_a(s,s_3)/2]\,[t_2-t_1-\delta_a(s,s_2)/2]\\
&+\frac{1}{2}[ d_r(s,s_3)^2+d_r(s,s_2)^2-d_r(s_2,s_3)^2]\Big\},
\end{align*}
where $t_i$ and $s_i$ stand for the coordinates of $x_i(\epsilon)$, since $t_i-t_1-\delta_a(s,s_i)/2=\epsilon[\tau_i-b(v_i)]+o(\epsilon)=\epsilon(\dd t-b)(X_i)+o(\epsilon)$. This is the polarized local content of the cochain (\ref{nqz}).
The correspondence $[\delta_a]_{AS}\leftrightarrow 2b$ already
suggests to take for the anisotropy functional
\begin{equation} \label{fuk}
A=\int_S h^{-1}(b,b) \dd \mu
\end{equation}
where $\mu$ is the volume form on $S$ according to the optical
metric $h$. That this should be the right functional that measures
the anisotropy is also suggested by the formula for the one-way
speed of light, which, according to the optical metric and using
coordinate time, reads (it can be easily obtained from Eq.
(\ref{jup}))
\begin{equation}
c=\frac{1}{1+b(e)}
\end{equation}
where $e$ is a normalized vector, $h(e,e)=1$, which gives the
direction in which we measure the one-way speed of light. If $b=0$
the one-way speed of light is isotropic.

In order to minimize the functional in practical situations one
would like to support the idea that a graph obtained from a
sufficiently fine triangulation of the Riemannian manifold would
lead to an anisotropy function which is approximately the same as
(\ref{fuk}). In this way one could apply the synchronization method
developed for graphs. The next remark shows that this is indeed the
case.

\begin{remark}
Compatibly with our interpretation of the manifold case in terms of
uniformities regarded as generalized graphs, we are going to show
that formula (\ref{fuk}) can  be obtained as follows
\begin{equation} \label{pnt}
\lim_{\epsilon \to 0}\frac{1}{4\epsilon^2}\int_{S\times S} [\delta(s_1,s_2)]^2
\dd m_{\epsilon}(s_1,s_2)=\alpha \int_S h^{-1}(b,b)\, \dd \mu
\end{equation}
where $\alpha$ is a proportionality constant,  $m_\epsilon$ is a
measure on $S\times S$ which for $\epsilon\to 0$ becomes almost
entirely supported on smaller and smaller neighborhoods of $\Delta
\subset S\times S$. Moreover, we shall take a measure such that
$m_{\epsilon}(\Delta)=0$ in analogy with the graph case. Other
conditions that $m$ satisfies are symmetry and
\begin{equation} \label{kke}
\int_{S\times A} \dd m_\epsilon=\mu(A)+o(\epsilon).
\end{equation}

Let $n$ be the dimension of $S$ and $V(S^{n-1})=2
\pi^{n/2}/\Gamma(n/2)$ the (n-1)-volume of a Euclidean sphere of
unit radius, let
\[
\dd m_\epsilon(s_1,s_2)=\frac{1}{V(S^{n-1}) n!\epsilon^{n+1}}\,
d_r(s_1,s_2)\, e^{- d_r(s_1,s_2)/\epsilon} \dd \mu(s_1)\dd \mu(s_2)
.
\]
For $s_1$ fixed and $\epsilon$ small most mass stays close to $s_1$
where the geometry is close to Euclidean. This fact helps to prove
that in the limit $\epsilon\to 0$ the  equality (\ref{kke}) holds.
It suffices to recall that in an Euclidean space of dimension $n$,
with $r^2={\bf x}\cdot {\bf x}$,
\begin{align*}
\int_0^{+\infty} r e^{- r/\epsilon} \dd
x^n=V(S^{n-1})n!\epsilon^{n+1} .
\end{align*}
In order to prove Eq. (\ref{pnt}) we need the equation, where ${\bf
b}$ is a vector and $b^2= {\bf b}\cdot {\bf b}=h^{-1}(b,b)$
\begin{align*}
\int_0^{+\infty} (2{\bf b}\cdot {\bf x}/\epsilon )^2 r
e^{-r/\epsilon} \dd x^n=& 4b^2 V(S^{n-2}) \int_0^{+\infty}  r^{n+2}
e^{- r/\epsilon} \dd
r\int_0^{\pi} \cos^2 \theta \sin^{n-2} \theta \dd \theta\\
&=4b^2 V(S^{n-2}) (n+2)!\epsilon^{n+1}\int_0^{\pi} \cos^2 \theta
\sin^{n-2} \theta \dd \theta.
\end{align*}
Thus for the interesting physical case $n=3$ we get
\[
\alpha=\frac{ V(S^{1}) 5!}{V(S^{2}) 3!}\int_0^{\pi} \cos^2 \theta
\sin \theta \dd \theta=20/3 .
\]
\end{remark}

\subsection{Minimization of the anisotropy functional: Hodge-de Rham
theorem} \label{cinq3}

The volume element for the Riemannian manifold $(S,h)$ is
\[
\dd \mu =\sqrt {\det h} \,d x^1\wedge d x^2 \wedge \cdots \wedge d
x^n=\varepsilon_{i_1 i_2 \cdots i_n}  d x^{i_1} \otimes d x^{i_2}
\otimes \cdots \otimes d x^{i_n}
\] where $\{x^i\}$ are coordinates on a chart of $S$ and
$\varepsilon$ is the Levi-Civita tensor. The Hodge star operator
$*:\Omega^k S\to \Omega^{n-k} S$ sends $k$-forms to $(n-k)$-forms
and is defined by
\[
(*\eta)_{i_1\cdots i_{n-k}}=\frac{1}{k!}\eta^{j_1 \cdots j_k}
\varepsilon_{j_1\cdots j_k i_1\cdots i_{n-k}}
\]
where indices are raised with $h^{-1}$. The Hodge dual satisfies
$**\eta=(-1)^{n(n-k)} \eta$. The codifferential $\delta: \Omega^k
S\to \Omega^{k-1} S$ is defined by\footnote{The factor $1/(n-k)!$
appears because we use the (physicist) convention on the wedge
product $\alpha \wedge \beta =\alpha \otimes \beta-\beta\otimes
\alpha$ and consequent formula for the exterior derivative in tensor
form $(d \eta)_{i_0\cdots i_{k}}= (k+1)! \, \p_{[i_0}
\eta_{i_1\cdots i_{k}]}=\p_{i_0} \eta_{i_1\cdots i_{k}}-\p_{i_1}
\eta_{i_0\cdots i_{k}}+\cdots$.}
\[
\delta= \frac{1}{(n-k)!}(-1)^{nk+n+1}
*d*=\frac{1}{(n-k)!}(-1)^k *^{-1} d *
\]
where $d$ is the exterior derivative. Note that its action over a
function gives zero. Up to a sign the codifferential is a divergence
on the first index, indeed if $\eta \in \Omega^k S$ it holds
\[
(\delta \eta)^{s_1 \cdots s_k}=-\eta^{s_0 s_1 \cdots
s_k}_{\qquad\quad  ;s_0}.
\]
The Laplacian $\Delta: \Omega^k  S\to \Omega^k$ is defined by
\[
\Delta \omega=-(\delta d+d \delta) \omega.
\]
A form $\omega$ such that $\Delta \omega=0$ is called {\em
harmonic}, a form which belongs to the image of $d$ is called {\em
exact} and a form which belongs to the image of $\delta$ is called
{\em co-exact}. A form which belongs to the kernel of $d$ is {\em
closed} and one which belongs to the kernel of $\delta$ is {\em
co-closed}.
Let us introduce the
scalar product on $\Omega^k S$
\[
(\alpha,\beta)_k=\int_S \alpha_{i_1 \cdots i_k} \beta^{i_1\cdots
i_k} \dd \mu=\frac{1}{(n-k)!}\int_S \alpha \wedge *\beta .
\]
If $\beta \in \Omega^{k} S$ and $\eta\in \Omega^{k-1} S$ it is easy
to check that $d$ and $\delta$ are adjoints of each other $(d
\eta,\beta)_k=(\eta, \delta\beta)_{k-1}$. A form $\omega\in \Omega^k
S$ is harmonic if and only if it is closed and coclosed because
\[(-\Delta \omega,\omega)_k=(d\omega,d\omega)_{k+1}+(\delta \omega,\delta \omega)_{k-1} .\]

The Hodge-de Rham  decomposition theorem states that in a compact
Riemannian manifold $(S,h)$ every form $b\in \Omega^k S$ splits
uniquely as the sum of three orthogonal components, one exact, one
harmonic and one co-exact: $b=\dd \alpha+\omega+\delta j$; and that
each element of $H^{k}(S)$ has one and only one harmonic
representative. Note that the exact component cannot be coclosed
otherwise it would be both closed and coclosed and hence harmonic.
The same argument holds for the coexact component.

The anisotropy functional reads simply
\begin{equation} \label{pny}
A=(b,b)_1=(\dd \alpha,\dd \alpha)_1+(\omega,\omega)_1+(\delta
j,\delta j)_1.
\end{equation}
Perturbing $b$ with an exact form and using the duality between $d$
and $\delta$ we find that $A$ is stationary if and only if
\begin{equation} \label{pkr}
\delta b=0,
\end{equation}
which means that in the Hodge-de Rham decomposition the exact
component $d \alpha$ must be co-closed and hence that it must vanish
(because otherwise it would be closed and co-closed and hence
harmonic). Actually, Eq. (\ref{pny}) clarifies that this stationary
point is also a minimum because it cancels the only removable term
$(\dd \alpha,\dd \alpha)_1$. The elimination of the exact component
$\dd \alpha$ can be easily accomplished with the replacement $t \to
t-\alpha$, see Eqs. (\ref{ks1})-(\ref{ks2}), which fixes the
synchronization. Thus the Hodge-de Rham theorem tells us that a
solution to the synchronization problem of minimizing the anisotropy
functional exists and that with respect to the synchronized clocks
\[
b=\omega+\delta j
\]
where $\omega\in H^{1}(S)$ is harmonic and $\delta j$ is co-exact.
The equation (\ref{pkr}) that we have just solved can be rewritten
$b^{ i}_{\ ;i}=0$ thus it is a kind of Coulomb (Lorentz) gauge
fixing.

\begin{theorem} \label{n:harm}
Let $(S,h)$ be a compact Riemannian manifold. The 1-form $b\in T^*S$
appearing in Eq. (\ref{jup}) satisfies Eq. (\ref{pkr}) if and only
if the anisotropy functional (\ref{pny}) is minimal  on $b$ if and
only if the time $t$ appearing in the metric (\ref{jup}) satisfies,
apart from $\dd t(k)=1$, the wave equation
\[
\hat{\square} t=0,
\]
where $ \hat{\square} t=\hat{g}^{\mu \nu} t_{\hat{;}\mu \hat{;}\nu}
$ is the D'Alembertian operator with respect to the metric
$\hat{g}=g/ \chi=g/(-k\cdot k)$ so that $k$ is Killing and
normalized with respect to $\hat{g}$.
\end{theorem}

\begin{proof}
The first part of the theorem has been already proved. For the
second part, since $k$ is normalized, $\hat{g}(k,k)=-1$. In what
follows we just take $\chi=1$ so that we don't have to put an hat in
all the terms, it is understood, however, that the equations refer
to $\hat{g}$, not to $g$.

There is a general procedure that allows us to mirror tensor fields
on $S$ to tensor fields on $M$ \cite[Appendix]{geroch71}. In
particular the covariant space metric on $S$ is represented on $M$
by $h_{\mu \nu}=g_{\mu \nu}+k_\mu k_\nu$, the contravariant space
metric by $h^{\mu \nu}=g^{\mu \nu}+k^\mu k^\nu$, the Kronecker delta
by $h^{\mu}_{\nu}=\delta^{\mu}_{\nu}+k^{\mu}k_{\nu}$, the 1-form $b$
by $\dd t+k_\mu$ and the covariant derivative on $S$ is obtained by
taking the spacetime covariant derivative and by projecting all the
components of the so obtained tensor with $h^{\mu}_{\nu}$. Thus the
equation on $(S,h)$, $b_{i}^{\ ; i}=0$ becomes (here, with some
ambiguity, $;$ is used for both the covariant derivative on space
and on spacetime the difference is understood from the type of index
used, latin for the quotient space, greek for spacetime)
\begin{align*}
0&=(g^{\mu \nu}+k^\mu k^\nu)(t_{,\nu}+k_{\nu})_{;\mu}\\
&=\square
t+\frac{1}{2}(g^{\mu \nu}+k^\mu k^\nu)
(k_{\nu;\mu}+k_{\mu;\nu})+\nabla_k(\dd t[k])=\square t
\end{align*}
restoring the hats the thesis follows.  $\square$

\end{proof}

\begin{definition} \label{n:harmdef}
Let $(S,h)$ be compact and connected, and let $b=\dd\alpha+\omega+\delta j$ be the Hodge-de Rham decomposition of $b$. The synchronization $t'=t-\alpha$, for which $b'=b-\dd\alpha=\omega+\delta j$, is called the {\em harmonic synchronization}. It always exists, and it is the only synchronization, up to an additive constant, such that $\delta b'=0$. Equivalently, $\alpha$ is a solution of the Poisson equation $-\Delta\alpha=\delta b$, which is solvable because $\int_S\delta b\,\dd\mu=0$, and it is smooth if $b$ is smooth, by elliptic regularity. By Theorem \ref{n:harm} the harmonic synchronization is the minimizer of the anisotropy functional (\ref{pny}) and it is characterized by the wave equation $\hat{\square} t'=0$, which explains the name: $t'$ is a harmonic time coordinate of $(M,\hat g)$. It need not be a time function, see Subsection \ref{n:counter}. If $S$ is not compact, we say that a synchronization is harmonic if $\delta b'=0$. Such synchronizations still exist, since on a connected non-compact manifold the Poisson equation is solvable for every smooth right-hand side: for a real-analytic metric this is a theorem of Malgrange \cite[Theor.\ 8]{malgrange55}, and for a smooth metric it follows from \cite[Theor.\ 5]{malgrange55} and from Aronszajn's unique continuation theorem for second order elliptic equations \cite{aronszajn57}. These synchronizations, however, are no longer unique, as $\alpha$ is determined only up to a harmonic function, and a condition at infinity is needed to select one of them. In the circular spacetimes of Section \ref{n:kerr} the time $t$ is harmonic, and it minimizes the anisotropy functional with respect to the compactly supported resynchronizations, see Theorem \ref{n:kerrthm}(b).
\end{definition}

The question whether the harmonic synchronization is a time function whenever some synchronization is a time function, namely whenever the spacetime is distinguishing, is the question that had kept the 2011 manuscript of this work unpublished. It is answered in the negative in the next subsection, while the rest of this section develops an alternative synchronization which does respect causality. In what follows we write $\vert b\vert_h=\sqrt{h^{-1}(b,b)}$ for the anisotropy of the one-way speed of light, and we use the formula
\begin{equation} \label{n:oms}
\hat g^{-1}=-\p_{t'}\otimes \p_{t'}+h^{cd}(\p_c+b'_c\,\p_{t'})\otimes(\p_d+b'_d\,\p_{t'}) ,
\end{equation}
valid in the coordinates $(t',x^c)$ associated with any synchronization $t'$, whose 1-form is $b'$. In particular $\hat g^{-1}(\dd t',\dd t')=-1+\vert b'\vert_h^2$, which gives again the equivalence between (a) and (d) in Proposition \ref{ouv}.

\subsection{The harmonic synchronization can violate causality} \label{n:counter}

The harmonic synchronization minimizes the mean square anisotropy, which does not prevent the anisotropy from exceeding 1 somewhere.

\begin{example} \label{n:ex}
Let $S=T^2=\mathbb{R}^2/(2\pi \mathbb{Z})^2$ with coordinates $(x,y)$, optical metric $h=\dd x^2+\varphi(x)^2\dd y^2$, where $\varphi(x)=1+\frac{1}{2}\cos x$, and let $b=\frac{C}{\varphi(x)}\,\dd x$ with $C>0$. Since $\sqrt{\det h}=\varphi$ we have $b^{i}_{\ ;i}=\varphi^{-1}\p_x(\varphi\, C/\varphi)=0$, thus $\delta b=0$ and the harmonic synchronization is $t$ itself. Its anisotropy is $\vert b\vert_h=C/\varphi$, which exceeds 1 where $\varphi<C$. Let us choose $C\in(\frac{1}{2},\frac{\sqrt{3}}{2})$. By Proposition \ref{ouv}, $t$ is not a time function. Nevertheless the spacetime is globally hyperbolic. Indeed, let $\langle 1/\varphi\rangle=\frac{1}{2\pi}\int_0^{2\pi}\frac{\dd x}{\varphi}=\frac{2}{\sqrt{3}}$ and let $u(x)=C\int_0^x(\frac{1}{\varphi}-\langle\frac{1}{\varphi}\rangle)\,\dd x'$, which is periodic. The resynchronization $t'=t-u$ has 1-form $b'=b-\dd u=C\langle 1/\varphi\rangle\, \dd x$, whose norm is the constant $2C/\sqrt{3}<1$. Thus $t'$ is a temporal function, and by Theorem \ref{n:sync} below its slices are Cauchy hypersurfaces.
\end{example}

\subsection{Minimax synchronizations and the synchronization theorem} \label{n:minimax}

The previous example suggests to replace the $L^2$ criterion of the anisotropy functional with an $L^\infty$ criterion. For a synchronization $t'=t-\alpha$ with $\alpha$ Lipschitz, we define the {\em maximal anisotropy}
\[
A_\infty[\alpha]=\textrm{ess\,sup}_S\, \vert b-\dd \alpha\vert_h .
\]
By Proposition \ref{ouv} the synchronization has timelike gradient if and only if $\vert b-\dd\alpha\vert_h<1$, so the problem of the existence of a synchronization which respects causality is directly related to the minimization of $A_\infty$. The minimal value is determined by the Sagnac effect.

\begin{definition} \label{n:sagdef}
Let $\gamma$ be a closed piecewise $C^1$ curve on $S$, and let $T^\pm(\gamma)=l(\gamma)\pm\oint_\gamma b$ be the times taken by light to go around $\gamma$ in the two directions, for instance in an optical fiber laid along $\gamma$, see Eq.\ (\ref{hst}). Their difference $2\oint_\gamma b$ is the Sagnac effect of $\gamma$. The {\em Sagnac ratio} of the frame is
\[
\beta=\sup_\gamma \frac{T^+(\gamma)-T^-(\gamma)}{T^+(\gamma)+T^-(\gamma)}=\sup_\gamma \frac{\oint_\gamma b}{l(\gamma)} \in [0,+\infty],
\]
where the supremum is over the non-constant closed curves.
\end{definition}

The Sagnac ratio does not depend on the synchronization and it can be measured with a single clock, compare with the network analog of Subsection \ref{n:graphmm}. It is finite if $S$ is compact. It coincides with the weight of the fundamental cocycle introduced by Garfinkle and Harris \cite{garfinkle97,harris15}. Point (a) of the next theorem relates this invariant to the bounds $\vert b-\dd\alpha\vert_h\le w<1$ in a given synchronization, which are used in curvature estimates \cite{garfinkle23}. The next result is proved in \cite[Sect.\ 12]{minguzzi12bolza}, where $(S,h,t)$ are denoted $(Q,a,\tau)$, while $b$ and $\beta$ have the same meaning and sign conventions, and where the stationary spacetime is regarded as the Maupertuis reduction of a spacetime admitting a parallel null vector, so that its causal properties can be studied by combining weak KAM theory \cite{fathi08} with causality theory.

\begin{theorem}[Synchronization theorem] \label{n:sync}
Let $(M,g)$ be as in Eq.\ (\ref{jup}) with $S$ compact, and let $\beta$ be its Sagnac ratio.
\begin{itemize}
\item[(a)] We have $\beta=\min_\alpha A_\infty[\alpha]$, where the minimum is over the Lipschitz functions, and $\beta=\inf_\alpha \max_S \vert b-\dd \alpha\vert_h$, where the infimum is over the smooth functions. The minimum is attained by the weak KAM solutions $u$ of the Lagrangian $L_b(v)=\frac{1}{2}h(v,v)+b(v)$, namely by the viscosity solutions of the {\em synchronization equation}
\begin{equation} \label{n:synceq}
\vert \dd u-b\vert_h=\beta .
\end{equation}
Such a function $u$ is Lipschitz and semi-concave, and the anisotropy $\vert b-\dd u\vert_h$ equals $\beta$ at every point of differentiability of $u$.
\item[(b)] The spacetime is totally vicious if $\beta>1$, chronological but not distinguishing if $\beta=1$, and globally hyperbolic if and only if it is distinguishing, if and only if $\beta<1$.
\item[(c)] Let $\beta<1$. Every synchronization $t'=t-\alpha$, with $\alpha$ Lipschitz, which minimizes the maximal anisotropy, $A_\infty[\alpha]=\beta$, is a time function whose level sets are compact Lipschitz Cauchy hypersurfaces; by (a) this applies in particular to $t_*=t-u$. More precisely, for every future directed causal curve from $x_0$ to $x_1$, whose projection on $S$ is $\gamma$,
\begin{equation} \label{n:quant}
t'(x_1)-t'(x_0)\ge (1-\beta)\, l(\gamma) .
\end{equation}
\item[(d)] At the points of differentiability of $u$, $\hat g(\nabla t_*,\nabla t_*)=-(1-\beta^2)$.
\item[(e)] If $\beta>0$, then in the stationary spacetime of metric $\hat{g}_\beta=-(\dd t-b)^2+\beta^2 h$, the critical member of the family $\hat g_\lambda=-(\dd t-b)^2+\lambda^2h$, the level sets of $t_*$ are horizons: they are achronal and through each of their points there passes a past inextendible lightlike geodesic contained in the level set. They are the projections of the horizons of a spacetime admitting a parallel null vector, of which $(M,\hat{g}_\beta)$ is the Maupertuis reduction.
\end{itemize}
\end{theorem}

\begin{proof}
For the convenience of the reader we prove the inequality $\beta\le A_\infty[\alpha]$ and points (c)-(d), which are elementary once $u$ is given; for the remaining statements we refer to \cite[Prop.\ 12.5, Cor.\ 12.6 and Theor.\ 12.7]{minguzzi12bolza}. If $\alpha$ is $C^1$ and $\gamma$ is a closed curve, then $\oint_\gamma b=\oint_\gamma(b-\dd\alpha)\le \max_S\vert b-\dd\alpha\vert_h\, l(\gamma)$, thus $\beta\le \max_S\vert b-\dd\alpha\vert_h$; for Lipschitz $\alpha$ the same conclusion follows by a standard regularization argument.

Let $\beta<1$ and let $\alpha$ be Lipschitz with $A_\infty[\alpha]=\beta$. Since $\vert \dd \alpha-b\vert_h\le \beta$ almost everywhere, for every piecewise $C^1$ curve $\gamma:[0,1]\to S$ we have
\begin{equation} \label{n:dom}
\alpha(\gamma(1))-\alpha(\gamma(0))\le \int_0^1 [\beta \vert\gamma'\vert_h+b(\gamma')]\, \dd s .
\end{equation}
For $\alpha\in C^1$ this follows by integrating the inequality $\dd \alpha(\gamma')\le b(\gamma')+\vert \dd \alpha-b\vert_h\vert \gamma'\vert_h$. For Lipschitz $\alpha$ the almost everywhere bound cannot be integrated directly, since a curve can lie in the set where $\alpha$ is not differentiable, but the inequality still holds: $\alpha-\frac{1}{2}\beta^2t$ is dominated by $L_b$ \cite[Prop.\ 12.5]{minguzzi12bolza}, which is equivalent to Eq.\ (\ref{n:dom}) \cite{fathi08}; alternatively, one can approximate $\alpha$ uniformly by smooth functions with almost the same bound. Let $\eta(s)=(t(s),\gamma(s))$ be a future directed causal curve. From Eq.\ (\ref{jup}) we have $(\dot t-b(\gamma'))^2\ge h(\gamma',\gamma')$ and $\dot t-b(\gamma')>0$, thus $\dot t\ge \vert \gamma'\vert_h+b(\gamma')$ (compare with Eq.\ (\ref{kun})). Integrating and subtracting Eq.\ (\ref{n:dom}) we obtain Eq.\ (\ref{n:quant}). If $\gamma$ is not constant the right-hand side of (\ref{n:quant}) is positive, while if $\gamma$ is constant then $t$ increases strictly along $\eta$. Thus $t'$ is a time function. Its level sets are graphs over $S$, hence compact. Let $\eta$ be a future inextendible causal curve. If $l(\gamma)=+\infty$ then $t'\to+\infty$ by Eq.\ (\ref{n:quant}). Otherwise $\gamma$ converges to a point of $S$, and since $\eta$ is inextendible $t$ cannot converge, so again $t'\to+\infty$. The past case is analogous, thus every inextendible causal curve intersects each level set, which is therefore a Cauchy hypersurface. Point (d) follows from Eq.\ (\ref{n:oms}) with $b'=b-\dd u$. $\square$
\end{proof}

By Theorem \ref{n:sync}(c) all the minimizers of $A_\infty$ respect causality when $\beta<1$, while the properties (d) and (e), namely the constancy of the anisotropy at the points of differentiability and, for $\beta>0$, the existence through every point of a level set of a past inextendible lightlike geodesic of the critical spacetime contained in it, are not shared by all of them. Since different selections will be compared in Subsection \ref{n:glob}, it is convenient to fix the terminology.

\begin{definition} \label{n:mmdef}
Let $S$ be compact. A synchronization $t'=t-\alpha$, with $\alpha$ Lipschitz, is a {\em minimax synchronization} if $A_\infty[\alpha]=\beta$, namely if it minimizes the maximal anisotropy. It is an {\em eikonal synchronization} if $\alpha$ is a viscosity solution of the synchronization equation (\ref{n:synceq}), namely a weak KAM solution of the Lagrangian $L_b$. It is an {\em $L^p$-limit synchronization} if $\alpha$ is a uniform limit of the normalized minimizers $\alpha_{p_j}$, $p_j\to\infty$, of the functionals $A_{p}$ of Subsection \ref{n:glob}.
\end{definition}

The name eikonal is motivated by two eikonal equations. First, by Theorem \ref{n:sync}(d), if $\beta<1$ the function $t'/\sqrt{1-\beta^2}$ solves, where $\alpha$ is differentiable, the timelike eikonal equation $\hat g(\nabla \hat t,\nabla\hat t)=-1$ of $\hat g$, which is also satisfied by the Lorentzian distance from a hypersurface. Second, for $\beta>0$ let $\hat g_\beta=-(\dd t-b)^2+\beta^2h$ be the critical metric of Theorem \ref{n:sync}(e). In the coframe $(\dd t-b,\dd x^i)$ it is block diagonal, thus writing $\dd t'=(\dd t-b)+(b-\dd\alpha)$ we get
\[
\hat g_\beta^{-1}(\dd t',\dd t')=-1+\beta^{-2}\vert b-\dd\alpha\vert_h^2 ,
\]
so that Eq.\ (\ref{n:synceq}) is the eikonal equation $\hat g_\beta^{-1}(\dd t',\dd t')=0$ of the critical spacetime. Thus the level sets of an eikonal synchronization are, where $\alpha$ is differentiable, null hypersurfaces of $(M,\hat g_\beta)$, namely wavefronts, and by Theorem \ref{n:sync}(e) they are horizons. Lipschitz functions which solve Eq.\ (\ref{n:synceq}) almost everywhere need not be viscosity solutions; the viscosity condition is the one under which Theorem \ref{n:sync} applies.

By Theorem \ref{n:sync}(a) the eikonal synchronizations are minimax synchronizations, and by Theorem \ref{n:lp} below the same holds for the $L^p$-limit synchronizations, but none of these classes need reduce to a single synchronization (up to constants), and in general an $L^p$-limit synchronization is not an eikonal synchronization, see Remark \ref{n:select}. Thus a stationary spacetime over a compact space which is distinguishing admits synchronizations respecting causality with the smallest maximal anisotropy compatible with the Sagnac effect, and the eikonal ones are obtained precisely as anticipated in the introduction: for $0<\beta<1$ one opens the light cones within the family $\hat g_\lambda$ up to the critical value $\lambda=\beta$, where the projected horizons appear, and uses them as simultaneity slices. (For $\beta>1$ the critical metric would be reached by narrowing the cones, at $\beta=1$ the given metric is already critical, $\hat g_\beta=\hat g_1=\hat g$, and for $\beta=0$ the form $b=\dd u$ is exact and $t-u$ is the static time of the ultrastatic metric $\hat g$.) With respect to an eikonal synchronization the anisotropy of the one-way speed of light is the same, $\beta$, at every point of differentiability. In Example \ref{n:ex} an eikonal synchronization is $t'=t-u$ and $\beta=2C/\sqrt{3}$, the circles $y=const.$ realizing the supremum in the definition of $\beta$.

\begin{remark}
The weak KAM solution $u$ need not be unique up to constants and it need not be differentiable. However, by Bernard's theorem \cite{bernard07} there are $C^{1,1}$ functions $\alpha$ such that $\vert b-\dd\alpha\vert_h\le \beta$ everywhere. Whenever $\beta<1$ the corresponding synchronizations are minimax synchronizations of class $C^{1,1}$ whose gradient is everywhere timelike; they are time functions, but not temporal functions in the sense of Subsection \ref{uno2}, which requires smoothness, unless they happen to be smooth. Smooth temporal functions with maximal anisotropy smaller than $\beta+\eta$ exist for every $\eta>0$ such that $\beta+\eta\le 1$, by the second statement of Theorem \ref{n:sync}(a) and Proposition \ref{ouv}.
\end{remark}

\subsection{Freely falling observers and synchronous reference frames} \label{n:ff}

In this subsection we work with the metric $\hat{g}$, in which the clocks of the stationary observers measure proper time; the statements apply to the physical metric on the open sets where the conformal factor $\chi$ is constant. For clocks confined to an equipotential surface, such as that of the rotating planet of Subsection \ref{djc}, they concern the metric induced on the worldtube of the surface. Physical free fall for a spatially varying $\chi$ is discussed at the end of this subsection. Let $t'$ be a synchronization with $\vert b'\vert_h<1$, and let $n=-\nabla t'/N$, $N=\sqrt{1-\vert b'\vert_h^2}$, be the future directed unit normal to its slices. The observers moving with velocity $n$ are called {\em Eulerian observers}. For them the slices are, to first order, Einstein simultaneity slices, so they measure an isotropic one-way speed of light.

\begin{proposition} \label{n:euler}
The anisotropy $\vert b'\vert_h$ at $s\in S$ is the speed of the Eulerian observers with respect to the stationary observer $s$. The Eulerian observers are freely falling if and only if the anisotropy $\vert b'\vert_h$ is constant over $S$.
\end{proposition}

\begin{proof}
We have $-\hat{g}(n,k)=\dd t'(k)/N=1/\sqrt{1-\vert b'\vert_h^2}$, which is the Lorentz factor of the relative velocity. For the second statement, let $v=\nabla t'$, so that $\hat g(v,v)=-N^2$ and $\nabla_v v=\frac{1}{2}\nabla \hat g(v,v)=-N\nabla N$, where we used the symmetry of $\nabla\nabla t'$. Thus $\nabla_n n=N^{-2}\nabla_v v+N^{-1}v(N^{-1})\, v$, whose component orthogonal to $n$ is $-N^{-1}$ times the projection of $\nabla N$ on the slice. Since $N$ is constant along the integral curves of $k$, it is constant on the slices if and only if it is constant on $S$. $\square$
\end{proof}

Let $\beta<1$, so that the normals to the slices are timelike. By Theorem \ref{n:sync}, for an eikonal synchronization the anisotropy is constant and equal to $\beta$ at the points of differentiability of $u$. Thus, on the open set where $u$ is smooth, the integral curves of the unit normal to the slices are geodesics of $\hat g$, the corresponding observers move with the same speed $\beta$ with respect to the stationary observers, and $\hat{t}=t_*/\sqrt{1-\beta^2}$ measures $\hat g$-proper time along them; in the weak sense explained below the slices are orthogonal to a congruence of observers freely falling in $(M,\hat g)$. In other words, an eikonal synchronization defines a {\em synchronous reference system} for $\hat g$ in the sense of Landau and Lifshitz \cite[Sect.\ 97]{landau75}, compatible with the stationarity of the spacetime. These observers are also freely falling for the physical metric $g$ on the open sets where the conformal factor $\chi$ is constant; in general no such conclusion follows from the statement for $\hat g$ alone, since timelike geodesics are not preserved by conformal rescalings, see the beginning of this subsection. Landau and Lifshitz observed that such systems can always be constructed locally, starting from an arbitrary spacelike hypersurface, but that the geodesics of the congruence generically intersect, giving rise to fictitious singularities. The synchronization theorem provides a global stationary synchronous system in a weak sense: the smooth geometric statements, such as Proposition \ref{n:euler}, apply where $u$ is regular. The projections of the freely falling observers are the integral curves of $(\dd u-b)^\sharp$, namely the calibrated curves of weak KAM theory. They can be continued indefinitely into the past without meeting each other, while in the future they can meet at the points where $u$ is not differentiable, which play the role of the caustics of the synchronous system: several backward calibrated curves may end at such a point, and their lifts describe free fall only up to it. Even so, every slice remains a Lipschitz Cauchy hypersurface.

\paragraph{Syntonization by redshift and physical synchronous frames} The previous statements refer to the normalized metric $\hat g$, and they transfer to the physical metric $g=\chi\hat g$ only where $\chi$ is constant. However, the conformal factor can be measured by the clocks themselves. Real clocks measure proper time, $\dd s=\sqrt{\chi}\,\dd t$ along the worldlines of the frame. If the clock $s_A$ sends two signals to the clock $s_B$, separated by the proper time $\Delta s_A$, and they are received separated by the proper time $\Delta s_B$, then, since by stationarity the two signals are separated by the same interval of $t$ at emission and at reception,
\[
\frac{\Delta s_B}{\Delta s_A}=\sqrt{\frac{\chi(s_B)}{\chi(s_A)}} ,
\]
which is the gravitational redshift. Thus the relative rates of the stationary clocks determine $\chi$ up to a constant factor, and this is precisely the information which is used to syntonize them, namely to establish property ${\bf z=0}$. After this rate correction the offsets of the clocks, namely the function $\alpha$ of a synchronization $t'=t-\alpha$, are still free, and using the redshift also for their determination the clocks can be syntonized and synchronized coherently. For a synchronization $t'=t-\alpha$ with $\vert b-\dd\alpha\vert_h<1$, the {\em physical lapse}, namely the $g$-proper time of its Eulerian observers per unit of $t'$, and its optimal upper bound are
\[
\Lambda_\chi(\alpha)=\frac{\sqrt{\chi}}{\sqrt{1-\vert b-\dd\alpha\vert_h^2}}\, ,\qquad \Lambda_*=\inf_\alpha\max_S\Lambda_\chi(\alpha),
\]
the infimum being over the smooth functions $\alpha$ with $\vert b-\dd\alpha\vert_h<1$. In \cite[Theor.\ 12.16]{minguzzi12bolza} it is shown, for $S$ compact and connected and $\beta<1$, by regarding the stationary spacetime as the Maupertuis reduction of a spacetime admitting a parallel null vector with a potential determined by $\chi$, that there are synchronizations $t-u_\chi$, with $u_\chi$ a weak KAM solution of the corresponding Lagrangian, whose physical lapse equals $\Lambda_*$ at every point of differentiability of $u_\chi$. Their anisotropy is $\sqrt{1-\chi/\Lambda_*^2}$, their slices are compact Lipschitz Cauchy hypersurfaces orthogonal, in the weak sense explained above, to observers freely falling in the physical metric, and $\Lambda_*(t-u_\chi)$ measures proper time along these geodesic normals. In other words, they define a stationary synchronous reference system for the physical metric, in which rate and simultaneity are fixed by the same freely falling clocks. For $\chi=1$ one recovers $\Lambda_*=(1-\beta^2)^{-1/2}$ and the eikonal synchronizations. The value of $\Lambda_*$ is determined by a Sagnac condition in which the length of each loop is weighted by the redshift, and these synchronizations are smooth under the corresponding Galilean convergence condition, a curvature condition of a new type which on a compact space characterizes the fields of inertial forces of rigid motions along Killing vector fields \cite[Theor.\ 12.17]{minguzzi12bolza}. For static spacetimes their slices are orthogonal to observers whose projections approach, in the infinite past, the set where $\chi$ is maximal; for the Schwarzschild exterior, with the normalization at infinity, the infalling branch gives the ingoing Painlev\'e-Gullstrand time, while the synchronizations considered in this work return the static time \cite[Example 12.19]{minguzzi12bolza}. The two choices are adapted, respectively, to freely falling observers and to observers at rest.

\subsection{A curvature condition} \label{n:curv}

Let $F=\dd b$, namely $F_{ab}=\p_a b_b-\p_b b_a$, and let us consider the condition
\begin{equation} \label{n:ncb}
\mathcal{R}_b(s,v):=Ric_h(v,v)-(D^cF_{ca})v^a+\frac{1}{4}F_{cd}F^{cd}\ge 0 \qquad \textrm{for every } (s,v)\in TS ,
\end{equation}
where $D$ is the Levi-Civita connection of $h$. It does not depend on the synchronization, since it involves only $\dd b$. Following \cite[Sect.\ 11]{minguzzi12bolza} we call it the {\em Galilean convergence condition}: $\mathcal{R}_b$ is the Ricci tensor of the Galilei (Newton-Cartan) connection of the classical spacetime $\mathbb{R}\times S$ associated with the Lagrangian $\frac{1}{2}h(v,v)+b(v)$, evaluated on the vectors $\p_t+v$. It is also the null convergence condition for the spacetime with a parallel null vector of which $(M,\hat g)$ is the Maupertuis reduction, and it implies the null convergence condition for all the spacetimes $-(\dd t-b)^2+\lambda^2h$, $\lambda>0$ \cite[Sect.\ 12.4]{minguzzi12bolza}.

\begin{theorem} \label{n:curvthm}
Let $S$ be compact, let $h$ and $b$ be of class $C^r$, $r\ge 3$, assume (\ref{n:ncb}) and $\beta<1$. Then $u$ is of class $C^r$ and unique up to constants, $K=(\dd u-b)^\sharp$ is a Killing vector field of $(S,h)$ of constant length $\beta$, and the eikonal synchronization $t_*=t-u$, which is therefore unique up to constants, coincides with the harmonic one. Moreover, the slices of $t_*$ are totally geodesic, the unit normal $n$ is a parallel vector field, and $(M,\hat g)$ is isometric to the ultrastatic spacetime (namely a product $-\dd T^2+h'$ with $h'$ independent of $T$, a static spacetime with constant lapse) $(\mathbb{R}\times S, -\dd T^2+h-K^\flat\otimes K^\flat)$, where $T=t_*/\sqrt{1-\beta^2}$. In it, the stationary observers move with speed $\beta$ along the orbits of the Killing vector field $K$.
\end{theorem}

\begin{proof}
The first statements are proved in \cite[Theor.\ 12.9]{minguzzi12bolza} with the area theorem for horizons. In the coordinates $(t_*,x)$ the metric reads $\hat{g}=-(\dd t_*+K^\flat)^2+h$. The vector fields $\p_{t_*}$ and $K$ (the latter regarded as a vector field on $M$ tangent to the slices) are commuting Killing vector fields of $\hat g$, because $L_Kh=0$, $L_K K^\flat=0$ and $L_K \dd t_*=0$. Let $n=[(1-\beta^2)\p_{t_*}+K]/\sqrt{1-\beta^2}$. A direct computation using $K^\flat(K)=\beta^2$ gives $\hat g(n,\cdot)=-\dd t_*/\sqrt{1-\beta^2}=-\dd T$ and $\hat g(n,n)=-1$. Thus $n$ is a unit Killing vector field whose dual 1-form is closed, so that $\nabla n$ is both antisymmetric and symmetric, hence it vanishes. The vector field $n$ is complete, being a constant combination of complete commuting vector fields, and its flow maps the slice $T=0$, whose induced metric is $h-K^\flat\otimes K^\flat$, isometrically onto the other slices. This gives the isometric splitting. Finally $k=\p_{t_*}=[\sqrt{1-\beta^2}\,n-K]/(1-\beta^2)$, and since $\vert K\vert^2_{h-K^\flat\otimes K^\flat}=\beta^2(1-\beta^2)$, the speed of $k$ with respect to $n$ is $\beta$. $\square$
\end{proof}

The first statements of Theorem \ref{n:curvthm} also admit an independent elliptic proof, which uses neither the spacetime with a parallel null vector nor the area theorem \cite[Theor.\ 12.18]{minguzzi12bolza}: if $\mathcal{R}_b(s,K_0(s))\ge 0$ for every $s$, where $K_0=(\dd u_0-b)^\sharp$ is the anisotropy field of the harmonic synchronization $t-u_0$, then $K_0$ is a Killing vector field of constant length and $u_0$ is the unique weak KAM solution up to constants, so that the harmonic synchronization is eikonal.

Thus, under the curvature condition, the stationary observers are nothing but observers moving rigidly along the orbits of an isometry in an ultrastatic universe, and the synchronization recovers the time of the static frame, in agreement with Subsection \ref{vtx}. For instance, for the round 3-sphere and $b=\epsilon\, h(\xi,\cdot)$, where $\xi$ is a unit Killing vector field tangent to the fibers of the Hopf fibration and $0<\epsilon<1$, condition (\ref{n:ncb}) holds \cite[Example 12.10]{minguzzi12bolza}; the corresponding rotating universe is the ultrastatic universe over a squashed (Berger) 3-sphere, observed by observers moving along the Hopf fibers. In Example \ref{n:ex}, instead, $F=0$ and $\mathcal{R}_b(s,v)$ is the Gaussian curvature $\frac{1}{2}\cos x/\varphi$ times $\vert v\vert_h^2$, which changes sign, consistently with the fact that the two synchronizations differ.

\subsection{Mean curvature and maximal hypersurfaces} \label{n:max}

The results of this section admit an interpretation in terms of the mean curvature of the simultaneity slices. Let $t'$ be a synchronization with $\vert b'\vert_h<1$, and let $n$ and $N$ be as in Subsection \ref{n:ff}. The mean curvature $H=\textrm{div}\, n$ of the slices in $(M,\hat g)$ is obtained from $\hat\square t'=\textrm{div}(-Nn)=-NH-n(N)$ and from Eq.\ (\ref{n:oms}):
\begin{equation} \label{n:meanc}
H=-\frac{1}{N}\, b'^{i}_{\ ;i}+\frac{1}{N^2}\, h^{-1}(b',\dd N), \qquad N=\sqrt{1-\vert b'\vert_h^2} .
\end{equation}
For a synchronization with constant anisotropy, in particular for an eikonal synchronization where it is smooth, the last term vanishes, thus its slices are maximal if and only if it is harmonic. This is the case under the curvature condition (\ref{n:ncb}), where the slices are even totally geodesic. For $\beta>0$ there is a second, related interpretation (for $\beta=0$ the critical metric degenerates, while the mean curvature computation in $\hat g$ still applies). The generators of the horizons of Theorem \ref{n:sync}(e) project onto the integral curves of $K=(\dd u-b)^\sharp$, and since $\vert K\vert_h=\beta$ is constant, the expansion of the horizons, namely their null mean curvature, equals $\textrm{div}_h K=-b'^{i}_{\ ;i}$. So the wave equation for a smooth eikonal synchronization has a double meaning: it states the vanishing of the mean curvature of the simultaneity slices in the normalized spacetime $(M,\hat g)$, and the vanishing of the expansion of the corresponding horizons in the critical spacetime. Under the curvature condition the latter is obtained from the area theorem for horizons, which gives the regularity and the harmonicity of the eikonal synchronization \cite[Theor.\ 12.9]{minguzzi12bolza}, while the ultrastatic splitting of $\hat g$ is Theorem \ref{n:curvthm} \cite[Cor.\ 12.13]{minguzzi12bolza}, and the rigidity of the physical construction holds under its own curvature condition \cite[Theor.\ 12.17]{minguzzi12bolza}. The maximality of the slices transfers to the physical metric $g=\chi\hat g$ only under an additional condition, which is not provided by the area theorem. Indeed, for a hypersurface of dimension $m$, with future unit normals related by $n_g=\chi^{-1/2}n_{\hat g}$, the mean curvatures $H=\textrm{div}\, n$ are related by
\[
H_g=\chi^{-1/2}\big[H_{\hat g}+\tfrac{m}{2}\, n_{\hat g}(\log \chi)\big],
\]
so that $H_{\hat g}=0$ implies $H_g=0$ if and only if $n_{\hat g}(\chi)=0$ on the slice, as it happens, by symmetry, for the circular spacetimes of Section \ref{n:kerr}, see Theorem \ref{n:kerrthm}(e).

It is natural to compare these results with the theory of maximal and constant mean curvature (CMC) hypersurfaces, which plays a central role in general relativity. In a stationary spacetime with compact Cauchy hypersurfaces every compact CMC Cauchy hypersurface $\Sigma$ is maximal. Indeed, the flow of $k$ preserves the volume, so the first variation of the volume along the normal component $-\hat{g}(k,n)>0$ of $k$ vanishes, namely $H\int_\Sigma (-\hat{g}(k,n))\, \dd \mu_\Sigma=0$. The existence of compact maximal and CMC Cauchy hypersurfaces was studied by Gerhardt \cite{gerhardt83} and Bartnik \cite{bartnik84}, whose results were used by Galloway to prove splitting theorems for spatially closed spacetimes \cite{galloway84}, see \cite{galloway19} for a review. Under the timelike convergence condition two different compact maximal spacelike hypersurfaces are totally geodesic \cite[Theor.\ 7.3]{gerhardt83} and, subject to a compactness condition, the region between them is static: in the Gaussian normal coordinates of the first one the metric reads $-\dd T^2+\gamma$, where $T$ is the Lorentzian distance from it and $\gamma$ does not depend on $T$ \cite[Theor.\ 7.4]{gerhardt83}. The compactness condition is used to keep the level sets of the Lorentzian distance in a compact set, and Gerhardt ensures it by assuming the existence of a compact spacelike hypersurface of strictly positive, or strictly negative, mean curvature \cite[p.\ 548]{gerhardt83}. By the volume argument just given such a hypersurface never exists in a stationary spacetime with compact Cauchy hypersurfaces, but in this case the compactness is provided by the Killing flow, as we show in the proof of the next result.

\begin{proposition} \label{n:tcc}
Let $(M,g)$ be a globally hyperbolic spacetime admitting a complete timelike Killing vector field $k$ with compact space of orbits, satisfying the timelike convergence condition and admitting a compact CMC Cauchy hypersurface. Then $(M,g)$ is isometric to an ultrastatic spacetime $(\mathbb{R}\times \Sigma, -\dd T^2+\gamma)$ with $\Sigma$ compact, $k=c\,\p_T+Y$ where $c>0$ and $Y$ is a Killing vector field of $(\Sigma,\gamma)$ with $\vert Y\vert_\gamma<c$, and the time $T/c$ is both the harmonic synchronization and a minimax synchronization of the frame of $k$. In particular, the harmonic synchronization respects causality.
\end{proposition}

\begin{proof}
The CMC hypersurface $\Sigma$ is maximal, and so are its images $\Sigma_s=\phi_s(\Sigma)$ under the flow $\phi$ of $k$. Let $c_0>0$ be the minimum of $\sqrt{-g(k,k)}$, which exists because this function is invariant under the flow and hence defined on the compact space of orbits. Every orbit of $k$ is an inextendible timelike curve, thus it meets $\Sigma$ exactly once, and for $p\in\Sigma_s$, $s>0$, the orbit segment from $\phi_{-s}(p)\in\Sigma$ to $p$ is a timelike curve of length at least $c_0 s$. Thus $d(\Sigma,p)\ge c_0 s$, where $d$ is the Lorentzian distance; in particular $\Sigma_s\subset I^+(\Sigma)$. Let $q\in J^+(\Sigma)$ be such that $d(\Sigma,q)\le L$, and let $s>L/c_0$. If $q\in J^+(\Sigma_s)$ there would be $p\in\Sigma_s$ with $p\le q$, and hence $d(\Sigma,q)\ge d(\Sigma,p)\ge c_0 s>L$, a contradiction. Since $\Sigma_s$ is a Cauchy hypersurface, $q\in J^-(\Sigma_s)$, thus
\[
\{q\in J^+(\Sigma): d(\Sigma,q)\le L\}\subset J^+(\Sigma)\cap J^-(\Sigma_s),
\]
which is compact by global hyperbolicity. This is the compactness condition needed in the proof of \cite[Theor.\ 7.4]{gerhardt83}, which therefore applies to $\Sigma$ and $\Sigma_s$: $\Sigma$ is totally geodesic, the level sets of $T=d(\Sigma,\cdot)$ are totally geodesic up to $\Sigma_s$, which is one of them, and the region between $\Sigma$ and $\Sigma_s$ is isometric to a product $-\dd T^2+\gamma$, where $\gamma$ is the metric induced on $\Sigma$. Since $s$ is arbitrary, $J^+(\Sigma)$ is isometric to $([0,+\infty)\times\Sigma,-\dd T^2+\gamma)$, and the same argument applied to the time reversed spacetime shows that $M$ is isometric to $(\mathbb{R}\times\Sigma,-\dd T^2+\gamma)$. A Killing vector field of a product with compact $\Sigma$ has the form $c\,\p_T+Y$ with $c$ constant and $Y$ a Killing vector field of $(\Sigma,\gamma)$. Indeed, writing it as $f\p_T+Y_T$, the Killing equations give $\p_T f=0$, $\p_T Y_T=\nabla f$ and $L_{Y_T}\gamma=0$ for every $T$, thus $\nabla f$ is parallel, hence it vanishes at a maximum point of $f$ and everywhere. Since $k$ is timelike, $\vert Y\vert_\gamma<c$.

With $V^2=c^2-\vert Y\vert_\gamma^2$, the normalized metric is $\hat g=g/V^2$, and in the coordinates which are constant along the orbits of $k$ a short computation gives $b=Y^\flat/V^2$, $h=(\gamma+Y^\flat\otimes Y^\flat/V^2)/V^2$, where $Y^\flat=\gamma(Y,\cdot)$, and $h^{-1}(b,\cdot)=c^{-2}V^2\, Y$. The flow of $Y$ commutes with that of $k$ and preserves $g$, $T$ and the slices, thus it preserves $h$ and $b$, and since $Y(\vert Y\vert_\gamma^2)=0$ for a Killing vector field, $b^{i}_{\ ;i}=c^{-2}[Y(V^2)+V^2\,\textrm{div}_h Y]=0$. Thus $T/c$ is the harmonic synchronization. Its anisotropy is $\vert b\vert_h=\vert Y\vert_\gamma/c<1$, which is constant along the orbits of $Y$. The closure of the one-parameter group generated by $Y$ in the compact isometry group of $(\Sigma,\gamma)$ is a torus, so every orbit is recurrent and can be closed with an arbitrarily short arc. Along such closed curves the ratio $\oint b/l$ approaches $\vert b\vert_h$ at the starting point, thus $\beta=\max_S\vert b\vert_h$, which shows that $T/c$ is a minimax synchronization. $\square$
\end{proof}

In this sense the theory of maximal hypersurfaces anticipates the rigid case of our results: under the timelike convergence condition, the existence of a compact CMC Cauchy hypersurface in a stationary spacetime forces the frame to be a rigidly moving frame in an ultrastatic universe, in which case the static time is both the harmonic synchronization and a minimax synchronization (and also an eikonal synchronization if $\vert Y\vert_\gamma$ is constant). Let us stress, however, the differences. First, the synchronization theorem holds without any curvature condition and its slices always exist, while compact CMC Cauchy hypersurfaces may fail to exist. Second, the eikonal slices are characterized by their constant lapse, namely by the geodesic character of the normal observers, rather than by their mean curvature: in the language of numerical relativity, the synchronization theorem provides a stationary geodesic slicing rather than a maximal slicing. Third, the curvature condition (\ref{n:ncb}) is not the timelike convergence condition of the physical spacetime, but rather the null convergence condition of the spacetime with a parallel null vector of which $(M,\hat g)$ is the Maupertuis reduction, which implies the null convergence condition for every spacetime of the family obtained by opening and closing the light cones. Finally, in the asymptotically flat case maximal slicings of stationary spacetimes are known to exist \cite{bartnik90,chrusciel94}. For the circular spacetimes of Section \ref{n:kerr}, such as Kerr spacetime, we shall see that the maximal slices orthogonal to the zero angular momentum observers are the slices of the global synchronization, of a harmonic synchronization, and of the minimax synchronization selected by a regional minimality criterion.

\subsection{Global, local and minimax synchronizations} \label{n:glob}

The synchronization of Section \ref{mqz} is global: it minimizes the quadratic functional (\ref{wer}), to which every pair of clocks contributes, whatever their distance. The harmonic synchronization, instead, minimizes the local functional (\ref{pny}). In this subsection we show how the minimax synchronizations fit into this picture. We assume that $S$ is compact and connected and that $\beta<1$, and we consider the synchronizations $t-\alpha$ with $\alpha$ Lipschitz, whose 1-form is $b-\dd\alpha$ and whose Fermat's map is
\[
a_\alpha(s,s')=a(s,s')+\alpha(s)-\alpha(s') .
\]
The two-way time $r(s,s')=a(s,s')+a(s',s)$ does not depend on the synchronization and it is positive for $s\ne s'$. Indeed, by Theorem \ref{gdw}(1) and Eq.\ (\ref{n:quant}), $a_u(s,s')\ge (1-\beta)\, d_h(s,s')$, where $d_h$ is the distance of $(S,h)$ and $u$ is as in Theorem \ref{n:sync}. As in Subsection \ref{n:graphmm}, the anisotropy of a pair of clocks $s\ne s'$ is
\[
\varepsilon_\alpha(s,s')=\frac{\vert a_\alpha(s,s')-a_\alpha(s',s)\vert}{r(s,s')} .
\]

\begin{proposition} \label{n:pairprop}
Let $\alpha$ be Lipschitz. For every $\rho>0$ we have $\sup\{\varepsilon_\alpha(s,s'): 0<d_h(s,s')<\rho\}\ge A_\infty[\alpha]$. If $A_\infty[\alpha]<1$, then $\varepsilon_\alpha(s,s')\le A_\infty[\alpha]$ for every $s\ne s'$, thus
\[
\sup_{s\ne s'}\varepsilon_\alpha(s,s')=A_\infty[\alpha].
\]
As a consequence, the minimal value of $\sup_{s\ne s'}\varepsilon_\alpha$ is $\beta$, and its minimizers are those of $A_\infty$, namely the minimax synchronizations.
\end{proposition}

\begin{proof}
Let $k=A_\infty[\alpha]<1$, let $s\ne s'$ and $\eta>0$, and let $\gamma$ be a curve from $s'$ to $s$ such that $l(\gamma)+\int_\gamma (b-\dd\alpha)\le a_\alpha(s',s)+\eta$. Let $L=l(\gamma)$ and $B=\int_\gamma(b-\dd\alpha)$. We have $\vert B\vert\le kL$, as in the proof of Theorem \ref{n:sync}. The reversed curve gives $a_\alpha(s,s')\le L-B$, while $a_\alpha(s',s)\ge L+B-\eta$. Since $B\ge -kL$ we have $(1-k)(L-B)\le (1+k)(L+B)$, thus $(1-k)\,a_\alpha(s,s')\le (1+k)\,[a_\alpha(s',s)+\eta]$. Letting $\eta\to 0$ and exchanging $s$ and $s'$ we obtain $\vert a_\alpha(s,s')-a_\alpha(s',s)\vert\le k\,[a_\alpha(s,s')+a_\alpha(s',s)]$, namely $\varepsilon_\alpha\le k$.

For the first statement, let $s$ be a point of differentiability of $\alpha$ and let $v\in T_sS$ be a unit vector. By Eq.\ (\ref{fvt}), $r(s,\exp_s\tau v)=2\tau+o(\tau)$ and $a_\alpha(s,\exp_s\tau v)-a_\alpha(\exp_s\tau v,s)=2\tau\,(b-\dd\alpha)(v)+o(\tau)$, thus $\varepsilon_\alpha(s,\exp_s\tau v)\to \vert (b-\dd\alpha)(v)\vert$ for $\tau\to 0^+$. Taking the supremum over $v$, and over the points of differentiability, which form a set of full measure by Rademacher's theorem, we obtain the claim. The last statement follows from Theorem \ref{n:sync}(a): if $A_\infty[\alpha]\ge 1$ the supremum is not smaller than $1>\beta$, while otherwise it equals $A_\infty[\alpha]\ge \beta$. $\square$
\end{proof}

Thus, for the maximal anisotropy, the distinction between global and local synchronization disappears: the maximal anisotropy over all the pairs of clocks, whatever their distance, is attained infinitesimally, and the minimax synchronizations solve at the same time a global and a local problem. In particular, they do not depend on the choice of a measure on $S\times S$, provided its support contains a neighborhood of the diagonal, as is the case for the measure $\mu\otimes\mu$ of Section \ref{mqz} if $\mu$ has full support.

For the quadratic criterion the relation between global and local problems is subtler, and minimax synchronizations appear in a limit. For $1<p<\infty$ let us consider the local functional
\[
A_p[\alpha]=\int_S \vert b-\dd\alpha\vert_h^p\, \dd\mu_h ,
\]
where $\mu_h$ is the Riemannian measure of $(S,h)$, so that $A_2$ is the anisotropy functional (\ref{pny}). These functionals are the limits of global functionals for measures concentrated on the diagonal. Indeed, let $n=\dim S$, let $\kappa_{n,p}=\int_{\vert v\vert<1}\vert v^1\vert^p\,\dd^n v$, and let $m_\rho$ be the measure on $S\times S$ with density $(\kappa_{n,p}\,\rho^{n+p})^{-1}$ with respect to $\mu_h\otimes\mu_h$ on $\{d_h(s,s')<\rho\}$ and zero elsewhere. By Eq.\ (\ref{fvt}), for $\alpha\in C^1(S)$,
\[
\lim_{\rho\to 0}\int_{S\times S}\vert a_\alpha(s,s')-a_\alpha(s',s)\vert^p\,\dd m_\rho=2^p A_p[\alpha] .
\]
The normalization can be understood as follows. For $s'=\exp_s w$ with $\vert w\vert_h<\rho$, Eq.\ (\ref{fvt}) gives $a_\alpha(s,s')-a_\alpha(s',s)=2\,(b-\dd\alpha)(w)+o(\rho)$, so the integrand is of order $\rho^p$, while the set of the points $s'$ paired with a given $s$ has volume of order $\rho^n$. Moreover, by rotational invariance, $\int_{\vert w\vert_h<\rho}\vert (b-\dd\alpha)(w)\vert^p\,\dd^n w=\kappa_{n,p}\,\rho^{n+p}\,\vert b-\dd\alpha\vert_h^p$, so that the factor $(\kappa_{n,p}\rho^{n+p})^{-1}$ turns the integral over the pairs $(s,s')$ into the integral over $S$ of the local density $2^p\vert b-\dd\alpha\vert_h^p$. If the measure on $S\times S$ is instead built from $\mu\otimes\mu$, with $\mu=f\mu_h$, the same argument gives the weighted functional $2^p\int_S f^2\,\vert b-\dd\alpha\vert_h^p\,\dd\mu_h$.
For $p=2$ this is the precise sense in which the local problem of Section \ref{hod} is the limit of the global problem of Section \ref{mqz}.

\begin{theorem} \label{n:lp}
Let $S$ be compact and connected. For $1<p<\infty$ the functional $A_p$ has a minimizer $\alpha_p\in W^{1,p}(S)$, unique up to an additive constant, which is the weak solution of
\[
\textrm{div}_h\big(\vert b-\dd\alpha\vert_h^{p-2}\,(b-\dd\alpha)^\sharp\big)=0 .
\]
For $p=2$ it gives the harmonic synchronization. Let us normalize $\int_S\alpha_p\,\dd\mu_h=0$. Then $(A_p[\alpha_p]/\mu_h(S))^{1/p}\to\beta$ for $p\to\infty$, and every sequence $p_j\to\infty$ has a subsequence along which $\alpha_{p_j}$ converges uniformly to a Lipschitz function $\alpha_\infty$ such that $A_\infty[\alpha_\infty]=\beta$. If $\beta<1$, then $t-\alpha_\infty$ is a minimax synchronization (an $L^p$-limit synchronization in the sense of Definition \ref{n:mmdef}), and by Theorem \ref{n:sync}(c) it is a time function whose level sets are compact Cauchy hypersurfaces, which satisfies Eq.\ (\ref{n:quant}).
\end{theorem}

\begin{proof}
For $p>1$ the integrand is strictly convex in $\dd\alpha$, and by the Poincar\'e inequality $A_p$ is coercive on the functions with zero mean, thus the existence and the uniqueness of the minimizer, as well as the Euler-Lagrange equation, follow from the direct method of the calculus of variations. Let $V=\mu_h(S)$, let $\Vert\cdot\Vert_q$ denote the $L^q$ norm, and let $u$ be as in Theorem \ref{n:sync}(a), so that $\vert b-\dd u\vert_h\le\beta$ almost everywhere. By minimality, $\Vert b-\dd\alpha_p\Vert_p\le \Vert b-\dd u\Vert_p\le \beta V^{1/p}$, and by H\"older's inequality, for $n<q\le p$,
\[
\Vert b-\dd\alpha_p\Vert_q\le V^{1/q-1/p}\,\Vert b-\dd\alpha_p\Vert_p\le \beta V^{1/q} .
\]
Thus, for every $q>n$, the functions $\alpha_p$, $p\ge q$, are bounded in $W^{1,q}(S)$, hence, by Morrey's embedding, in $C^{0,1-n/q}(S)$. By a diagonal argument every sequence $p_j\to\infty$ has a subsequence along which $\alpha_{p_j}$ converges uniformly to a function $\alpha_\infty$, and weakly in $W^{1,q}(S)$ for every $q$. By the weak lower semicontinuity of the norm, $\Vert b-\dd\alpha_\infty\Vert_q\le \liminf_j \Vert b-\dd\alpha_{p_j}\Vert_q\le\beta V^{1/q}$, and letting $q\to\infty$ we get $\vert b-\dd\alpha_\infty\vert_h\le \beta$ almost everywhere, so that $\alpha_\infty$ is Lipschitz and, by Theorem \ref{n:sync}(a), $A_\infty[\alpha_\infty]=\beta$. The same inequalities give $\Vert b-\dd\alpha_\infty\Vert_q\le V^{1/q}\liminf_j \Vert b-\dd\alpha_{p_j}\Vert_{p_j}V^{-1/p_j}$, and letting $q\to\infty$, $\beta\le\liminf_j (A_{p_j}[\alpha_{p_j}]/V)^{1/p_j}$, while $(A_p[\alpha_p]/V)^{1/p}\le\beta$ for every $p$. Since the sequence is arbitrary, $(A_p[\alpha_p]/V)^{1/p}\to\beta$. The last statement follows from Theorem \ref{n:sync}(c). $\square$
\end{proof}

\begin{remark} \label{n:select}
The minimizers of $A_\infty$ form a convex set, which contains the weak KAM solutions of Theorem \ref{n:sync} and the $C^{1,1}$ functions provided by Bernard's theorem. By a theorem of Fathi, all of them are differentiable on the projected Aubry set of the Lagrangian $L_b$, which contains the closed curves realizing the supremum in Definition \ref{n:sagdef} when they exist, with the same differential there, so that their anisotropy equals $\beta$ on this set \cite{fathi08}. Outside the projected Aubry set they can be different. The weak KAM solutions are extremal, since their anisotropy equals $\beta$ at every point of differentiability, which makes their slices orthogonal to observers freely falling in $(M,\hat g)$, in the weak sense of Subsection \ref{n:ff}. The limits of the $L^p$ minimizers are expected, on the contrary, to lower the anisotropy wherever possible, in analogy with the absolutely minimizing Lipschitz extensions of Aronsson \cite{aronsson67} and with the absolute minimizers of $L^\infty$ functionals \cite{barron01}, which arise as limits of $L^p$ minimizers. We leave open the question whether the limit is unique and characterized by a local minimality property of this type. The same dichotomy appears for networks, where the shortest path solution of Theorem \ref{n:graphthm} plays the role of the weak KAM solutions and the strict minimax resetting is the limit of the $\ell^p$ minimizers (Remark \ref{n:strict}).

Under the curvature condition (\ref{n:ncb}) all these selections coincide: $t_*=t-u$ is harmonic and has constant anisotropy $\beta$ (Theorem \ref{n:curvthm}), thus $\textrm{div}_h(\vert b-\dd u\vert_h^{p-2}(b-\dd u)^\sharp)=\beta^{p-2}\,\textrm{div}_h (b-\dd u)^\sharp=0$, so that $\alpha_p=u$, up to a constant, for every $p$, and $\alpha_\infty=u$. In general, however, an $L^p$-limit synchronization is not an eikonal synchronization. On the flat torus $S=\mathbb{R}^2/(2\pi\mathbb{Z})^2$ let $h=\dd x^2+\dd y^2$ and $b=f(x)\,\dd y$ with $f(x)=(2+\cos x)/4$. For every $p>1$ we have $\textrm{div}_h(\vert b\vert_h^{p-2}b^\sharp)=\p_y[f(x)^{p-1}]=0$, so that, by strict convexity, $\alpha_p=0$ up to a constant, and $\alpha_\infty=0$. Since $\oint_\gamma b=\oint_\gamma f\,\dd y\le \frac{3}{4}\, l(\gamma)$ for every closed curve, with equality for the circle $x=0$, we have $\beta=\max f=3/4=A_\infty[0]$. Thus $t$ is an $L^p$-limit synchronization, and a minimax synchronization, whose anisotropy $f(x)$ is not constant, so it is not an eikonal synchronization. In the circular stationary axisymmetric spacetimes of Section \ref{n:kerr}, with respect to the time $t$ the 1-form $b$ is proportional to the differential of the angular coordinate of $S$, and its coefficient, as well as $h$, does not depend on this coordinate, thus $\textrm{div}_h(\vert b\vert_h^{p-2}b^\sharp)=0$, and by convexity, for every $p$ and for every compactly supported resynchronization, the functional $A_p$ over any region containing its support does not decrease. There, however, $S$ is not compact and, unless the frame is static ($b=0$, in which case $t$ is Einstein's synchronization), the anisotropy of $t$ is not constant, so that $t$ is not an eikonal synchronization, and the global minimization of the maximal anisotropy does not select $t$ uniquely; the selection principle is the regional one of Theorem \ref{n:kerrthm}(c).
\end{remark}

\section{Rigidly rotating observers in Kerr spacetime} \label{n:kerr}

In this section we study the synchronization of the clocks of a rigidly rotating frame in a circular stationary axisymmetric spacetime, the main example being the exterior of a Kerr black hole. Recall that in Boyer-Lindquist coordinates the Kerr metric reads \cite{misner73}
\begin{align*}
g=&-\Big(1-\frac{2mr}{\Sigma}\Big)\dd t^2-\frac{4mar\sin^2\theta}{\Sigma}\,\dd t\,\dd \phi+\frac{\Sigma}{\Delta}\,\dd r^2+\Sigma\, \dd \theta^2\\
&+\Big(r^2+a^2+\frac{2ma^2r\sin^2\theta}{\Sigma}\Big)\sin^2\theta\, \dd \phi^2 ,
\end{align*}
where $\Sigma=r^2+a^2\cos^2\theta$ and $\Delta=r^2-2mr+a^2$. It is of the form (\ref{n:papa}) with $x^A=(r,\theta)$, and it is invariant under $(t,\phi)\to(-t,-\phi)$. In the general case (\ref{n:papa}) we denote by
\[
\omega=-\frac{g_{t\phi}}{g_{\phi\phi}}, \qquad \alpha^2=\frac{g_{t\phi}^2-g_{tt}g_{\phi\phi}}{g_{\phi\phi}},
\]
the angular velocity and the lapse of the zero angular momentum observers (ZAMOs), whose worldlines are orthogonal to the slices $t=const.$ \cite{bardeen72}. For Kerr spacetime $g_{t\phi}^2-g_{tt}g_{\phi\phi}=\Delta\sin^2\theta$.

Let $\Omega\in\mathbb{R}$ and let $k_\Omega=\p_t+\Omega\p_\phi$ be the generator of the frame of the observers rotating rigidly with angular velocity $\Omega$. The region $M_\Omega$ where it is timelike is
\[
M_\Omega=\{V^2>0\},\qquad V^2=-g(k_\Omega,k_\Omega)=-(g_{tt}+2\Omega g_{t\phi}+\Omega^2 g_{\phi\phi}) .
\]
In Kerr spacetime, for $\Omega=0$ this is the region outside the ergosphere, while for $\Omega\ne 0$ it is a region bounded by two light surfaces, an inner one close to the horizon and an outer one which generalizes the light cylinder of a rotating frame in Minkowski spacetime. Let $\psi=\phi-\Omega t$ be the angle measured in the rotating frame, so that $k_\Omega=\p_t$ in the coordinates $(t,x^A,\psi)$. A direct computation shows that the normalized metric $\hat g=g/V^2$ takes the form (\ref{jup}), with $\chi=1$, time $t$ and
\begin{equation} \label{n:kerrbh}
b=\frac{g_{t\phi}+\Omega g_{\phi\phi}}{V^2}\, \dd \psi, \qquad h=\frac{g_{t\phi}^2-g_{tt}g_{\phi\phi}}{V^4}\, \dd\psi^2+\frac{g_{AB}\,\dd x^A\dd x^B}{V^2} .
\end{equation}
(For Kerr spacetime we verified this decomposition, and the identities below, with a computer algebra system.) The space $S$ is the region of the $(x^A,\psi)$ space where $V^2>0$, and it is not compact.

\begin{theorem} \label{n:kerrthm}
Let $(M,g)$ be circular stationary axisymmetric, let $\Omega\in\mathbb{R}$ and assume that $S$ is connected. The synchronization $t$ of the clocks of the rigidly rotating frame $k_\Omega$ has the following properties.
\begin{itemize}
\item[(a)] Its anisotropy is
\begin{equation} \label{n:kerreps}
\varepsilon=\vert b\vert_h=\frac{\vert g_{t\phi}+\Omega g_{\phi\phi}\vert}{\sqrt{g_{t\phi}^2-g_{tt}g_{\phi\phi}}}=\frac{\sqrt{g_{\phi\phi}}\,\vert\Omega-\omega\vert}{\alpha},
\end{equation}
which is the speed of the observers of the frame with respect to the ZAMOs. Moreover
$1-\varepsilon^2=g_{\phi\phi}V^2/(g_{t\phi}^2-g_{tt}g_{\phi\phi})$,
so $t$ is a temporal function on $M_\Omega$, $\varepsilon\to 1$ at the light surfaces, and $\varepsilon=0$ precisely on the axis of symmetry, where $g_{\phi\phi}=0$, and on the corotation surface $\omega=\Omega$.
\item[(b)] $\delta b=0$ and $\hat\square t=\square t=0$, namely $t$ is a harmonic synchronization (on the non-compact space $S$ the harmonic synchronizations are not unique without boundary or asymptotic conditions): for every compactly supported resynchronization $t-\alpha$, the anisotropy functional over any region containing the support of $\alpha$ does not decrease.
\item[(c)] For every axisymmetric open set $U\subset S$ and every Lipschitz function $\alpha$,
\[
\textrm{ess\,sup}_U\, \vert b-\dd\alpha\vert_h\ge \sup_U \varepsilon ,
\]
with equality for $\alpha=0$. Thus $t$ minimizes the maximal anisotropy on every axisymmetric region. Moreover, if $\vert b-\dd\alpha\vert_h\le \varepsilon$ almost everywhere on $S$, then $\alpha$ is constant.
\item[(d)] The global synchronization of Section \ref{mqz} returns $t$ (Theorem \ref{n:circ}).
\item[(e)] The slices $t=const.$ are maximal hypersurfaces of both $(M,g)$ and $(M,\hat g)$.
\end{itemize}
\end{theorem}

\begin{proof}
(a). By Eq.\ (\ref{n:kerrbh}), $\vert b\vert_h^2=b_\psi^2 h^{\psi\psi}=(g_{t\phi}+\Omega g_{\phi\phi})^2/(g_{t\phi}^2-g_{tt}g_{\phi\phi})$, and the second expression follows from $g_{t\phi}+\Omega g_{\phi\phi}=g_{\phi\phi}(\Omega-\omega)$ and from the definition of $\alpha$. The identity for $1-\varepsilon^2$ follows by expanding the squares. The slices $t=const.$ are orthogonal to the ZAMOs, which are therefore the Eulerian observers of the synchronization $t$, thus by Proposition \ref{n:euler} the anisotropy is the relative speed. Since $g_{\phi\phi}>0$ on the axisymmetric orbits, $\varepsilon<1$ on $M_\Omega$ and $t$ is a temporal function by Proposition \ref{ouv}.

(b). The 1-form $b$ has only the $\psi$-component, which does not depend on $\psi$, and $h$ does not mix $\dd\psi$ with the other differentials, thus
\[
b^{i}_{\ ;i}=(\det h)^{-1/2}\,\p_\psi\big[(\det h)^{1/2}h^{\psi\psi}b_\psi\big]=0 .
\]
 Since the anisotropy functional changes under $b\to b-\dd\alpha$ by $-2(b,\dd\alpha)_1+(\dd\alpha,\dd\alpha)_1=(\dd\alpha,\dd\alpha)_1\ge 0$, where we integrated by parts, $t$ is a minimizer among the compactly supported resynchronizations. The equation $\hat\square t=0$ is the theorem on the harmonic synchronization, while $\square t=0$ follows from $g^{tA}=0$ and from the independence of the metric coefficients of $t$ and $\phi$.

(c). The vector field $h^{-1}(b,\cdot)$ is tangent to the orbits of $\p_\psi$, which are circles of $h$-length $l$, and on each of them $\vert b\vert_h=\varepsilon$ is constant. Thus, with $e$ the unit tangent oriented by $b$, $\oint b(e)\,\dd l=\varepsilon\, l$. For almost every orbit $C\subset U$ the restriction of $\alpha$ to $C$ is Lipschitz and its derivative coincides almost everywhere with $\dd\alpha(e)$, thus
\[
\varepsilon\, l=\oint_C (b-\dd\alpha)(e)\, \dd l\le l\ \textrm{ess\,sup}_C \vert b-\dd\alpha\vert_h ,
\]
which gives the inequality. If $\vert b-\dd\alpha\vert_h\le\varepsilon$ almost everywhere, the same computation shows that $(b-\dd\alpha)(e)=\varepsilon$ almost everywhere on almost every orbit, thus $\dd\alpha(e)=0$ there, and since $b$ has no component orthogonal to $e$, $\vert b-\dd\alpha\vert_h^2=\varepsilon^2+\vert \dd\alpha\vert_h^2\le \varepsilon^2$, which implies $\dd\alpha=0$ almost everywhere.

(e). The isometry $P_{\psi_0}$ of the proof of Theorem \ref{n:circ} maps the slice $t=0$ into itself and reverses its future normal, so it changes the sign of the mean curvature: $H(P_{\psi_0}(x))=-H(x)$. Since $H$ is invariant under the rotations and $P_{\psi_0}$ preserves the orbits of $\p_\phi$, we get $H=0$. The normal to the slices is proportional to $\p_t+\omega\p_\phi$, which annihilates $V$, thus the conformal factor does not affect the mean curvature, see Eq.\ (\ref{n:meanc}). $\square$
\end{proof}

Thus in Kerr spacetime every rigidly rotating frame, whatever its angular velocity $\Omega$, recovers through an internal procedure based on the exchange of light signals the same simultaneity: that of the zero angular momentum observers, given by the Boyer-Lindquist time. More precisely, $t$ is returned by the global synchronization of Section \ref{mqz} with the integrals first taken along the orbits of $\p_\phi$ (Theorem \ref{n:circ}), it is a harmonic synchronization, and it minimizes the maximal anisotropy on every axisymmetric region, point (c). The latter regional criterion, rather than the minimization of the single global supremum, is the selection principle which singles out $t$, unless the frame is static. In the static case $b=0$, which for Kerr spacetime means $a=0$ and $\Omega=0$, namely the static observers of Schwarzschild spacetime, we have $\varepsilon=0$ and $\beta=0$, and $t$ is Einstein's synchronization, the unique minimax synchronization up to constants and an eikonal one. In the other cases $\varepsilon$ is continuous and not constant, since it vanishes on the axis but not everywhere, thus there is a non-empty open set $W\subset S$ such that $\sup_W\varepsilon<\sup_S\varepsilon$. If $\alpha$ is a smooth non-constant function with compact support $K\subset W$, then by rescaling $\alpha$ we can obtain $\vert b-\dd\alpha\vert_h<\sup_S\varepsilon$ on $K$, while outside $K$ the anisotropy is unchanged: $t-\alpha$ has the same global maximal anisotropy as $t$, and by (c) with $U=S$ it is also a global minimizer. For Kerr spacetime with $a\ne 0$ or $\Omega\ne 0$, the region $M_\Omega$ is not compact and its Sagnac ratio equals $1$, approached near the light surfaces, so that $\sup_S \varepsilon=1$. The synchronization theorem does not apply directly, and the anisotropy of $t$ is not constant, so $t$ is not an eikonal synchronization; nevertheless the time $t$ is a temporal function on $M_\Omega$. The anisotropy on the equatorial plane is shown in Figure \ref{n:fig}.

\begin{figure}[ht]
\begin{center}
 \includegraphics[width=10cm]{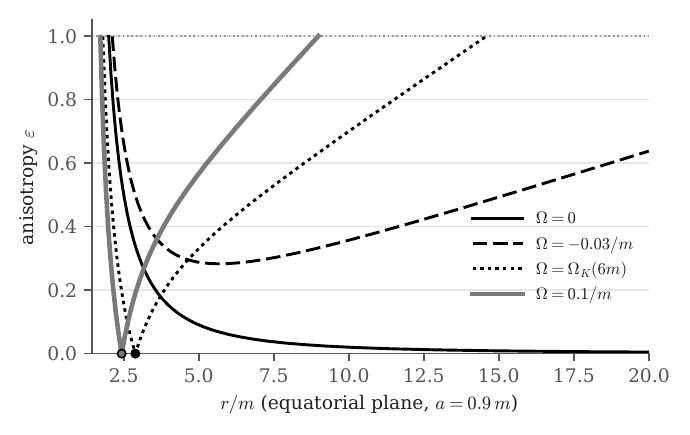}
\end{center}
\caption{Anisotropy $\varepsilon$ of the one-way speed of light, Eq.\ (\ref{n:kerreps}), for the Boyer-Lindquist synchronization of rigidly rotating frames on the equatorial plane of a Kerr black hole with $a=0.9\,m$, from the horizon $r_+$ outwards. It is also the speed of the observers of the frame relative to the ZAMOs. Each curve is drawn on the region $M_\Omega$, where it is bounded by $1$, a value approached at the light surfaces; the dots mark the corotation radius $\omega(r)=\Omega$ where the anisotropy vanishes. For $\Omega=0$ (static observers) the region $M_\Omega$ begins at the ergosphere $r=2m$. $\Omega_K(6m)$ denotes the angular velocity of the prograde circular geodesic of radius $6m$.} \label{n:fig}
\end{figure}

\begin{remark}
Along an orbit of $\p_\psi$ the Sagnac effect measured by the clocks of the frame, which measure the proper time $V\dd t$, is $2V\oint b=4\pi(g_{t\phi}+\Omega g_{\phi\phi})/V$, in agreement with \cite{tartaglia98}. It vanishes on the corotation surface, consistently with the fact that the observers which measure no Sagnac effect are precisely the ZAMOs \cite{costa21}. For the static observers ($\Omega=0$) far from the black hole, on the equatorial plane, Eq.\ (\ref{n:kerreps}) gives $\varepsilon\simeq 2ma/r^2=2J/r^2$, namely the frame dragging velocity $\omega r$, while for $\Omega\ne 0$ the anisotropy grows linearly with $r$ until the outer light surface. For the frame of the Earth the analogous gravitomagnetic contribution is of the order $10^{-15}$, while the contribution of the rotation, $\Omega r/c$, is of the order $10^{-6}$.
\end{remark}

Theorems \ref{n:circ} and \ref{n:kerrthm} apply without changes to the Kerr-Newman spacetime and, more generally, to the gravitational field of rigidly rotating stars, which is circular for perfect fluids in purely azimuthal motion \cite{kundt66}, as well as to the rotating planet of Subsection \ref{djc} once the gravitomagnetic field is included.

\section{Conclusions} \label{sei}

We have given a new synchronization method that generalizes Einstein's and which works in absence of round-trip conditions. It depends on a measure on space, thus it becomes natural in all those cases in which a privileged measure can be found. The measure is demanded to be finite, and integrability conditions are needed for the variational characterization, but in presence of symmetries these conditions can be strongly weakened. The method assumes that the clocks can be syntonized, or at least that they keep the same rate for the duration of the synchronization process (${\bf z=0}$ condition).

The method has been applied to three different models of the space of clocks: a network in which every clock communicates with the other clocks; a network in which every clock communicates with a set of privileged neighboring clocks; and a smooth Riemannian approximation of the latter scheme, which describes the clocks at rest in a stationary spacetime. In the first case we have proved that the time so constructed respects causality. In the second and third cases the quadratic minimization of the anisotropy leads, respectively, to the unique Nash equilibrium of a game played by the clocks and to the harmonic synchronization, which solves a wave equation.

The harmonic synchronization can violate causality. For this reason we have introduced the minimax synchronizations, which minimize the maximal anisotropy of the one-way speed of light. For networks the minimal value is the maximal Sagnac ratio of the cycles, and the optimal resetting of the clocks is computed by a shortest path algorithm. On a model of a global fiber network, the quadratic and the minimax resettings reduce the errors due to the Sagnac effect by a factor between five and seven with respect to Einstein's synchronization along a spanning tree.

For the clocks at rest in a stationary spacetime with compact space their existence is the content of the synchronization theorem, proved in the companion paper \cite{minguzzi12bolza}: if the spacetime is distinguishing, all the minimax synchronizations respect causality, and the eikonal ones have the same anisotropy at every point of differentiability and simultaneity slices orthogonal, in a weak sense which allows for caustics, to a congruence of observers freely falling for the metric $\hat g=g/(-k\cdot k)$ normalized by the Killing vector field $k$, so that they define a stationary synchronous reference system for $\hat g$, and for $g$ on the open sets where $k\cdot k$ is constant.

Since the conformal factor $-k\cdot k$ can be measured, up to a constant, by the clocks themselves through the gravitational redshift, the same information which serves to syntonize them can be used to synchronize them coherently: for a compact space of clocks and a distinguishing spacetime, the companion paper \cite[Theor.\ 12.16]{minguzzi12bolza} provides synchronizations whose slices are orthogonal, in the same weak sense, to observers freely falling in the physical metric $g$, the synchronized time being their proper time up to a constant factor, with a minimal maximal lapse determined by a redshift-weighted Sagnac condition.

If $\beta<1$ and a curvature condition holds, namely the Galilean convergence condition (\ref{n:ncb}), which is the null convergence condition of the spacetime with a parallel null vector of which the stationary spacetime is the Maupertuis reduction, the eikonal synchronization is unique up to constants and coincides with the harmonic synchronization, its slices are totally geodesic, and the stationary observers turn out to be rigidly moving observers in an ultrastatic universe, whose static time is recovered by the synchronization. We have also related these results to the theory of maximal hypersurfaces, and we have shown that the minimax synchronizations solve at the same time a global and a local problem, and that the limits of the $L^p$ minimizers of the anisotropy, the case $p=2$ being the harmonic synchronization, are minimax synchronizations, although in general not eikonal ones.

Furthermore, we have shown that, applied to the surface of the Earth (which is equipotential and hence respects the assumption on syntonization), the method leads to a time which coincides with the coordinate time of the non-rotating frame, and that every rigidly rotating frame in Kerr spacetime recovers the Boyer-Lindquist time, namely the simultaneity of the zero angular momentum observers, with the global method (integrating first along the orbits of the axial symmetry), as a harmonic synchronization, and as the synchronization which minimizes the maximal anisotropy on every axisymmetric region.

What is interesting is the fact that these results are obtained without using any information on the shape of the geoid or on the angular velocity of the frame, because the method we propose is internal: it makes no reference to specific exterior descriptions of the signal propagation. For this reason the new methods could find application for the precise synchronization of clocks over the Earth and, more broadly, of computational networks.

To this end, the next steps are the study of the sensitivity of the methods to measurement noise, to changing delays and to clock drift, and of the communication cost of estimating the required quantities, as well as the comparison with operational implementations of PTP and NTP. They might also have implications for foundational physics, where they could clarify the role of time splittings for canonical quantization schemes.

\appendix

\section{Data of the network example} \label{n:app}

\begin{sloppypar}
The nodes of the network of Subsection \ref{n:numex} are the following 28 cities, placed at their usual geographic coordinates, with the abbreviations used in Fig.~\ref{n:globe}: New York (NYC), Chicago (CHI), Seattle (SEA), Los Angeles (LAX), Miami (MIA), London (LON), Paris (PAR), Frankfurt (FRA), Stockholm (STO), Moscow (MOW), Marseille (MRS), Lisbon (LIS), Fortaleza (FOR), S\~ao Paulo (SAO), Lagos (LOS), Cape Town (CPT), Mombasa (MBA), Cairo (CAI), Dubai (DXB), Mumbai (BOM), Singapore (SIN), Hong Kong (HKG), Beijing (PEK), Tokyo (TYO), Perth (PER), Sydney (SYD), Auckland (AKL), Honolulu (HNL). The 43 links are: New York--Chicago; Chicago--Seattle; Seattle--Los Angeles; Los Angeles--Chicago; New York--Miami; Miami--Los Angeles; New York--London; New York--Paris; Miami--Fortaleza; Fortaleza--S\~ao Paulo; Fortaleza--Lisbon; Fortaleza--Lagos; London--Paris; Paris--Frankfurt; London--Frankfurt; Frankfurt--Stockholm; Stockholm--Moscow; Frankfurt--Moscow; Paris--Marseille; Paris--Lisbon; Marseille--Frankfurt; Lisbon--Lagos; Lagos--Cape Town; Cape Town--Mombasa; Mombasa--Dubai; Marseille--Cairo; Cairo--Dubai; S\~ao Paulo--Cape Town; Dubai--Mumbai; Mumbai--Singapore; Moscow--Beijing; Beijing--Hong Kong; Beijing--Tokyo; Hong Kong--Tokyo; Hong Kong--Singapore; Singapore--Perth; Perth--Sydney; Sydney--Auckland; Sydney--Honolulu; Auckland--Los Angeles; Honolulu--Los Angeles; Tokyo--Seattle; Tokyo--Los Angeles. The Earth is modeled as a sphere of radius $R=6371\,$km, the routes are the great circle arcs, whose length is multiplied by $1.3$, and the areas $A_z$ are computed along the great circle arcs. In scenario B the asymmetries $\eta$ are drawn from a normal distribution with standard deviation $10\,$ns, and their sample root mean square is $8.7\,$ns. The Sagnac ratio $\beta_G$ has been computed by bisection with the detection of negative cycles, the quadratic resetting by least squares, the minimax resetting with the Bellman-Ford algorithm, the strict minimax resetting with a sequence of linear programs, and the $\ell^p$ minimizers with Newton's method. The script is included in the source files of this work.
\end{sloppypar}

\begin{acknowledgements}
This work is an update of a research started in the years 2010-11 with arXiv:1009.3005. For the writing of this version we had the assistance of AI tools from Anthropic (Claude Opus~5.5) and OpenAI (GPT-6~Astra). This work was supported in part by a grant of access to OpenAI models through the ChatGPT for Academic Researchers program.
\end{acknowledgements}


\end{document}